\documentclass[11pt]{article}
\usepackage{setspace}
\usepackage{graphicx}
\usepackage{lscape}
\usepackage{pdflscape}
\usepackage{afterpage}
\usepackage{pdfpages}
\usepackage{float}
\usepackage{booktabs}
\usepackage[utf8]{inputenc}
\usepackage{indentfirst}
\usepackage{arydshln}
\usepackage{arydshln}
\usepackage{tikz}
\usepackage{lipsum}
\usepackage{pgffor}
\usepackage{dsfont}
\usepackage{amsfonts}
\usepackage{amsmath}
\allowdisplaybreaks
\usepackage{mathtools}

\usepackage{xfrac}
\usepackage{amssymb}
\usepackage{bigints}
\usepackage{graphicx}
\usepackage{url}				% Para escrever links
\usepackage{caption}
\usepackage{subcaption} % Required for creating figures with multiple parts (subfigures)
\usepackage{enumitem}
\usepackage{relsize,exscale}
\usepackage{multirow}
\usepackage{natbib}
\usepackage{bbm}
\setcitestyle{authoryear, open={(},close={)}}

\usepackage[flushleft]{threeparttable}

\usepackage{bibunits}

\usepackage{datetime}

\newdateformat{monthyeardate}{%
	\monthname[\THEMONTH] \THEYEAR}

\usepackage{titlesec}
\titleformat*{\section}{\large\bfseries}
\titleformat*{\subsection}{\large\bfseries}

\newcommand{\E}{\operatorname{\mathbb{E}}}
\newcommand{\Var}{\operatorname{Var}}

\newcommand{\dto}{\xrightarrow{d}}


\newcounter{parentnumber}
\usepackage{theorem}

\newtheorem{theorem}{Theorem}

\newtheorem{assumption}{Assumption}

\newtheorem{proposition}{Proposition}
\newtheorem{lemma}{Lemma}
\DeclareMathOperator*{\plim}{plim}
\newtheorem{remark}{Remark}

\newenvironment{proof}[1][Proof]{\noindent \textbf{#1.} }{\  \rule{0.5em}{0.5em}}
\providecommand{\U}[1]{\protect\rule{.1in}{.1in}}

\usepackage{hyperref} % More descriptive referencing
\hypersetup{
	colorlinks=true, breaklinks=true, bookmarks=true,bookmarksnumbered,
	urlcolor=blue, linkcolor=blue, citecolor=blue, % Link colors
	pdftitle={}, % PDF title
	pdfauthor={\textcopyright}, % PDF Author
	pdfsubject={}, % PDF Subject
	pdfkeywords={}, % PDF Keywords
	pdfcreator={pdfLaTeX}, % PDF Creator
	pdfproducer={LaTeX with hyperref and ClassicThesis} % PDF producer
}

\begin{document}
\begin{bibunit}[chicago]
	\setstretch{1}
	\title{
{\LARGE Partial Identification under Stratified Randomization}\thanks{We thank seminar participants at the São Paulo School of Economics - FGV and at the 47th Meeting of the Brazilian Econometric Society (SBE) for their valuable comments.}
}

	\author{
	  Bruno Ferman\thanks{São Paulo School of Economics - FGV. Email: \href{mailto:bruno.ferman@fgv.br}{bruno.ferman@fgv.br}}\and Davi Siqueira\thanks{São Paulo School of Economics - FGV. Email: \href{mailto:davi.siqueira@fgv.br}{davi.siqueira@fgv.br}}  \and Vitor Possebom\thanks{São Paulo School of Economics - FGV. Email: \href{mailto:vitor.possebom@fgv.br}{vitor.possebom@fgv.br}}
	}
	\date{}

	\maketitle

	\newsavebox{\tablebox} \newlength{\tableboxwidth}

	%%%%%%%%%%%%%%%%%%%%%%%%%%%%%%%%%%
	%%       Abstract                %
	%%%%%%%%%%%%%%%%%%%%%%%%%%%%%%%%%%

\begin{center}
	
% \footnotesize
First Draft: \monthyeardate\today; This Draft: \monthyeardate\today

\

		\large{\textbf{Abstract}}
	\end{center}

This paper develops a unified framework for partial identification and inference in stratified experiments with attrition, accommodating both equal and heterogeneous treatment shares across strata. For equal-share designs, we apply recent theory for finely stratified experiments to Lee bounds, yielding closed-form, design-consistent variance estimators and properly sized confidence intervals. Simulations show that the conventional formula can overstate uncertainty, while our approach delivers tighter intervals. When treatment shares differ across strata, we propose a new strategy, which combines inverse probability weighting and global trimming to construct valid bounds even when strata are small or unbalanced. We establish identification, introduce a moment estimator, and extend existing inference results to stratified designs with heterogeneous shares, covering a broad class of moment-based estimators which includes the one we formulate. We also generalize our results to designs in which strata are defined solely by observed labels.

	\


        \


	\newpage

	\doublespacing

% ----------------------------------------------------------------------
%  SECTION:  Introduction
% ----------------------------------------------------------------------
\section{Introduction}
\label{sec:intro}

Randomized experiments are central to causal inference, but their practical implementation often faces two complications: differential outcome attrition and stratified (blocked) randomization. Attrition can undermine the identification of average treatment effects (ATEs), particularly when missingness is related to treatment status. At the same time, various experiments now employ stratified or blocked assignment, randomizing units within strata to ensure balance and improve efficiency. The combination of these features creates unresolved challenges for valid inference, especially when treatment shares vary across blocks.\footnote{Examples of randomized experiments facing both stratified assignment and differential outcome attrition (or selection into observed outcomes) include the National Job Corps evaluation, where randomization is stratified by center with known, potentially varying assignment probabilities, and wages are observed only for the employed \citep{lee2009training}; the Colombian job-training program Jóvenes en Acción, where randomization is stratified by training site-by-class and follow-up differs by treatment status for men \citep{attanasio2011subsidizing}; the Colombian school voucher program (PACES), where lottery assignment occurs within city and cohort-specific lotteries, and ICFES exam-taking differs by voucher status \citep{angrist2006long}; and Project STAR, where randomization is stratified by school, and ACT/SAT outcomes are observed only for test takers, with class size affecting test taking \citep{krueger2001effect}. For further discussions of stratified randomization and its prevalence in field experiments, see \citet{duflo2007using} and \citet{bruhn2009pursuit}.} In many of these settings, researchers turn to Lee monotone-selection bounds \citep{lee2009training} to address differential attrition. However, the original Lee bounds framework was developed for settings with independent treatment assignment and does not account for stratified randomization, so additional theory is needed to adapt it to such designs.    

To address these issues, we develop a unified framework for partial identification and inference in stratified experiments with outcome attrition, accommodating both equal and heterogeneous treatment shares across blocks. In the equal-share case, we show that Lee’s procedure continues to identify the bounds, but that the conventional standard errors are conservative under this design. To conduct valid inference, we rely on the finely stratified experiment theory of \citet{bai2023efficiency} and apply it to the context of Lee bounds. This provides a closed-form, design-consistent variance estimator that accurately reflects the negative correlation among assignments induced by equal-share randomization, yielding tighter confidence intervals for Lee bounds. We also extend this approach to encompass designs where blocks are formed either on observed covariates or purely on discrete labels such as schools, clinics, or randomization waves, even when no additional covariate balancing occurs within blocks.

In the heterogeneous-shares case, the standard unconditional Lee bounds are invalid, and the conditional version, based on trimming within each stratum, is often unstable, biased, or undefined in finely stratified experiments with small or sparsely treated blocks. To address these issues, we introduce a novel bounding strategy that uses inverse probability weighting and a single global trimming share to construct bounds that remain valid and robust even with small or unbalanced blocks. We establish identification, propose a moment estimator, and develop an inference framework that generalizes to arbitrary, stratum-specific treatment shares.

While our equal-share results rely on a direct application of the framework in \citet{bai2023efficiency}, the heterogeneous-shares case requires a separate analysis. Thus, we build on their results to extend inference theory to stratified designs with heterogeneous treatment shares, encompassing a broad class of moment-based estimators that includes the one we introduce. Stratified assignment with varying treatment shares requires additional arguments and leads to a modified variance decomposition. The resulting inference procedure remains tractable even in complex experiments with many small or unbalanced blocks.

This study contributes to several strands of the economics literature, including partial identification with Lee bounds in the presence of attrition and design-consistent inference under stratified randomization. On the partial identification side, early work on nonparametric bounds for randomized experiments with missing data established core tools for analyzing treatment effects under attrition \citep{horowitz2000nonparametric,manski2003partial}. Recent surveys provide broad overviews of this literature and discuss its implications for empirical work \citep{tamer2010partial,kline2023recent}. Our contribution within this strand is to focus on a simple and widely used monotone-selection procedure and to show how it can be implemented in stratified randomized experiments with attrition in a way that respects the randomization and delivers design-consistent variance estimation.

Within this toolkit, \citet{lee2009training} remains especially influential, offering sharp bounds for the average treatment effect among always-observed units under random assignment and monotonicity. Subsequent work extends this approach to settings where selection patterns vary with covariates. \citet{semenova2025generalized} generalizes Lee’s monotonicity to a covariate-indexed version, and related work develops debiased machine-learning estimators for set- and intersection-bound problems \citep{semenova2023debiased,semenova2025aggint}. Despite these developments, Lee bounds remain widespread in applied work due to their transparency and practicality, which motivates our study of their implementation in stratified experiments.

In parallel, stratified randomization has become standard in field experiments, improving efficiency with fixed treatment shares within covariate-homogeneous or common-shock blocks \citep{bruhn2009pursuit, bai2023efficiency}. Nevertheless, most empirical applications with Lee bounds continue to estimate their standard errors as if treatment were assigned completely at random, ignoring the dependence introduced by stratification. As we show through theoretical analysis and simulation evidence, this practice can substantially overstate uncertainty in finely stratified experiments, leading to confidence intervals that are markedly more conservative than intended and highlighting the need for variance estimators that are consistent with the experimental design.

At the same time, the literature on inference for stratified experiments has primarily focused on point-identified average treatment effects \citep{bai2022inference}, although recent advances, such as the GMM-based framework of \citet{bai2023efficiency}, provide a sufficiently general theoretical foundation for working with the estimators considered in this paper. Existing approaches have yet to combine monotone-selection bounds with the dependence structure induced by stratification, or to explicitly address partial identification when treatment shares vary across blocks, a case in which standard Lee bounds are no longer valid and Lee bounds that include block fixed effects as covariates may be infeasible in finely stratified designs with small or sparsely treated blocks. These gaps motivate the contributions for inference under stratified assignment and attrition introduced in this paper.

The remainder of the paper is organized as follows. Section~\ref{sec:design_equal_share_sec} lays out our framework for stratified experimental designs. Section~\ref{sec:variance_equal_share} presents results for designs with homogeneous treatment shares across strata. Section~\ref{sec:lee_ipw_estimator} develops our approach for designs that allow heterogeneous treatment shares. Section~\ref{sec:lee_ipw_inference} derives inference for these general stratified designs. Section~\ref{sec:conclusion} concludes.

\section{General Setting under Heterogeneous-Shares Stratification}
\label{sec:design_equal_share_sec}

In this section, we consider a stratified design in which treatment shares may vary across blocks, or heterogeneous-shares stratification. The study population consists of $n$ units, indexed by $i = 1, \dots, n$, each associated with observed baseline covariates $X_i \in \mathbb{R}^{d_x}$. Each unit is assigned a binary treatment $D_i \in \{0,1\}$, and we define potential outcomes $(Y_i^*(1), Y_i^*(0))$ and potential selection indicators $(S_i(1), S_i(0)) \in \{0,1\}^2$. The observed outcome and selection indicator are then $S_i = D_i S_i(1) + (1-D_i) S_i(0)$ and $Y_i = S_i\left[D_i Y_i^*(1) + (1-D_i) Y_i^*(0)\right]$, so $Y_i$ is only observed if $S_i = 1$. Let $Q_n$ denote the joint distribution of $\big(Y^{*(n)}(1), Y^{*(n)}(0), S^{(n)}(1),\\ S^{(n)}(0), X^{(n)}\big)$, where the superscript $(n)$ indicates stacking over units $i = 1,\dots,n$. Throughout, we assume i.i.d.\ sampling from a superpopulation, that is, $Q_n = Q^n$, where $Q$ is the marginal distribution of $\big(Y_i^*(1), Y_i^*(0), S_i(1), S_i(0), X_i\big)$.

Heterogeneous-shares stratification refers to partitioning the sample into $G$ blocks, $\Lambda(X^{(n)}) \\= \{\lambda_g : g = 1, \dots, G\}$, with (possibly unequal) sizes $|\lambda_g| = N_g$ so that $\sum_{g=1}^G N_g = n$. Let $b_i \in \{1, \ldots, G\}$ denote the stratum membership of unit $i$, so that $i \in \lambda_{b_i}$. Blocks are constructed so that units within each block are similar in observed covariates, a design known as covariate stratification. Within each block $g$, treatment is assigned so that exactly $T_g \in \{1,\dots, N_g-1\}$ units receive treatment, with block-specific treatment share $\eta_g := T_g / N_g \in (0,1)$. Importantly, we allow $\{\eta_g\}_{g=1}^G$ (equivalently, $\eta_i=\Pr(D_i=1\mid X^{(n)})=\eta_{b_i}$) to depend on the full realized covariate matrix $X^{(n)}$ through the stratification rule $\Lambda(X^{(n)})$, and not only on $X_i$. Formally,
\[
(D_i : i \in \lambda_g) \,\Big|\, X^{(n)} \sim \operatorname{Unif} \left\{ \mathbf{d} \in \{0,1\}^{N_g} : \sum_{i \in \lambda_g} d_i = T_g \right\}, 
\qquad g = 1, \dots, G,
\]
where assignments are independent across blocks and each block treats exactly a share $\eta_g$ of its units (with $\eta_g$ allowed to vary across $g$). This flexible design allows both block sizes and assignment rates to reflect the logistical, ethical, or other substantive features of real-world experiments. Strata might correspond to geographic regions, clinics, schools, or other covariate-based groups, with treatment shares $\eta_g$ chosen to meet power or policy targets.

This blockwise randomization ensures that, conditional on the full set of observed covariates, the collection of potential outcomes and selection indicators $\{Y_i^*(1), Y_i^*(0), S_i(1), S_i(0)\}_{i=1}^n$ is jointly independent of the treatment vector $(D_1, \dots, D_n)$. Each unit in a given block $g$ has assignment probability $\eta_g$. Given this assignment mechanism and the conditions described above, the marginal distribution of the observed data $(X_i, D_i, Y_i, S_i)$ is identical across $i = 1, \dots, n$, and does not depend on $n$. However, due to the fixed numbers of treated units within blocks, the observed data are not generally independent across units. As a result, the collection $\{(X_i, D_i, Y_i, S_i)\}_{i=1}^n$ forms an identically distributed, but dependent, sample from a common distribution. When the $\eta_g$ differ across blocks, unconditional independence between $D_i$ and potential outcomes need not hold at the sample level because treated and control groups weight blocks differently.

Importantly, fixing $T_g$ within each block introduces negative correlation among assignments, which has substantial consequences for inference. Accounting for the exact assignment mechanism is therefore critical for valid inference, as it ensures that variance estimators properly reflect this dependence structure and fully realize the efficiency gains available under design-consistent variance estimation.

To establish the large-sample results used later for inference with heterogeneous treatment shares, we state a set of design and regularity conditions under which our Lee-type estimators satisfy a central limit theorem we will propose later. The conditions are analogous to those in \citet{bai2023efficiency} but extended to allow block-specific treatment shares, specifically with a change in Assumption \ref{ass:Ahet}. They pertain only to the design and regularity required for inference, and identification assumptions are presented in Section \ref{sec:pi_assumptions}.

\begin{assumption}[i.i.d.\ sampling]\label{ass:iid}
Let $Q$ be the marginal distribution of $(Y_i^*(1),Y_i^*(0),S_i(1),\\S_i(0),X_i)$ and let $Q_n$ be the joint distribution of $\big(Y^{*(n)}(1),Y^{*(n)}(0),S^{(n)}(1),S^{(n)}(0),X^{(n)}\big)$. We assume i.i.d.\ sampling from a superpopulation, i.e.\ $Q_n = Q^n$.
\end{assumption}
\begin{assumption}[Covariate balancing]\label{ass:B}
As $n\to\infty$ with $G\to\infty$, block sizes are uniformly bounded, i.e.\ there exists a constant
$\bar N<\infty$ such that
\[
\Pr\!\left(\max_{1\le g\le G} N_g \le \bar N\right)\ \longrightarrow\ 1.
\]
Moreover,
\[
\frac{1}{n}\sum_{g=1}^{G}\ \max_{i,j\in\lambda_g}\|X_i-X_j\|^2 \xrightarrow{p} 0,
\]
so that blocks become asymptotically homogeneous in covariates.
\end{assumption}
\begin{assumption}[Stratification]\label{ass:Ahet}
Given the full covariate matrix $X^{(n)}=(X_1,\dots,X_n)$ (equivalently, conditioning on stratum membership $b_i$ and covariates $X_i$), the collection $\{Y_i^*(1),Y_i^*(0),\\S_i(1),S_i(0)\}_{i=1}^n$ is jointly independent of the treatment vector $(D_1,\dots,D_n)$. Furthermore, conditional on $X^{(n)}$, the blockwise assignment vectors $\{(D_i: i\in\lambda_g)\}_{g=1}^G$ are independent across blocks, and within each block $\lambda_g$ treatment is assigned uniformly over all allocations with exactly $T_g$ treated units:
\[
(D_i : i \in \lambda_g)\,\big|\,X^{(n)} \sim \mathrm{Unif}\!\left\{\mathbf d\in\{0,1\}^{N_g} : \sum_{i\in\lambda_g} d_i=T_g\right\},\qquad g=1,\dots,G,
\]
with possibly heterogeneous shares $\eta_g:=T_g/N_g\in(0,1)$. Here $\eta_g$ (equivalently, $\eta_i=\Pr(D_i=1\mid X^{(n)})=\eta_{b_i}$) is allowed to depend on the full covariate matrix $X^{(n)}$ through the design map $X^{(n)}\mapsto \big(\Lambda(X^{(n)}),\{T_g\}_{g=1}^G\big)$, rather than only on $X_i$.
\end{assumption}
\begin{assumption}[Moment regularity]\label{ass:C} Let $m(\cdot)=(m_s(\cdot):1\le s\le d_\theta)'$. The moment functions are such that

\begin{enumerate}[label=(\alph*)]
\item For every $\epsilon>0$,
\[
\inf_{\theta\in\Theta:\,\|\theta-\theta_0\|>\epsilon}\ \big\|\E\big[m(X_i,D_i,Y_i;\theta)\big]\big\|\;>\;0.
\]

\item $\E[m(X_i,D_i,Y_i;\theta)]$ is differentiable at $\theta_0$ with a nonsingular derivative
\[
M \;=\; \left.\frac{\partial}{\partial\theta'}\,\E\!\left[m(X,D,Y;\theta)\right]\right|_{\theta=\theta_0}\, .
\]

\item For $1\le s\le d_\theta$ and $d\in\{0,1\}$,
\[
\E\!\left[\big(m_s(X,d,Y(d),\theta)-m_s(X,d,Y(d),\theta_0)\big)^2\right]\ \to\ 0
\quad\text{as }\theta\to\theta_0.
\]

\item For $1\le s\le d_\theta$, $\{m_s(x,d,y,\theta):\theta\in\Theta\}$ is pointwise measurable in the sense that there exists a countable set $\Theta^\ast$ such that for each $\theta\in\Theta$, there exists a sequence $\{\theta_m\}\subset\Theta^\ast$ such that $m_s(x,d,y,\theta_m)\to m_s(x,d,y,\theta)$ as $m\to\infty$ for all $x,d,y$.

\item (i) $\displaystyle \sup_{\theta\in\Theta}\ \E\!\left[\|m(X,d,Y(d),\theta)\|\right]<\infty$ for $d\in\{0,1\}$.
(ii) $\{m(x,1,y,\theta):\theta\in\Theta^\ast\}$ and $\{m(x,0,y,\theta):\theta\in\Theta^\ast\}$ are $Q$-Donsker.

\item For $d\in\{0,1\}$, let $\mu_d(x,\theta):=\E[m(X,d,Y(d),\theta)\mid X=x]$.
There exists $L_d<\infty$ such that for all $x,x'\in\mathbb R^{d_x}$,
\[
\sup_{\theta\in\Theta^\ast}\,|\mu_d(x,\theta)-\mu_d(x',\theta)|
\ \le\ L_d\|x-x'\|.
\]
\end{enumerate}
\end{assumption}

\begin{remark}[Label-based stratification]
The setting in this section does not directly extend to designs with many small blocks defined purely by discrete labels, such as schools, clinics, or randomization waves. Appendix~\ref{sec:var_discrete_blocks} treats this case separately, using a label-based triangular-array framework appropriate for experiments with many small strata. The results are somewhat similar, despite different proofs.
\end{remark}

\section{Inference under Equal-Share Stratification}
\label{sec:variance_equal_share}

% ----------------------------------------------------------------------
\subsection{Setup: Equal-Share Stratification}
\label{sec:design_equal_share_real}
% ----------------------------------------------------------------------

This section specializes the design to a common treatment share across strata. Let $\Lambda(X^{(n)})=\{\lambda_g: g=1,\dots,G\}$ denote the observable blocks with equal sizes $N_g=k$ and labels $b_i\in\{1,\dots,G\}$. Within block $g$, exactly $\ell$ units are assigned to treatment without replacement, with a common share $\eta:=\ell/k\in(0,1)$ for all $g$. To formally characterize the assignment mechanism in this equal-share stratification setting, we state the following assumption:
\newcounter{saveassumpES}
\setcounter{saveassumpES}{\value{assumption}}
\renewcommand{\theassumption}{3$'$}
\begin{assumption}[Equal-share stratification]\label{ass:A}
Given the full covariate matrix $X^{(n)}=(X_1,\\\dots,X_n)$ (equivalently, conditioning on stratum membership $b_i$ and covariates $X_i$), the collection $\{Y_i^*(1),Y_i^*(0),S_i(1),S_i(0)\}_{i=1}^n$ is jointly independent of the treatment vector $(D_1,\dots,D_n)$.
Furthermore, conditional on $X^{(n)}$, the blockwise assignment vectors $\{(D_i: i\in\lambda_g)\}_{g=1}^G$ are independent across blocks, and within each block $\lambda_g$ treatment is assigned uniformly over all allocations with exactly $\ell$ treated units:
\[
(D_i: i\in\lambda_g)\,\big|\,X^{(n)} \sim \operatorname{Unif}\!\left\{\mathbf d\in\{0,1\}^{k}:\ \sum_{i\in\lambda_g} d_i=\ell\right\},\qquad g=1,\dots,G,
\]
with $\eta:=\ell/k\in(0,1)$ common to all $g$.
\end{assumption}
\renewcommand{\theassumption}{\arabic{assumption}}
\setcounter{assumption}{\value{saveassumpES}}

% ----------------------------------------------------------------------
\subsection{Partial Identification and Consistency of Lee Bounds under Equal-Share Stratification}
\label{sec:pi_assumptions}
% ----------------------------------------------------------------------

We now discuss the conditions used for the partial identification of the always-observed effect in stratified designs. In \citet{lee2009training}, the unconditional Lee bounds are derived under an independence condition in which treatment is unconditionally independent of the vector of potential outcomes and selection indicators. In our setting, Lemma~\ref{lem:UI_equalshare} shows that this unconditional independence is implied by the equal-share stratified randomization design. In addition, we impose the monotonicity in selection condition from the original Lee framework, which does not follow from the experimental design and must be justified on empirical grounds.

\begin{assumption}[Monotone selection]\label{ass:MONO}
\[
\Pr\!\big[S_i(1)\ge S_i(0)\big]=1.
\]
Equivalently, one may assume the reverse monotonicity $\Pr[S_i(1) \leq S_i(0)] = 1$; all results then follow by swapping treated and control in the trimming step. We maintain $\Pr[S_i(1) \geq S_i(0)] = 1$ for exposition.
\end{assumption}

Under equal-share stratified randomization, treated and control units draw the same mixture over strata, so treatment is unconditionally independent of the vector of potential outcomes and selection indicators. The next result formalizes this implication. At the same time, fixing $T_g$ within each block induces negative within-block correlation in assignments. Thus, observations are identically distributed but dependent, and inference must account for this design-induced dependence.

\begin{lemma}[Equal-share stratification implies unconditional independence]\label{lem:UI_equalshare}
Suppose Assumption~\ref{ass:A} holds. Then
\[
(Y_i^*(1),Y_i^*(0),S_i(1),S_i(0))\ \perp\!\!\!\perp\ D_i.
\]
\end{lemma}

\begin{proof}
Block randomization implies stratum-level independence,
\(
(Y_i^*(1),Y_i^*(0),S_i(1),S_i(0))\ \perp\!\!\!\perp\ D_i \mid b_i,
\)
where $b_i$ denotes stratum membership. Equal shares give $\Pr(D_i=1\mid b_i=g)=\eta$ for all $g$, hence $D_i\perp b_i$. Let $W_i:=(Y_i^*(1),Y_i^*(0),S_i(1),S_i(0))$. Then
\[
\Pr(D_i=1\mid W_i)
= \sum_{g=1}^G \Pr(D_i=1\mid b_i=g)\,\Pr(b_i=g\mid W_i)
= \eta \sum_{g=1}^G \Pr(b_i=g\mid W_i)
= \eta,
\]
which is constant in $W_i$. Therefore $(Y_i^*(1),Y_i^*(0),S_i(1),S_i(0))\ \perp\!\!\!\perp\ D_i$.
\end{proof}

In the blocked design studied here, each stratum uses a common treatment share. By Lemma~\ref{lem:UI_equalshare}, block randomization delivers unconditional independence at the sample level, while selection monotonicity (Assumption~\ref{ass:MONO}) remains a substantive identifying restriction. Under these conditions, the unconditional Lee bounds interval continues to characterize the identified set for the ATE among always-observed units, and the sharpness argument is unchanged. The next result states partial identification and sharpness in this setting. We then turn to consistency.

\begin{proposition}[Partial ID and sharpness under equal-share stratification]
\label{prop:lee_sharp_equalshare}
Suppose Assumption~\ref{ass:MONO} holds, and Assumptions~\ref{ass:iid}, \ref{ass:B}, \ref{ass:A}, and \ref{ass:C} also hold, which specify an equal-share stratification design. Let
\[
q \;=\; 1 - \frac{\Pr(S=1\mid D=0)}{\Pr(S=1\mid D=1)} \in [0,1],
\]
let $y_\alpha$ denote the $\alpha$-quantile of $Y$ among $\{D=1,S=1\}$, and define
\begin{align*}
\mu_0 &= \E\big[\,Y \mid D=0, S=1\,\big],\\
\mu_1^{LB} &= \E\big[\,Y \mid D=1, S=1,\, Y \le y_{1-q}\,\big],\\
\mu_1^{UB} &= \E\big[\,Y \mid D=1, S=1,\, Y \ge y_{q}\,\big].
\end{align*}
Then the average treatment effect among always-observed units,
\[
\Delta_{AO} \;=\; \E\!\left[\,Y^*(1)-Y^*(0)\mid S(1)=S(0)=1\,\right],
\]
is partially identified by the Lee bounds interval
\[
\Delta_{AO} \in \big[\,\mu_1^{LB}-\mu_0,\; \mu_1^{UB}-\mu_0\,\big],
\]
and these bounds are sharp. Hence, with equal-share stratification, the identified set coincides with the i.i.d.\ case, and Lee bounds remain sharp.
\end{proposition}

\begin{proof}
The argument follows \citet[Prop.~1(a)]{lee2009training}, which relies on unconditional independence and monotonicity in selection. Under equal-share stratification, Lemma~\ref{lem:UI_equalshare} shows that Assumptions~\ref{ass:iid}, \ref{ass:B}, \ref{ass:A}, and \ref{ass:C} imply unconditional independence. Together with Assumption~\ref{ass:MONO}, this verifies all conditions used in \citet{lee2009training}, so his proof applies verbatim, including the least-favorable construction that yields sharpness.
\end{proof}

\vspace{0.3cm}

With identification in place, the only departure from the i.i.d.\ case concerns the uniform convergence step used in the consistency proof. The i.i.d.\ ULLN invoked by \citet[Thm.~2.6]{newey1994large}, cited by \citet{lee2009training}, is replaced by a version valid for block randomization with a common share, built from finite-population limit theory and uniform convergence results. The proposition below records the resulting consistency claim, and the proof supplies only this substitution.

\begin{proposition}[Consistency under equal-share stratification]\label{prop:lee_consistency_equalshare}
Suppose Assumption~\ref{ass:MONO} holds, and Assumptions~\ref{ass:iid}, \ref{ass:B}, \ref{ass:A} and \ref{ass:C} also hold, which specify an equal-share stratification design with a common treatment share and independent blocks. Let $\hat q = 1 - \widehat{\Pr}(S=1\mid D=0)/\widehat{\Pr}(S=1\mid D=1)$, let $\hat y_\alpha$ denote the empirical $\alpha$-quantile of $Y$ among $\{D=1,S=1\}$, and define the sample analogs
\begin{align*}
\hat\mu_0 &= \E_n[\,Y \mid D=0,S=1\,],\\
\hat\mu_1^{LB} &= \E_n[\,Y \mid D=1,S=1,\, Y \le \hat y_{1-\hat q}\,],\\
\hat\mu_1^{UB} &= \E_n[\,Y \mid D=1,S=1,\, Y \ge \hat y_{\hat q}\,].
\end{align*}
Then
\[
\big(\hat\mu_1^{LB} - \hat\mu_0,\; \hat\mu_1^{UB} - \hat\mu_0\big)
\;\xrightarrow{p}\;
\big(\mu_1^{LB} - \mu_0,\; \mu_1^{UB} - \mu_0\big).
\]
\end{proposition}

\begin{proof}
The argument adapts \citet[Prop.~2, Proof]{lee2009training} to equal-share stratification by replacing the i.i.d.\ ULLN with a version valid for block randomization. Under Assumptions~\ref{ass:iid}, \ref{ass:B}, \ref{ass:A} and \ref{ass:C}, Lemma~\ref{lem:UI_equalshare} shows that unconditional independence holds, so the conditions used in \citet{lee2009training} are satisfied under Assumption~\ref{ass:MONO}. Blockwise LLNs follow from \citet[][Lem.~2.1, Rem.~2.1]{hajek1960limiting}, and uniform convergence results (including triangular arrays) follow from \citet[Thm.~1, Thm.~3(a), Thm.~5; end of §3.4]{andrews1992generic}. With this substitution, \citet[Thm.~2.6]{newey1994large} yields consistency exactly as in \citet{lee2009training}. See Appendix~\ref{app:proof_consistency_equalshare} for a detailed proof.
\end{proof}

\vspace{0.3cm}

Hence, partial identification and sharpness carry over to the equal-share stratified design, and consistency follows once the uniform-convergence step is adapted to blocked assignment. What changes is sampling behavior, because assignment without replacement creates within-block dependence, so inference must use variance formulas that respect this feature.

% ----------------------------------------------------------------------
\subsection{Asymptotic Distribution and Variance for Lee Bounds under Equal-Share Stratification}
\label{sec:variance_and_assumptions_equal_share}
% ----------------------------------------------------------------------

Understanding the asymptotic behavior of Lee bounds under equal-share stratification is crucial for valid and efficient inference in finely stratified experimental designs. This setting is characterized by partitioning the population into blocks with a common treatment share, inducing within-block dependence that alters the variance relative to i.i.d.\ sampling. To rigorously justify inference, we work under Assumption~\ref{ass:A} together with the design and regularity conditions in Assumptions~\ref{ass:iid}, \ref{ass:B} and \ref{ass:C}. These are the same as in \citet{bai2023efficiency}, so we first represent the Lee lower bound as a just-identified GMM estimator and then apply their central limit theorem and variance results to this moment representation.

Let $Z_i = (Y_i,S_i,D_i,X_i)$ and define
$L_i = \mathbf{1}\{D_i=1,\,S_i=1,\,Y_i \le y_{1-p}\}$ and
$U_i = \mathbf{1}\{D_i=1,\,S_i=1,\,Y_i > y_{1-p}\}$,
where $y_{1-p}$ is the $(1-p)$ quantile of $Y_i$ among $\{D_i=1,S_i=1\}$. Parameterize
$\theta = (\Delta_{AO}^{LB},\mu_0,p,\alpha)^\prime$,
where $\Delta_{AO}^{LB}$ is the Lee lower bound, $\mu_0 = \E[Y_i \mid D_i=0,S_i=1]$ is the control mean, $p$ is the trimming share, and $\alpha = \Pr(S_i=1\mid D_i=0)$ is the control-group selection rate, so that the trimmed treated mean is $\mu_1^{LB} = \mu_0 + \Delta_{AO}^{LB}$. The corresponding moment vector is
\[
m(Z_i,\theta) =
\begin{pmatrix}
\big(Y_i - (\mu_0 + \Delta_{AO}^{LB})\big)\, S_i D_i L_i
  - \big(Y_i - \mu_0\big)\, S_i (1-D_i)
\\[0.35em]
\big(U_i - p\big)\, S_i D_i
\\[0.35em]
\left(S_i - \dfrac{\alpha}{1-p}\right) D_i
\\[0.55em]
\big(S_i - \alpha\big)\, (1-D_i)
\end{pmatrix},
\qquad
\E[m(Z_i,\theta_0)] = 0.
\]
In this parametrization, the Lee lower bound is the first component of $\theta$, and an analogous moment system characterizes the upper bound by using indicators that trim the lower tail of the treated outcome distribution instead of the upper tail.

\begin{proposition}[Asymptotics of Lee bounds under equal-share stratification]
\label{prop:lee-clt}
Sup\-pose Assumptions~\ref{ass:iid}, \ref{ass:B}, \ref{ass:A} and \ref{ass:C} hold. Then, the Lee bounds estimator obeys the following large-sample distribution:
\[
\sqrt{n}\,(\hat\theta_n - \theta_0) \;\dto\; \mathcal{N}(0, V_\eta),
\]
where $V_\eta = M^{-1}\Omega_\eta M^{-T}$, with $M = \partial_\theta \E[m(Z_i, \theta)]\big|_{\theta = \theta_0}$, and the variance $\Omega_\eta$ decomposes as
\[
\Omega_\eta = \E\left[ \eta\,\Var(m_1 \mid X) + (1-\eta)\,\Var(m_0 \mid X) \right]
    + \Var\left[ \eta\,\mu_1(X) + (1-\eta)\,\mu_0(X) \right],
\]
where $m_d = m(Z_i, \theta_0) \mathbf{1}\{D_i = d\}$ and $\mu_d(X) = \E[m_d \mid X]$ for $d \in \{0,1\}$. In particular, the asymptotic variance of the Lee lower bound corresponds to the $(1,1)$ element of $V_\eta$.
\end{proposition}

\begin{proof}
The proof is immediate from Theorem 3.1 of \cite{bai2023efficiency}, since the assumptions made are equivalent to Assumptions 3.1-3.3 (plus i.i.d.\ sampling) in that paper. Covariate balance ensures that block-level imbalances vanish at the $\sqrt{n}$ rate, allowing the asymptotic variance to be accurately captured by $V_\eta$. Furthermore, the Lee bounds estimator meets the required moment regularity conditions because its moments are defined only by linear and indicator functions.
\end{proof}

An implementation caveat for variance estimation concerns the moment vector. In \citet{lee2009training}, under i.i.d.\ assignment, the moment vector used in the sandwich variance targets only the trimmed treated mean, and the control mean is incorporated only when assembling the final variance of the estimator. Under equal-share stratification, the trimmed treated mean and the control mean are correlated by within-block assignment. Omitting the control mean from the moment vector ignores this covariance and can bias the estimated standard errors. Accordingly, we specify the first moment condition as the difference between the trimmed treated mean and the control mean and then estimate the variance with the design-consistent formula of \citet{bai2023efficiency}.

Given this moment specification, the variance decomposition presented in \citet{bai2023efficiency} applies directly to the Lee bounds estimator and enables consistent estimation even with small block sizes. Notably, omitting the block correction term, as in conventional i.i.d.\ variance estimation, leads to conservative standard errors. These observations reinforce an important theoretical result: the variance under equal-share stratification is always less than or equal to that under i.i.d.\ assignment \citep{bai2023efficiency}. Incorporating the blocked assignment structure into variance estimation is thus essential, not only to avoid overstating uncertainty but also to fully realize the efficiency gains in empirical applications.

\begin{remark}[Equal shares and label-based designs]
With equal shares, unconditional independence holds by design and standard Lee bounds apply in both covariate-based and label-based stratified designs. The main difference lies in how precision can be improved. When strata are formed by grouping units with similar covariates and within-stratum diameters shrink as $n$ grows, it is possible to borrow information from nearby strata. In label-based designs with many small blocks, there are no meaningful neighbors to use for pooling. Appendix~\ref{sec:var_discrete_blocks} considers this setting and requires at least two treated and two control observations in each stratum.
\end{remark}

% ----------------------------------------------------------------------
\subsection{Simulation Results}
\label{sec:lim_cond}
% ----------------------------------------------------------------------

We now present simulation evidence with two purposes. First, we assess the performance of the design-consistent variance estimator for Lee bounds under equal-share stratification. Second, we study how conditional Lee bounds, obtained by trimming within each stratum and then aggregating the resulting intervals, perform in stratified designs with small strata. The estimation problems that arise in this setting motivate a pooled alternative that draws on information across strata.

In an equal-share matched-pair design, the design-consistent variance estimator based on \citet{bai2023efficiency} delivers standard errors that closely match the Monte Carlo standard deviation of the Lee bounds estimator. By contrast, the conventional i.i.d.\ variance estimator substantially overstates uncertainty. Its average standard error is $0.0569$ against an empirical standard deviation of $0.0397$, inflating coverage to about $99.5\%$ instead of the nominal $95\%$. This confirms that incorporating the blocked assignment structure is essential for accurate inference under equal-share stratification.

A second simulation considers stratified designs with many small strata. A natural alternative in this setting is to apply Lee bounds conditional on strata and then aggregate the resulting intervals. In our designs, the difficulty lies in estimation rather than identification. With only a few treated units and controls in each stratum, the cell-level trimming share and the associated trimmed means are based on very little information and can be sensitive to single observations, including in sparse or nearly empty cells. This produces conditional Lee bounds that are noisy and sometimes poorly centered, even when the assumptions behind Lee bounds hold. Appendix~\ref{app:sim_dgpss} reports the simulation designs, numerical results and a broader discussion of additional limitations of conditional Lee bounds in sparse stratified designs.

\section{Lee–IPW Bounds under Heterogeneous-Shares Stratification}
\label{sec:lee_ipw_estimator}

\subsection{Partial Identification Strategy with Lee-IPW Bounds}
\label{sec:id_het}

Many field trials and randomized experiments assign different fractions of units to treatment across strata. Examples include oversampling in small sites to ensure power, demographic quotas that create unequal assignment rates, and logistical constraints that yield classes, villages, or clinics of different sizes and treatment proportions. Under stratified designs with heterogeneous treatment shares, the design yields only stratum-level (conditional) independence, so unconditional independence between treatment and potential outcomes generally fails. Consequently, conventional unconditional Lee bounds are invalid. The following lemma formalizes this point.

\begin{lemma}[Heterogeneous-shares stratification implies stratum-level independence]\label{lem:SLI_hetero}
Suppose Assumption~\ref{ass:Ahet} holds. Then
\[
(Y_i^*(1),Y_i^*(0),S_i(1),S_i(0))\ \perp\!\!\!\perp\ D_i \ \big|\  b_i .
\]
Moreover, writing $W_i:=(Y_i^*(1),Y_i^*(0),S_i(1),S_i(0))$,
\[
\Pr(D_i=1\mid W_i)=\sum_{g=1}^G \eta_g\,\Pr(b_i=g\mid W_i),
\]
so unconditional independence $(W_i\perp\!\!\!\perp D_i)$ generally fails unless the right-hand side is constant in $W_i$ (e.g., equal shares or stratum composition invariant in $W_i$). In heterogeneous-shares settings, this stratum-level independence is therefore the appropriate independence condition.
\end{lemma}

\begin{proof}
Uniform assignment within each $\lambda_g$ implies
\(
(Y_i^*(1),Y_i^*(0),S_i(1),S_i(0))\ \perp\!\!\!\perp\ D_i \mid b_i,
\)
establishing stratum-level independence. For the marginal statement,
\[
\Pr(D_i=1\mid W_i)
=\sum_{g=1}^G \Pr(D_i=1\mid b_i=g)\,\Pr(b_i=g\mid W_i)
=\sum_{g=1}^G \eta_g\,\Pr(b_i=g\mid W_i),
\]
which is constant in $W_i$ only in the noted special cases. Otherwise, unconditional independence fails.
\end{proof}

Building on Lemma~\ref{lem:SLI_hetero}, the partial identification of treatment effects with heterogeneous treatment shares must work with this stratum-level condition that reflects the block randomization design. The Lee-IPW estimator adopts the stratum-level independence implied by Assumption~\ref{ass:Ahet}, together with Assumption~\ref{ass:MONO}, allowing for flexible assignment probabilities while restoring valid partial identification.

Our target estimand is the average treatment effect among always-observed units,
$
\Delta_{AO} = \mathbb{E}\left[ Y^*(1) - Y^*(0) \,\big|\, S(1) = S(0) = 1 \right].
$
The Lee–IPW method provides a principled strategy for handling stratified experiments with heterogeneous treatment shares. By employing inverse probability weighting (IPW), it re-balances the data to account for varying assignment proportions across strata and then applies a single trimming share to the weighted sample. This procedure yields more stable and reliable bounds, particularly in studies with many small or sparsely treated strata, where traditional approaches can be biased or infeasible.

In this framework, let $p = \Pr(D = 1)$ denote the overall treated share, and let $w_{c,g} = (1-p)/(1-\eta_g)$ denote the weight applied to control units in stratum $g$. After weighting by $w_{c,g}$, the distribution of control observations matches the distribution of controls in the overall sample. This adjustment corrects for the over- or under-representation of controls from strata with heterogeneous treatment shares. The always-observed control mean is then identified as
\[
\mu_{0}=E[Y^{*}(0)\mid AO \text{ (Always-observed)}]=\frac{E[(1-D)S\,w_{c,g}Y]}{E[(1-D)S\,w_{c,g}]}.
\]

Similarly, for the treated–observed units, we define $\delta = \Pr(D_{g,i} = 1 \mid AO)$ as the probability that a unit is treated, conditional on being always-observed, and reweight each treated–observed outcome by $h_{g,i} = \delta / \eta_g$, so that $\widetilde{Y}_{g,i} = h_{g,i} Y_{g,i}$. This ensures that the mean of the reweighted treated outcomes corresponds to the mean for always-observed treated potential outcomes: $\mathbb{E}[\widetilde{Y} \mid D=1, AO] = \mathbb{E}[Y^*(1) \mid AO]$. It is the same principle that motivated $w_{c,g}$, now applied to the treated arm with a focus on always-observed units. The weight can be written as a functional of the data and design (see Appendix~\ref{sec:iden_Y1} for details), thus it is identified.

To identify the appropriate trimming share, monotonicity implies that the excess response rate in the treated group, after control reweighting to mirror the treated units distribution across strata, corresponds to the proportion $q$ of “induced” outcomes to be trimmed. With $w_{q,g} = \frac{\eta_g (1-p)}{(1-\eta_g)p}$ denoting the weight applied to control units in stratum $g$ when constructing $q$, the trimming share is given by
\[
q = \frac{\Pr[S=1 \mid D=1] - \mathbb{E}[S\,w_{q,g} \mid D=0]}{\Pr[S=1 \mid D=1]}.
\]

Finally, the bounds for the always-observed average treatment effect are constructed by trimming the lower (or upper) $q$ fraction from the distribution of reweighted treated–observed outcomes. Let $\tilde{y}_\alpha$ be the $\alpha$-quantile of $\widetilde{Y}$. The bounds are given by
\[
\mu_1^{LB} = \mathbb{E}[\widetilde{Y} \mid D=1, S=1, \widetilde{Y} \le \tilde{y}_{1-q}], \quad
\mu_1^{UB} = \mathbb{E}[\widetilde{Y} \mid D=1, S=1, \widetilde{Y} \ge \tilde{y}_{q}],
\]
and thus,
\[
\Delta_{AO} \in [\mu_1^{LB} - \mu_0,\, \mu_1^{UB} - \mu_0].
\]

Considering the preceding definitions and construction, the identified set for the always-observed treatment effect is characterized by the following result.

\begin{proposition}[Partial Identification with Lee-IPW Bounds]
Under the stratum-level independence condition implied by the stratified assignment (Lemma \ref{lem:SLI_hetero}) and Assumption~\ref{ass:MONO} (monotonicity), the average treatment effect among always-observed units,
\[
\Delta_{AO} = \mathbb{E}\left[ Y^*(1) - Y^*(0) \mid S(1) = S(0) = 1 \right],
\]
is partially identified by the interval
\[
[\mu_1^{LB} - \mu_0,\ \mu_1^{UB} - \mu_0],
\]
where
\begin{align*}
\mu_0 &= \frac{\mathbb{E}[(1-D)S w_{c,g} Y]}{\mathbb{E}[(1-D)S w_{c,g}]}, \\
\mu_1^{LB} &= \mathbb{E}[\widetilde{Y} \mid D=1, S=1,\, \widetilde{Y} \leq \tilde{y}_{1-q}], \\
\mu_1^{UB} &= \mathbb{E}[\widetilde{Y} \mid D=1, S=1,\, \widetilde{Y} \geq \tilde{y}_q], \\
q &= \frac{\Pr[S=1 \mid D=1] - \mathbb{E}[S w_{q,g} \mid D=0]}{\Pr[S=1 \mid D=1]}.
\end{align*}
All weights and quantiles are as defined in the prior discussion.
\end{proposition}
\begin{proof}
See Appendix~\ref{sec:identification} for a detailed proof and derivation.
\end{proof}

\vspace{0.3cm}

Conditional Lee bounds are sharp under stratum-level randomization and monotonicity, as they fully exploit differential selection within each block, but their reliability breaks down when strata are small, sparsely treated, or contain outliers, often resulting in bias, high variance, or undefined bounds. The Lee–IPW procedure, while not sharp because it does not use all within-block selection information and may yield somewhat wider intervals, addresses these limitations by pooling information and applying a unified trimming rule. This approach provides stable and feasible inference without ad hoc adjustments for ill-defined trimming shares, offering robust partial identification even in finely stratified or sparsely populated designs. In practice, Lee–IPW intervals are especially valuable in applied settings where strata are small or imbalanced, and offer a reliable tool for empirical analysis in challenging experimental designs.

% ----------------------------------------------------------------------
\subsection{Lee-IPW Estimator and Moment Conditions}
\label{sec:ipw_est_mom}
% ----------------------------------------------------------------------

The Lee–IPW estimator translates the identification strategy from the previous section into a practical procedure for bounding the average treatment effect among always-observed units. All calculations rely solely on observed data and the known randomization shares for each stratum, and the estimator is consistent even with small or highly unbalanced blocks.

The estimation proceeds by first computing the realized treatment rate $\hat{p} = n^{-1} \sum_{i=1}^n D_i$ and the design share $\eta_g = T_g / N_g$ within each stratum. For each unit, block weights are assigned to controls and for trimming purposes as $\hat{w}_{c,g} = (1-\hat{p})/(1-\eta_g)$ and $\hat{w}_{q,g} = \eta_g (1-\hat{p}) / [(1-\eta_g)\hat{p}]$. In the following, all summations are over the entire sample.

To estimate the trimming share, which is fundamental for partial identification, the procedure calculates
\[
\hat{q} = 1 - \frac{\hat{p} \sum_{i} (1-D_i) S_i \hat{w}_{q,b_i}}{(1-\hat{p}) \sum_{i} D_i S_i}, \qquad \hat{q} \in [0, 1].
\]
This value represents the proportion of treated–observed outcomes to trim, corresponding to the excess response rate attributable to treatment, after adjusting for heterogeneous assignment.

Next, the always-observed treatment probability is estimated by first computing, within each stratum, the proportion of controls observed: $\hat{m}_g = \sum_{i\in\lambda_g} (1-D_i) S_i / \sum_{i\in\lambda_g} (1-D_i)$. The overall probability is then
\[
\hat{\delta} = \frac{\sum_i D_i \hat{m}_{b_i}}{\sum_i \hat{m}_{b_i}}.
\]

Each treated, observed outcome is then reweighted as $\widetilde{Y}_i = (\hat{\delta} / \eta_{b_i}) Y_i$, aligning the mean of the reweighted treated sample with that for always-observed units. The empirical quantiles $\hat{\tilde{y}}_{1-\hat{q}}$ and $\hat{\tilde{y}}_{\hat{q}}$ among these values set the cutoffs for trimming. Then, the lower and upper bounds for the treated mean are
\[
\hat{\mu}_1^{LB} = \frac{\sum_{i} D_i S_i\, \mathbbm{1}\{\widetilde{Y}_i \leq \hat{\tilde{y}}_{1-\hat{q}}\} \, \widetilde{Y}_i}{\sum_{i} D_i S_i\, \mathbbm{1}\{\widetilde{Y}_i \leq \hat{\tilde{y}}_{1-\hat{q}}\}}, \qquad
\hat{\mu}_1^{UB} = \frac{\sum_{i} D_i S_i\, \mathbbm{1}\{\widetilde{Y}_i \geq \hat{\tilde{y}}_{\hat{q}}\} \, \widetilde{Y}_i}{\sum_{i} D_i S_i\, \mathbbm{1}\{\widetilde{Y}_i \geq \hat{\tilde{y}}_{\hat{q}}\}},
\]
while the control mean is estimated by
\[
\hat{\mu}_0 = \frac{\sum_{i} (1-D_i) S_i\, \hat{w}_{c,b_i}\, Y_i}{\sum_{i} (1-D_i) S_i\, \hat{w}_{c,b_i}}.
\]
The Lee–IPW bounds for the always-observed average treatment effect are hence given by
\[
\hat{\Delta}^{LB} = \hat{\mu}_1^{LB} - \hat{\mu}_0, \qquad \hat{\Delta}^{UB} = \hat{\mu}_1^{UB} - \hat{\mu}_0.
\]

To place the estimator within a more formal framework and to connect it with the inference procedures discussed in the next section, it is necessary to express the Lee–IPW bounds estimator as the solution to a system of moment equations. The characterization is summarized below.

\begin{proposition}[Moment Characterization of the Lee-IPW Lower Bound]
\label{prop:lee_ipw_moments}
Let $\theta = (\mu_{1}^{LB}, \mu_{0}, \tilde{y}_{1-q}, \delta, q)^{\top} \in \Theta \subset \mathbb{R}^5$ collect the parameters for the lower Lee-IPW bound, and let $Z_i = (Y_i, S_i, D_i, b_i)$ collect the observed outcome, selection indicator, treatment assignment, and stratum membership for unit $i$. Define the moment vector $m(Z_i, \theta) = (m_{1i}(\theta), m_{2i}(\theta), m_{3i}(\theta),\\ m_{4i}(\theta), m_{5i}(\theta))^\top$ by
\begin{align*}
m_{1i}(\theta) &= (\widetilde{Y}_i - \mu_{1}^{LB})\, D_i S_i\, \mathbbm{1}\{\widetilde{Y}_i \le \tilde{y}_{1-q}\}, \\
m_{2i}(\theta) &= (Y_i - \mu_0)\, (1-D_i) S_i\, w_{c,b_i}, \\
m_{3i}(\theta) &= [\mathbbm{1}\{\widetilde{Y}_i > \tilde{y}_{1-q}\} - q]\, D_i S_i, \\
m_{4i}(\theta) &= r_{b_i} (D_i - \delta), \\
m_{5i}(\theta) &= \frac{1-q}{p} D_i S_i - \frac{1}{1-p} (1-D_i) S_i\, w_{q,b_i},
\end{align*}
where $p = \Pr(D = 1)$ and
\[
r_{b_i} := \frac{\sum_{j \in \lambda_{b_i}} (1-D_j) S_j}{\sum_{j \in \lambda_{b_i}} (1-D_j)}
\]
denotes the observed selection rate among controls in the stratum containing unit $i$. Let $\Delta^{LB}$ denote the lower Lee-IPW bound constructed in Section~\ref{sec:id_het}. Under stratified randomization in Assumption~\ref{ass:Ahet}, so that Lemma~\ref{lem:SLI_hetero} holds, and monotone selection in Assumption~\ref{ass:MONO}, there exists $\theta_0 = (\mu_{1}^{LB}, \mu_{0}, \tilde{y}_{1-q}, \delta, q)^{\top} \in \Theta$ such that
\[
\mathbb{E}[m(Z_i, \theta_0)] = 0
\qquad\text{and}\qquad
\Delta^{LB} = \mu_{1}^{LB} - \mu_{0}.
\]
Moreover, any $\theta \in \Theta$ satisfying $\mathbb{E}[m(Z_i, \theta)] = 0$ has its first two components equal to the trimmed treated mean and control mean that define $\Delta^{LB}$ and its fifth component equal to the trimming share $q$. Hence the system $\mathbb{E}[m(Z_i, \theta)] = 0$ is equivalent to the Lee-IPW lower-bound construction.
\end{proposition}

The lower Lee–IPW bound is constructed using the system above; the upper bound replaces the indicator $\mathbbm{1}\{\widetilde{Y}_i \le y_{1-q}\}$ with $\mathbbm{1}\{\widetilde{Y}_i \ge y_{q}\}$ in the relevant moments. The just-identified GMM estimator $\hat{\theta}_n$ solves $\frac{1}{n} \sum_{i=1}^n m(Z_i, \hat{\theta}_n) = 0$, with Jacobian matrix $\hat{M}_n = \frac{1}{n} \sum_{i=1}^n \partial_{\theta} m(Z_i, \hat{\theta}_n)$. The estimator for the lower bound is then given by $\hat{\Delta}^{LB} = e_1^\top \hat{\theta}_n - e_2^\top \hat{\theta}_n$, where $e_1$ and $e_2$ are standard basis vectors selecting the appropriate components.

This unified approach allows the Lee–IPW estimator to be implemented efficiently in both large and small samples, providing robust inference for the always-observed average treatment effect even when standard approaches fail.

\begin{remark}[Heterogeneous shares and label-based designs]
Under heterogeneous treatment shares, randomization delivers independence within strata but not unconditional independence, so unconditional Lee bounds are not valid. The Lee-IPW bounds developed in this section restore validity under the covariate-based stratification framework used here. When strata are defined by many discrete labels, we consider the same Lee-IPW bounds, and Appendix~\ref{sec:var_discrete_blocks} establishes the corresponding partial identification and inference results in a label-based triangular-array setting.
\end{remark}

\section{Inference for Moment-Based Estimators under Heterogeneous-Shares Stratification}
\label{sec:lee_ipw_inference}

% ----------------------------------------------------------
\subsection{Asymptotic Distribution under Heterogeneous-Shares Stratification}
\label{sec:asy_het}
% ----------------------------------------------------------

This section develops inference procedures for the Lee-IPW bounds introduced in Section~\ref{sec:lee_ipw_estimator}, focusing on the statistical challenges that arise when treatment shares differ across strata. Such heterogeneity is common in empirical studies, but falls outside the equal-share framework discussed in Section~\ref{sec:variance_equal_share} and in \citet{bai2023efficiency}. Consequently, standard variance formulas and asymptotic results do not generally apply, making it necessary to extend existing theory for valid inference. The main contribution here, then, is to generalize the central limit theorem of \citet{bai2023efficiency} to accommodate stratum-specific treatment shares, covering any moment estimator that satisfies the required regularity conditions and yielding valid, design-consistent inference in complex settings. The Lee-IPW bounds estimator, introduced in this article, demonstrates the usefulness of this broader framework.

Now, we formally analyze asymptotic inference for a broad class of generalized method of moments (GMM) estimators under stratified experimental designs with possibly heterogeneous treatment shares, $\eta_g$, across strata. The approach applies to any estimator whose moment vector satisfies the regularity conditions in Assumption~\ref{ass:C}, which ensures the validity of the central limit theorem in this setting. The Lee-IPW bounds introduced in Section~\ref{sec:lee_ipw_estimator} represent an important instance of this general framework, but the results here accommodate any such moment-based estimator.

Let $m(Z_i,\theta)$ denote a vector of moment functions that may depend on observed data and known design shares, and let
\[
M := \left.\frac{\partial}{\partial\theta'}\,\E\!\left[m(Z_i,\theta)\right]\right|_{\theta=\theta_0}
\]
be the corresponding population Jacobian. Allowing for heterogeneous treatment shares leads to an asymptotic ``meat''
matrix of the form
\begin{align*}
\Omega_{\eta(\cdot)}
  =\;&\E\!\left[
      \eta(X)\,\Var\!\big(m(X,1,Y(1);\theta_0)\mid X\big)
      + \big(1-\eta(X)\big)\,\Var\!\big(m(X,0,Y(0);\theta_0)\mid X\big)
     \right] \\
  &\quad + \Var\!\left[
      \eta(X)\,\mu_1(X) + \big(1-\eta(X)\big)\,\mu_0(X)
    \right],
\end{align*}
where $\mu_d(X) := \E[m(X,d,Y(d);\theta_0)\mid X]$ and $\eta(X):=\Pr(D=1\mid X)$.
In this notation, $X$ denotes the observed covariates used by the design, which may include stratum indicators.

\begin{theorem}[Asymptotics under heterogeneous-shares stratification]\label{thm:clt_het_gmm}
Let $\hat{\theta}_n$ be the GMM estimator defined by the sample analog of $m(Z_i,\theta)$.
Suppose Assumptions~\ref{ass:iid}, \ref{ass:B}, \ref{ass:Ahet}, \ref{ass:C}, \ref{ass:CB-LLN}, and \ref{ass:CB-IMB} hold.
Then,
\[
\sqrt{n}\,(\hat\theta_n-\theta_0)\ \dto\ \mathcal N(0,V_\ast),
\qquad
V_\ast = M^{-1}\Omega_{\eta(\cdot)} M^{-1\prime},
\]
where $M$ is nonsingular.
\end{theorem}

\begin{proof}
See Appendix~\ref{proof_hetshares_asympt} for a detailed proof.
\end{proof}

\vspace{0.3cm}

This result establishes asymptotic normality and provides a variance formula for any moment estimator that satisfies the required regularity conditions. In particular, when applied to the Lee-IPW bounds estimator, it yields the asymptotic distribution and variance expressions for the lower and upper bounds.

% ----------------------------------------------------------------------
\subsection{Variance Decomposition and Estimation under Heterogeneous-Shares Stratification}
\label{sec:var_het}
% ----------------------------------------------------------------------

The asymptotic covariance matrix $V_{\eta(\cdot)}$ from Theorem~\ref{thm:clt_het_gmm} extends the equal-share variance structure to handle arbitrary, stratum-specific treatment shares. This formulation is valid for any estimator whose moment vector satisfies the regularity conditions outlined earlier, including the Lee-IPW bounds as our example.

To establish consistency of the variance estimator under heterogeneous shares, we impose two additional
regularity conditions that are directly analogous to Assumptions~3.4--3.5 in \citet{bai2023efficiency}.
The first controls the construction of the involution $\pi(\cdot)$ used to handle singleton within-arm
cross-products by requiring paired blocks to be asymptotically close in covariates. The second provides
local regularity (integrability, empirical process, and smoothness) conditions ensuring the stability of
the variance components that involve conditional means and second moments.

\begin{assumption}[Paired-block covariate closeness]\label{ass:CB-neighbor}
Let $\Lambda(X^{(n)})=\{\lambda_g:g=1,\dots,G\}$ be the covariate-based blocks from Section~\ref{sec:design_equal_share_sec},
and let $\pi:\{1,\dots,G\}\to\{1,\dots,G\}$ be a fixed involution (a pairing map) such that
$\pi(\pi(g))=g$ and $\pi(g)\neq g$ for all $g$.
Assume that the paired blocks are asymptotically close in baseline covariates:
\[
\frac{1}{G}\sum_{g=1}^{G}\ 
\max_{i\in\lambda_g,\ i'\in\lambda_{\pi(g)}} \|X_i-X_{i'}\|^2
\ \xrightarrow{p}\ 0.
\]
\end{assumption}

\noindent\textit{Remark.}
Under Assumption~\ref{ass:B} (covariate balancing), it is typically possible to re-order the blocks and then
construct $\pi(\cdot)$ so that $(g,\pi(g))$ pairs blocks with similar covariates.
For example, one may apply a nonbipartite matching procedure to the block means
$\{\bar X_g\}_{g=1}^G$, where $\bar X_g := N_g^{-1}\sum_{i\in\lambda_g} X_i$,
to obtain a pairing $\pi(\cdot)$ such that Assumption~\ref{ass:CB-neighbor} holds.

\begin{assumption}[Local regularity for variance components]
\label{ass:CB-var-regularity}
There exists $\delta>0$ such that, for each $d\in\{0,1\}$:
\begin{enumerate}[label=(\alph*)]

\item \textbf{(Local uniform integrability).}
\begin{align*}
\lim_{\lambda\to\infty}\ 
\E\Bigg[
&\sup_{\theta\in\Theta:\ \|\theta-\theta_0\|<\delta}
\big\|m\!\left(X_i,d,Y_i(d);\theta\right)\big\|^2 \\
&\qquad\times
\mathbf 1\Bigg\{
\sup_{\theta\in\Theta:\ \|\theta-\theta_0\|<\delta}
\big\|m\!\left(X_i,d,Y_i(d);\theta\right)\big\|
>\lambda
\Bigg\}
\Bigg]
=0.
\end{align*}

\item \textbf{(Local Glivenko--Cantelli for conditional means and second moments).}
For each $1\le s\le d_\theta$, the classes
\begin{align*}
\Big\{
&\E\!\big[m_s\!\left(X_i,d,Y_i(d);\theta\right)\mid X_i=x\big]
:\ \|\theta-\theta_0\|<\delta
\Big\}
\end{align*}
and
\begin{align*}
\Big\{
&\E\!\big[
m_s\!\left(X_i,d,Y_i(d);\theta\right)\,
m\!\left(X_i,d,Y_i(d);\theta\right)^\top
\mid X_i=x\big]
:\ \|\theta-\theta_0\|<\delta
\Big\}
\end{align*}
are $Q$-Glivenko--Cantelli.

\item \textbf{(Local uniform Lipschitzness in covariates).}
Each component of
$\E[m(X,d,Y(d);\theta)\mid X=x]$ and
$\E[m(X,d,Y(d);\theta)m(X,d,Y(d);\theta)^\top\mid X=x]$
is Lipschitz in $x$ with a constant that is uniform over
$\{\theta\in\Theta:\ \|\theta-\theta_0\|<\delta\}$; that is, there exists $L_d<\infty$ such that
for all $x,x'\in\mathbb R^{d_x}$,
\begin{align*}
\sup_{\theta\in\Theta:\ \|\theta-\theta_0\|<\delta}
\Big|
&\E\!\big[m_s(X,d,Y(d);\theta)\mid X=x\big]
-
\E\!\big[m_s(X,d,Y(d);\theta)\mid X=x'\big]
\Big| \\
&\le L_d\|x-x'\|,
\end{align*}
and similarly for each component of
$\E[m(X,d,Y(d);\theta)m(X,d,Y(d);\theta)^\top\mid X=x]$.

\end{enumerate}
\end{assumption}

An important contribution of this section is a decomposition of the “meat” of the sandwich variance in terms of unconditional moments, extending beyond the conditional variance formulation\footnote{We propose a novel decomposition for the heterogeneous-shares variance by rewriting $\Omega_{\eta(\cdot)}$ as in \eqref{eq:V-eta-decomp}. This differs from the variance decomposition for the equal-share case in \citet{bai2023efficiency} and yields a slightly different plug-in implementation. \citet[Remark~4.4]{bai2023efficiency} discusses heterogeneous assignment probabilities and states the asymptotic variance in its original form with conditional variances, whereas our decomposition removes these conditional terms, allowing for tractable and consistent plug-in variance estimators.}. This expression is derived by applying the definition of conditional variance and properties of expectations to the original asymptotic variance formula:
\begin{align}
V_{\eta(\cdot)}
&= M^{-1}\,\Omega_{\eta(\cdot)}\,(M^{-1})^{\top}, \label{eq:V-eta-decomp}\\[0.25em]
\Omega_{\eta(\cdot)}
&:= \E\!\big[\eta_g\, m_1 m_1^{\top}\big] + \E\!\big[(1-\eta_g)\, m_0 m_0^{\top}\big] \nonumber\\
&\quad - \E\!\big[\eta_g(1-\eta_g)\, \mu_1 \mu_1^{\top}\big]
      - \E\!\big[\eta_g(1-\eta_g)\, \mu_0 \mu_0^{\top}\big] \nonumber\\
&\quad + \E\!\big[\eta_g(1-\eta_g)\, \mu_1 \mu_0^{\top}\big]
      + \E\!\big[\eta_g(1-\eta_g)\, \mu_0 \mu_1^{\top}\big] \nonumber\\
&\quad - \E\!\big[\eta_g \mu_1 + (1-\eta_g)\mu_0\big]\,
        \E\!\big[\eta_g \mu_1 + (1-\eta_g)\mu_0\big]^{\top}. \nonumber
\end{align}
This unconditional decomposition is particularly important in practice, since directly estimating conditional variances within small strata can be highly variable or even infeasible when some strata contain very few treated or control units. By expressing the variance entirely in terms of unconditional means and cross-products, this approach enables reliable and consistent variance estimation even in finely stratified designs with heterogeneous treatment shares.

To facilitate practical estimation, the variance formula can be expressed as a collection of sample averages and cross-products that are straightforward to compute. Specifically, let $\hat m_i = m(Z_i, \hat\theta_n)$ denote the estimated moment for each unit. For each stratum $g$, let $N_{d,g}$ represent the number of units assigned to arm $d$, $w_g = N_g/n$ the weight of stratum $g$ in the sample, and $\eta_g = T_g/N_g$ the observed treatment share. For empirical implementation, the variance components can be written as functions of sample averages and cross-products computed from the observed data. The primary unit-level quantities are
\begin{align*}
  \hat{A}_{1,n} &= \frac{1}{n} \sum_{i=1}^{n} D_i\, \hat m_i \hat m_i^{\top}, \quad
  \hat{A}_{0,n} = \frac{1}{n} \sum_{i=1}^{n} (1-D_i)\, \hat m_i \hat m_i^{\top}, \quad
  \hat{A}_{3,n} = \left( \frac{1}{n} \sum_{i=1}^{n} \hat m_i \right)\left( \frac{1}{n} \sum_{i=1}^{n} \hat m_i \right)^{\top}.
\end{align*}

In addition to these unit-level terms, cross-products within each block are incorporated to reduce noise and stabilize variance estimation by pooling information across pairs, which is especially important in finely stratified designs with small or uneven strata. These are defined as
\begin{align*}
\zeta_n(1,0)
&:= \sum_{g=1}^{G} w_g\,\eta_g(1-\eta_g)\,\frac{1}{2}\Bigg[
\Bigg(
  \frac{1}{N_{1,g}} \sum_{i\in\lambda_g:\,D_i=1}\hat m_i
\Bigg)
\Bigg(
  \frac{1}{N_{0,g}} \sum_{j\in\lambda_g:\,D_j=0}\hat m_j
\Bigg)^{\!\top}
\nonumber\\
&\hspace{3.1em}
+
\Bigg(
  \frac{1}{N_{0,g}} \sum_{j\in\lambda_g:\,D_j=0}\hat m_j
\Bigg)
\Bigg(
  \frac{1}{N_{1,g}} \sum_{i\in\lambda_g:\,D_i=1}\hat m_i
\Bigg)^{\!\top}
\Bigg],
\\[0.35em]
\zeta_n(1,1)
&:= \sum_{g=1}^{G} w_g\,\eta_g(1-\eta_g)\,\widehat\varsigma_{g,n}(1,1),
\\
\zeta_n(0,0)
&:= \sum_{g=1}^{G} w_g\,\eta_g(1-\eta_g)\,\widehat\varsigma_{g,n}(0,0),
\\[0.35em]
\hat{B}_n
&:= -\bigl[\, \zeta_n(1,1) + \zeta_n(0,0) - 2\,\zeta_n(1,0) \,\bigr].
\label{eq:Bn-def}
\end{align*}

\noindent Here, for $d\in\{0,1\}$, the within-arm cross-product component $\widehat\varsigma_{g,n}(d,d)$ is defined as
\[
\widehat\varsigma_{g,n}(d,d)
:=
\begin{cases}
\displaystyle
\frac{2}{N_{d,g}(N_{d,g}-1)}
\sum_{\substack{i<i'\\ i,i'\in\lambda_g:\,D_i=D_{i'}=d}}
\hat m_i\,\hat m_{i'}^{\top} ,
& \text{if } N_{d,g}\ge 2,\\[1.1em]
\displaystyle
\hat m_{i_g(d)}\,\hat m_{i_{\pi(g)}(d)}^{\top},
& \text{if } N_{d,g}=1,
\end{cases}
\]
where $i_g(d)$ denotes the (unique) index in block $g$ such that $D_{i_g(d)}=d$ when $N_{d,g}=1$.
In the case $N_{d,g}=1$, the between-block product uses a paired block $\pi(g)$, where $\pi:\{1,\dots,G\}\to\{1,\dots,G\}$ is a fixed involution with $\pi(\pi(g))=g$ and $\pi(g)\neq g$ for all $g$.\footnote{After re-ordering blocks, one may construct $\pi(\cdot)$ so that paired blocks $(g,\pi(g))$ are asymptotically close in covariates (cf.\ Assumption~3.4 of \citet{bai2024inference}).}

With these components, the estimated variance “meat” is assembled as
\[
  \widehat\Omega_n = \hat{A}_{1,n} + \hat{A}_{0,n} + \hat{B}_n - \hat{A}_{3,n},
\]
and the full estimated covariance matrix is
\[
  \widehat V_{\eta(\cdot)} = \hat M_n^{-1} \widehat\Omega_n\, \hat M_n^{-T}, \qquad
  \hat M_n = \frac{1}{n} \sum_{i=1}^n \partial_\theta m(Z_i, \hat\theta_n).
\]
For the Lee-IPW lower and upper bounds, the standard error is given by the variance of a difference, that is,
\[
  \hat\sigma_{\bullet}^2 = \widehat V_{\eta(\cdot),11}^{\,\bullet} + \widehat V_{\eta(\cdot),22}^{\,\bullet} - 2\,\widehat V_{\eta(\cdot),12}^{\,\bullet},
\]
where $\widehat V_{\eta(\cdot),ij}^{\,\bullet}$ denotes the $(i,j)$ entry of the estimated covariance matrix, constructed using the appropriate moment conditions for the lower or upper bound.

\begin{theorem}[Consistency of $\widehat V_{\eta(\cdot)}$ under heterogeneous-shares stratification]
\label{prop:var_het}
Suppose Assumptions~\ref{ass:iid}, \ref{ass:B}, \ref{ass:Ahet}, \ref{ass:C}, \ref{ass:CB-neighbor}, and
\ref{ass:CB-var-regularity} hold, and $\hat M_n \xrightarrow{p} M$. Let $\widehat\Omega_n$ and
$\widehat V_{\eta(\cdot)}$ be defined as in this section. Then,
\[
\widehat V_{\eta(\cdot)} \xrightarrow{p} V_{\eta(\cdot)}:=M^{-1}\Omega_{\eta(\cdot)}M^{-T}.
\]
Consequently, $n^{-1} \hat\sigma_{\bullet}^2$ delivers consistent variances for $\hat\Delta^{LB}$ and $\hat\Delta^{UB}$, using the appropriate moment vectors.
\end{theorem}

\begin{proof}
See Appendix~\ref{proof_hetshares_var} for a detailed proof.
\end{proof}

\vspace{0.3cm}

The estimator relies on sample averages and block-level cross-products, making it broadly applicable and efficient in computation. The correction term $\hat{B}_n$ is particularly important when treatment shares are highly variable or block sizes are small, as it adjusts for the negative covariance introduced by fixing treated counts within strata. Neglecting this feature can lead to overconservative inference. The variance decomposition and estimation methods presented here apply not only to Lee-IPW bounds, but more generally to any estimator meeting the required regularity conditions in stratified experimental designs with heterogeneous shares.

% ----------------------------------------------------------------------
\subsection{Empirical Implementation of Standard Errors for Lee-IPW Bounds}
\label{sec:boot_het}
% ----------------------------------------------------------------------

The analytic sandwich variance estimator developed in Section~\ref{sec:var_het} is recommended as the primary method for inference, as it is consistent under a wide range of stratification and blocking schemes, and directly incorporates the negative covariance that arises when the number of treated units per stratum is fixed by design. This feature distinguishes it from classical i.i.d.\ and equal-share variance estimators, which do not fully capture the dependence structure present in complex experimental designs with heterogeneous treatment shares. In empirical applications, the sandwich estimator is generally preferable for constructing confidence intervals and conducting hypothesis tests, especially when treatment shares or block sizes are unbalanced.

In all cases, Lee-IPW bounds should be accompanied by confidence intervals constructed using the recommended procedures \citep{lee2009training} and the standard error estimator presented here. When the confidence limits agree in sign, conclusions regarding the sign of the treatment effect are robust to monotone selection. For designs with highly variable treatment shares or small blocks, special attention should be paid to the correction term in the variance formula, as neglecting this adjustment may result in overly conservative inference.

Overall, the methods presented in this section extend an existing inference framework to accommodate arbitrary heterogeneity in treatment assignment and block structure. They provide a coherent and flexible approach for valid uncertainty quantification, applicable not only to Lee-IPW bounds but to a broad class of moment estimators in stratified experiments.

\section{Conclusion}\label{sec:conclusion}

This paper develops a unified, design-consistent approach for valid inference on Lee bounds in randomized experiments subject to differential outcome attrition and stratified randomization. We establish that, under equal-share stratification, the variance reduction induced by fixing the treatment allocation within blocks can be fully captured using recent theory for finely stratified experiments. Our results clarify that conventional variance estimators based on i.i.d.\ assumptions systematically overstate uncertainty when block-level assignment is present, leading to unnecessarily conservative inference. By instead using design-consistent variance estimators that incorporate the block structure and fixed treatment shares, practitioners can achieve tighter, more informative confidence intervals for Lee bounds.

Among its contributions, our work introduces a formal innovation by extending the scope of valid inference beyond covariate-based stratification to include designs where stratification is determined solely by discrete block labels, such as schools or survey waves, a setting not previously explicitly accommodated in the literature \citep{bai2023efficiency}. We show that the assumptions required for sharp identification and consistency of Lee bounds are preserved under equal-share stratification, and that the efficiency gains from block randomization are readily realized under both covariate-based and label-based stratification settings.

When treatment shares vary across blocks, conventional Lee bounds based on unconditional trimming are no longer valid, and conditional bounds constructed within small or sparsely treated strata can be highly variable, biased, or undefined. To address these challenges, we introduce the Lee–IPW bounds, which employ inverse probability weighting and a global trimming rule to restore validity and robustness in heterogeneous-shares stratified designs. For these bounds, we provide a closed-form, design-consistent variance estimator that ensures reliable inference even with small or unbalanced blocks.

Beyond this specific partial identification setting, we extend existing inference results to accommodate a broad class of moment-based estimators under heterogeneous-shares stratified designs. This extension builds on the literature on finely stratified experiments, enabling valid and efficient inference in a wider range of complex experimental settings.

Practically, our recommendations are as follows. Researchers should always retain and account for the full block structure of their experimental design, record actual treatment shares within each stratum, and use the appropriate design-consistent variance estimator. Specifically, $\widehat V_{\eta}$ should be used for equal shares, and $\widehat V_{\eta(\cdot)}$ for heterogeneous shares. When strata are small or sparsely treated, the Lee–IPW bounds provide a robust and well-defined alternative to conditional Lee bounds, as they avoid the instability and potential bias that can arise from such procedures in these settings. In all cases, relying on standard i.i.d.\ variance formulas ignores negative assignment dependence and typically yields overly conservative confidence intervals.

In summary, this work provides a general, practical, and efficient set of tools for valid inference on both conventional Lee bounds and the newly introduced Lee–IPW bounds in the presence of outcome attrition under a wide range of stratified experimental designs. Beyond these partial identification settings, our results further extend to inference for a broad class of moment-based estimators, enabling design-consistent variance estimation even in experiments with complex or heterogeneous-shares stratification. By connecting partial identification with more contemporaneous inference theory and robust variance estimation, these advances help ensure that treatment effect analyses remain both credible and informative across diverse empirical research contexts.

\newpage

\appendix

\begin{center}
	\huge
	Appendix
\end{center}

\doublespacing
\normalsize

\section{Proof of Proposition \ref{prop:lee_consistency_equalshare}}\label{app:proof_consistency_equalshare}
We follow \citet[Thm.~2.6]{newey1994large} and \citet[Prop.~2]{lee2009training} closely, except that the i.i.d.\ ULLN used there is replaced by a ULLN valid under equal–share stratification. Under Assumptions~\ref{ass:iid} and \ref{ass:A}, sampling within each block is simple random sampling without replacement, which induces negative within-block dependence. \citet[][Lem.~2.1 and Rem.~2.1]{hajek1960limiting} shows that, for block averages, without-replacement sampling behaves in large samples as independent sampling from the same finite-population distribution. Hence, block averages of $g(Z_i,\theta)$ satisfy the usual pointwise laws of large numbers. To obtain uniform convergence over a totally bounded parameter space, combine pointwise convergence with stochastic equicontinuity as in \citet[Thm.~1]{andrews1992generic}. The generic variants under a Lipschitz condition and under termwise stochastic equicontinuity also apply \citep[Thm.~3(a), Thm.~5]{andrews1992generic}, and the triangular-array extension at the \citep[][end of §3.4]{andrews1992generic} covers growing numbers of blocks with varying sizes. Together, these results yield the required ULLN under equal-share stratification.

With the blocked-design ULLN established, the only place where \citet{newey1994large} invoked i.i.d.\ sampling, the uniform convergence of sample moments, is now covered under equal-share stratification. From here, we can follow \citet[Prop.~2, Proof]{lee2009training} verbatim. The sample objective converges uniformly to its population analogue, the selection rates and the induced trimming share are consistently estimated so $\hat q\to_p q$, the cutoffs are regular so the trimming and quantile maps are continuous, and the trimmed moments inherit uniform convergence. Applying \citet[Thm.~2.6]{newey1994large} then yields consistency of the unconditional Lee bounds under equal-share stratification. We do not reproduce \citet[Prop.~2, Proof]{lee2009training} in full, since our only modification is to supply the ULLN under the blocked design.

\newpage

\section{Simulation Designs and Monte Carlo Behavior}
\label{app:sim_dgpss}

\subsection{Simulation 1: Equal-Share Matched Pairs and Design-Consistent Variance}
\label{app:sim1_dgp}

This simulation evaluates the finite-sample performance of the design-consistent variance estimator for Lee bounds in an equal-share matched-pair design and compares it to the conventional i.i.d.\ variance estimator that ignores blocking.

\paragraph{Stratification and treatment assignment.}
Each replication draws a sample of size $n = 10000$. For each unit $i$ we generate a covariate $X_i \sim \mathcal{N}(0,1)$ and an independent noise term $\varepsilon_i \sim \mathcal{N}(0,1)$, and set
\[
Y_i \;=\; 2X_i + 2 + \varepsilon_i.
\]
We sort observations by $X_i$ and form matched pairs by grouping consecutive units, so that each pair contains two units with similar covariates. Within each pair, one unit is randomly assigned to treatment and the other to control with equal probability. This yields an equal-share matched-pair design with $1{:}1$ allocation that satisfies stratum-level independence.

\paragraph{Selection mechanism.}
Attrition is introduced in a way that makes treated units more likely to be observed than controls. Let $S_i(d)$ denote the potential selection indicator under treatment status $d \in \{0,1\}$. For each $i$,
\[
\Pr\big(S_i(1) = 1\big) = 0.8, \qquad \Pr\big(S_i(0) = 1\big) = 0.7,
\]
independently across units. Realized selection is $S_i = S_i(D_i)$, so selection favors treated units in a manner consistent with the monotone selection pattern underlying Lee bounds. For treated units with $S_i = 1$ we add an independent shift $U_i \sim \text{Unif}(0,2)$ to the observed outcome to introduce additional treated-side variation, reflecting the empirical motivation for Lee-type trimming.

\paragraph{Observed data and target.}
The observed outcome is
\[
Y_i^{\text{obs}} \;=\;
\begin{cases}
Y_i + U_i, & \text{if } D_i = 1 \text{ and } S_i = 1,\\
Y_i,       & \text{if } D_i = 0 \text{ and } S_i = 1,\\
\text{missing}, & \text{if } S_i = 0,
\end{cases}
\]
and the analyst observes $(Y_i^{\text{obs}}, D_i, S_i, g_i)$, where $g_i$ indexes the matched pair. The Lee lower bound is computed using the matched-pair structure, and the simulation compares (i) the design-consistent variance estimator that incorporates this blocked assignment to (ii) the conventional i.i.d.\ variance estimator that ignores blocking. As reported in Section~\ref{sec:lim_cond}, the design-consistent estimator closely tracks the Monte Carlo variability of the Lee lower bound, while the i.i.d.\ estimator is substantially conservative.

\paragraph{Monte Carlo behavior.}
In this design, the design-consistent variance estimator closely tracks the Monte Carlo variability of the Lee lower bound, while the conventional i.i.d.\ variance estimator is substantially conservative, with average standard errors well above the empirical standard deviation. These patterns reproduce the main results reported in Section~\ref{sec:lim_cond}.

\subsection{Simulation 2: Heavy-Tailed Outcomes}
\label{app:sim2_dgp}

This simulation is designed to study the performance of conditional Lee bounds and pooled Lee-IPW bounds in a finely stratified design with many small strata and heavy-tailed outcomes, while preserving the standard identification assumptions. The design isolates an outlier sensitivity problem for cell-by-cell trimming procedures: within each stratum, selection asymmetries are concentrated on the single largest untreated potential outcome.

\paragraph{Stratification and treatment assignment.}
Each replication draws a sample of size $n = 2000$. For each unit $i$ we draw a covariate $X_i \sim \mathcal{N}(0,1)$ and sort observations by $X_i$. We partition the ordered sample into $G = 100$ equal-sized strata of size $20$, so strata correspond to covariate blocks. Within each stratum, treatment is assigned uniformly over all allocations with exactly $10$ treated units and $10$ control units, i.e.\ $\eta_g = 0.5$ for all $g$. This is a stratified, equal-share assignment mechanism.

\paragraph{Potential outcomes.}
Independently of $X_i$, we generate a heavy-tailed shock $V_i$ from a Pareto distribution with tail index $\alpha = 2.2$, truncated at its $99.5$th percentile. The untreated potential outcome is
\[
Y_i^*(0) \;=\; 2X_i + 2 + 12V_i,
\]
which has a pronounced upper tail. The treatment effect is constant and equal to one:
\[
Y_i^*(1) \;=\; Y_i^*(0) + 1.
\]

\paragraph{Selection mechanism.}
Selection is monotone in treatment and asymmetric in the upper tail. Within each stratum $g$, we flag as an outlier the single unit with the largest value of $Y_i^*(0)$. Potential selection outcomes $S_i(d)$ are generated using treatment- and outlier-specific selection probabilities:
\begin{align*}
&\Pr\big(S_i(d)=1 \,\big|\, i \text{ is an outlier}\big) \;=\;
\begin{cases}
1.00, & \text{if } d=1,\\
0.01, & \text{if } d=0,
\end{cases}
\\
&\Pr\big(S_i(d)=1 \,\big|\, i \text{ is not an outlier}\big) \;=\;
\begin{cases}
0.98, & \text{if } d=1,\\
0.94, & \text{if } d=0.
\end{cases}
\end{align*}
This construction implies $S_i(1)\ge S_i(0)$ almost surely for every $i$, so
Assumption~\ref{ass:MONO} holds for the DGP. Realized selection is $S_i=S_i(D_i)$. In practice, we apply a minimal within-stratum adjustment to avoid degenerate trimming cells for conditional Lee bounds. In each stratum, we ensure that the treated and control arms each contain at least one selected and at least one unselected observation. We also enforce at least three selected treated observations and at least two selected control observations per stratum. These adjustments are only meant to avoid pathological cells and do not affect the qualitative behavior of the procedures.

\paragraph{Observed data.}
The observed outcome is
\[
Y_i \;=\;
\begin{cases}
Y_i^*(D_i), & \text{if } S_i=1,\\
\text{missing}, & \text{if } S_i=0,
\end{cases}
\]
and the analyst observes $(Y_i, D_i, S_i, g_i)$, where $g_i$ indexes the stratum. By construction, the data satisfy equal-share stratified randomization, monotone selection, and a constant treatment effect of one, but feature heavy upper tails and selection asymmetries concentrated on within-stratum outliers. These features make the cell-specific trimming shares underlying conditional Lee bounds highly sensitive to local tail behavior and fractional trimming, which underlies the distortions documented in Section~\ref{sec:lim_cond}.

\paragraph{Monte Carlo behavior.}
In each replication, we compute conditional Lee bounds and pooled Lee-IPW bounds. Conditional Lee bounds are obtained by trimming within each stratum and then averaging the stratum-specific intervals. Identification is correct by construction, since stratified randomization, monotone selection, and a constant treatment effect all hold. The core difficulty is estimation rather than identification: each stratum contains only a small number of treated and control units and serves as the basis for a separate nonparametric trimming rule.

The key object is the cell-level trimming share
\[
\tau_g \;=\; 1 - \frac{\Pr(S=1 \mid D=0,g)}{\Pr(S=1 \mid D=1,g)},
\]
which must be estimated from a small number of observations in stratum $g$. With many small strata and heavy right tails, the estimated share $\widehat{\tau}_g$ is sensitive to single observations. In many replications, $\widehat{\tau}_g$ does not align with an integer fraction of treated units, so the conditional procedure uses fractional trimming. The boundary treated observation at the cutoff is kept with a fractional weight between zero and one so that, in expectation, exactly a $\widehat{\tau}_g$ share of treated units is trimmed within the cell.

Because outcomes have a heavy upper tail and selection is particularly favorable to large treated outcomes, especially for within-stratum outliers, the boundary unit can be an extreme treated value. Small changes in $\widehat{\tau}_g$ or in the ordering of outcomes can then translate into large movements in the trimmed treated mean. In repeated samples, the conditional lower bound can exceed the true effect of one by a large margin, and the aggregated conditional interval can fail to contain the true effect even though monotonicity and a constant treatment effect hold.

By contrast, the pooled Lee-IPW procedure applies a single global trimming share tied to the aggregate selection pattern and does not require separate trimming rules in each cell. As a result, it is much less sensitive to local tail behavior. In the simulations, the pooled Lee-IPW lower bound remains safely below one and coverage stays close to the nominal level. These experiments therefore highlight that, with many small strata and heavy tails, per-stratum conditional Lee bounds can behave poorly due to outlier-driven instability, while the pooled Lee-IPW procedure remains stable and informative.

\subsection{Further Limitations of Conditional Lee Bounds}
Further complications arise for conditional Lee bounds in edge cases. If a cell contains no observed treated units or no controls, the trimmed mean is not well defined, and it becomes impossible to construct a valid contribution from that stratum. If $\widehat{\tau}_g$ is very close to one, trimming removes nearly all treated observations in that cell, so the resulting bound is effectively uninformative. The proliferation of very small blocks can also produce many noisy micro-bounds, so their combined imprecision may overwhelm the information contributed by larger blocks. As a result, the overall confidence interval can be wide and poorly centered, offering little practical value for inference.

\newpage

\section{Identification of Bounds on Treatment Effects with Lee's Approach and IPW}\label{sec:identification}

We derive bounds on the average treatment effect for the always–observed sub-population when treatment probabilities differ across strata \(g=1,\dots,G\).  The target parameter is  
\[
\Delta_{ATE}\;\equiv\;
\E\!\Bigl[
      Y^{\ast}_{g,i}(1)-Y^{\ast}_{g,i}(0)
      \,\Big|\,
      S_{g,i}(1)=1,\,
      S_{g,i}(0)=1
   \Bigr].
\]

%----------------------------------------------------------------------
\subsection{Identification of 
\texorpdfstring{$\E\!\bigl[Y^{\ast}_{g,i}(0)\,\big|\,S_{g,i}(1)=1,S_{g,i}(0)=1\bigr]$}{E[Y*(0)|always observed]}}
\label{sec:iden_Y0}
%----------------------------------------------------------------------

Under conditional independence (CIA) and monotonicity, the control mean for the always–observed group is point-identified. Note that
\begin{align}
\E\!\bigl[S_{g,i}(0)\,Y^{\ast}_{g,i}(0)\bigr]
   &\overset{(\text{LIE})}{=}
     \Pr\!\bigl[S_{g,i}(0)=1\bigr]\;
     \E\!\bigl[Y^{\ast}_{g,i}(0)\mid S_{g,i}(0)=1\bigr]
     \notag\\
   &\overset{(\text{IPW + Monot.})}{=}
     \E\!\Bigl[
        S_{g,i}\,
        \Bigl\{\tfrac{1-p}{1-\eta_g}\Bigr\}
        \,\Big|\,D_{g,i}=0
     \Bigr]\;
     \E\!\Bigl[
        Y^{\ast}_{g,i}(0)
        \,\Big|\,S_{g,i}(1)=1,\,S_{g,i}(0)=1
     \Bigr],
   \label{eq:id_Y0_left_g}
\end{align}
where \(p=\Pr(D_{g,i}=1)\) and \(\eta_g=T_g/N_g\).  Also,
\begin{align}
\E\!\bigl[S_{g,i}(0)\,Y^{\ast}_{g,i}(0)\bigr]
   &\overset{(\text{LIE + IPW})}{=}
     \E\!\Bigl[
        S_{g,i}\,Y_{g,i}\,
        \Bigl\{\tfrac{1-p}{1-\eta_g}\Bigr\}
        \,\Big|\,D_{g,i}=0
     \Bigr]
     \notag\\
   &\overset{(\text{LIE + Monot.})}{=}
     \Pr\!\bigl[S_{g,i}=1\mid D_{g,i}=0\bigr]\cdot
     \notag\\ &\qquad \qquad \qquad \qquad \qquad\E\!\Bigl[
        Y_{g,i}\,
        \Bigl\{\tfrac{1-p}{1-\eta_g}\Bigr\}
        \,\Big|\,D_{g,i}=0,\,S_{g,i}(1)=1,\,S_{g,i}(0)=1
     \Bigr].
   \label{eq:id_Y0_right_g}
\end{align}
Equating \eqref{eq:id_Y0_left_g} and \eqref{eq:id_Y0_right_g} gives
\begin{equation}\label{eq:iden_Y0}
\begin{alignedat}{2}
\E\bigl[Y^*_{g,i}(0)\mid \,S_{g,i}(1)=1,\,S_{g,i}(0)=1\bigr]
&=
\frac{%
  \Pr\bigl(S_{g,i}=1\mid D_{g,i}=0\bigr)\;
  \E\Bigl[
    Y_{g,i}\{\tfrac{1-p}{1-\eta_g}\}
    \;\Big|\;
    \substack{D_{g,i}=0,\\ S_{g,i}(1)=1, \\ S_{g,i}(0)=1}
  \Bigr]
}{%
  \E\bigl[S_{g,i}\{\tfrac{1-p}{1-\eta_g}\}\mid D_{g,i}=0\bigr]
}\\
&\overset{\mathrm{(Monot.)}}{=}
\frac{%
  \Pr\bigl(S_{g,i}=1\mid D_{g,i}=0\bigr)\;
  \E\Bigl[
    Y_{g,i}\{\tfrac{1-p}{1-\eta_g}\}
    \;\Big|\;
    \substack{D_{g,i}=0,\\ S_{g,i}=1}
  \Bigr]
}{%
  \E\bigl[S_{g,i}\{\tfrac{1-p}{1-\eta_g}\}\mid D_{g,i}=0\bigr]
}\,.
\end{alignedat}
\end{equation}
where the numerator and denominator involve only observable variables and the known inverse probability weight \(w_{c,g}=(1-p)/(1-\eta_g)\).

%----------------------------------------------------------------------
 \subsection{Identification of 
\texorpdfstring{$\E\!\bigl[Y^{\ast}_{g,i}(1)\,\big|\,S_{g,i}(1)=1,S_{g,i}(0)=1\bigr]$}{E[Y*(1)|always observed]}}
\label{sec:iden_Y1}
%----------------------------------------------------------------------

Identifying  
\(\E\!\bigl[Y^{\ast}_{g,i}(1)\,\big|\,S_{g,i}(1)=1,S_{g,i}(0)=1\bigr]\)  
is harder than for the control arm because the set \(\{D_{g,i}=1,S_{g,i}=1\}\) mixes (i) always–observed units and (ii) units observed only when treated.  To isolate the always–observed group among treated individuals, we exploit the following reweighting. Define
\[
h_{g,i}
   \;\equiv\;
   \frac{
     \Pr\!\bigl[D_{g,i}=1\mid S_{g,i}(1)=1,S_{g,i}(0)=1\bigr]}
     {\eta_g},
\]
and define the re-weighted outcome for treated observations 
\[
\widetilde Y_{g,i}
   \;\equiv\;
   h_{g,i}\,Y_{g,i}.
\]

Notice that \(h_{g,i}\) is identified since
\begin{align}
\Pr\!\bigl[D_{g,i}=1\mid S_{g,i}(1)=1,S_{g,i}(0)=1\bigr]
   &\!\!\overset{\text{(Bayes' Rule)}}{=}
   \frac{
     \Pr\!\bigl[S_{g,i}(1)=1,S_{g,i}(0)=1\mid D_{g,i}=1\bigr]\;p}
     {\Pr\!\bigl[S_{g,i}(1)=1,S_{g,i}(0)=1\bigr]}
   \notag\\
   &\!\!\overset{\text{(LIE)}}{=}
   \frac{
     \E\!\Bigl[
        \E\!\bigl[S_{g,i}(1)S_{g,i}(0)\mid X_{g,i}, D_{g,i}=1\bigr]
        \Bigm|D_{g,i}=1
     \Bigr]\;p}
     {\E\!\Bigl[
        \E\!\bigl[S_{g,i}(1)S_{g,i}(0)\mid X_{g,i}\bigr]
     \Bigr]}
   \notag\\
   &\!\!\overset{\text{(CIA \& Monotonicity)}}{=}
   \frac{
     \E\!\Bigl[
        \E\!\bigl[S_{g,i}\mid X_{g,i},D_{g,i}=0\bigr]
        \Bigm|D_{g,i}=1
     \Bigr]\;p}
     {\E\!\Bigl[
        \E\!\bigl[S_{g,i}\mid X_{g,i},D_{g,i}=0\bigr]
     \Bigr]},
   \label{eq:h_ident}
\end{align}
so \(h_{g,i}\) is fully determined by observable quantities.

Now, for any \(\tilde{y}\in\mathbb{R}\), note that the conditional expectation of the indicator for \(\widetilde{Y}_{g,i}\) can be decomposed as
\begin{align}
\mathbb{E}\Bigl[\left. \mathbf{1} \left\{ \widetilde{Y}_{g,i} \leq \tilde{y} \right\} \right\vert D_{g,i} = 1,\, S_{g,i}=1\Bigr] 
&= \left(1 - q\right) \cdot \mathbb{E}\Bigl[\mathbf{1} \left\{ \widetilde{Y}_{g,i} \leq \tilde{y} \right\} \Big\vert D_{g,i} = 1,\, S_{g,i}(1)=1,\, S_{g,i}(0)=1\Bigr] \nonumber \\
& \quad +\, q \cdot \mathbb{E}\Bigl[\mathbf{1} \left\{ \widetilde{Y}_{g,i} \leq \tilde{y} \right\} \Big\vert D_{g,i} = 1,\, S_{g,i}(1)=1,\, S_{g,i}(0)=0\Bigr],
\label{Ytilde-decomposition}
\end{align}
and \(q\) is defined and identified by
\begin{align}
q &\equiv \Pr\Bigl[S_{g,i}(1)=1,\,S_{g,i}(0)=0 \,\Big\vert\, D_{g,i}=1,\,S_{g,i}=1\Bigr] \nonumber \\
  &= \frac{\Pr\Bigl[S_{g,i}=1\,\vert\,D_{g,i}=1\Bigr] - \Pr\Bigl[S_{g,i}(0)=1\,\vert\,D_{g,i}=1\Bigr]}{\Pr\Bigl[S_{g,i}=1\,\vert\,D_{g,i}=1\Bigr]} \nonumber \\
  & \text{(i.e., the fraction of observed treated individuals whose observation is induced by treatment)} \nonumber \\
  &= \frac{\Pr\Bigl[S_{g,i}=1\,\vert\,D_{g,i}=1\Bigr] - \mathbb{E}\Bigl[S_{g,i}(0)\,\big\vert\,D_{g,i}=1\Bigr]}{\Pr\Bigl[S_{g,i}=1\,\vert\,D_{g,i}=1\Bigr]} \nonumber \\
  &= \frac{\Pr\Bigl[S_{g,i}=1\,\vert\,D_{g,i}=1\Bigr] - \mathbb{E}\Bigl[\left. S_{g,i} \Bigl\{\frac{\eta_g(1-p)}{(1-\eta_g)\,p}\Bigr\} \right\vert\,D_{g,i}=0\Bigr]}{\Pr\Bigl[S_{g,i}=1\,\vert\,D_{g,i}=1\Bigr]},
\label{att_selection}
\end{align}
where the final equality follows from the IPW representation of \(\mathbb{E}\bigl[S_{g,i}(0)\,\vert\,D_{g,i}=1\bigr]\).

Let \(\tilde y_q\) and \(\tilde y_{1-q}\) be the \(q\)- and \((1-q)\)-quantiles of \(\widetilde Y_{g,i}\) within \(\{D_{g,i}=1,S_{g,i}=1\}\).  Then the Lee‐style bounds which partially identify $\E\!\bigl[
      \widetilde Y_{g,i}
      \mid D_{g,i}=1,S_{g,i}(1)=S_{g,i}(0)=1
   \bigr]$ are
\begin{align}
\E\!\bigl[
   \widetilde Y_{g,i}
   \mid D_{g,i}=1,S_{g,i}=1,\widetilde Y_{g,i}\le \tilde y_{1-q}
\bigr]
   &\;\le\;
   \E\!\bigl[
      \widetilde Y_{g,i}
      \mid D_{g,i}=1,S_{g,i}(1)=S_{g,i}(0)=1
   \bigr]
   \notag\\
   &\;\le\;
   \E\!\bigl[
      \widetilde Y_{g,i}
      \mid D_{g,i}=1,S_{g,i}=1,\widetilde Y_{g,i}\ge \tilde y_q
   \bigr].
   \label{eq:trim_bounds}
\end{align}

Next, to verify that this partially identified mean recovers the always–observed mean potential outcome, observe that
{\footnotesize
\begin{align}
\mathbb{E}\!\Bigl[\widetilde{Y}_{g,i}\,\Big\vert\,&D_{g,i}=1,S_{g,i}(1)=1,S_{g,i}(0)=1\Bigr] \notag\\
&\overset{(\mathrm{LIE})}{=} \frac{\mathbb{E}\!\Bigl[D_{g,i}\,\widetilde{Y}_{g,i}\,\Big\vert\, S_{g,i}(1)=1,S_{g,i}(0)=1\Bigr]}{\Pr\!\Bigl[D_{g,i}=1\,\Big\vert\,S_{g,i}(1)=1,S_{g,i}(0)=1\Bigr]} \nonumber\\[1mm]
&\overset{\text{(by definition)}}{=} \mathbb{E}\!\Biggl[\left.\frac{D_{g,i}\,Y_{g,i}}{\eta_g}\,\right\vert\, S_{g,i}(1)=1,S_{g,i}(0)=1\Biggr] \nonumber\\[1mm]
&\overset{(\mathrm{LIE})}{=} \mathbb{E}\left[\left. \mathbb{E}\!\Bigl[\frac{D_{g,i}\,Y_{g,i}}{\eta_g} \,\Big\vert\,X_{g,i},\,S_{g,i}(1)=1,S_{g,i}(0)=1\Bigr] \right\vert\,S_{g,i}(1)=1,S_{g,i}(0)=1\right] \nonumber\\[1mm]
&\overset{\text{(CIA)}}{=} \mathbb{E}\Biggl[\Bigl.\frac{\mathbb{E}\Bigl[D_{g,i}\,\big\vert\,X_{g,i},S_{g,i}(1)=1,S_{g,i}(0)=1\Bigr]}{\eta_g}\nonumber\\[1mm]
&\quad\quad\quad\times\mathbb{E}\Bigl[Y_{g,i}\,\big\vert\,X_{g,i},D_{g,i}=1,S_{g,i}(1)=1,S_{g,i}(0)=1\Bigr]\Bigr\vert S_{g,i}(1)=1,S_{g,i}(0)=1\Biggr]\nonumber\\[1mm]
&\overset{%
  \substack{%
    (\E[D_{g,i}\mid X_{g,i},\,S_{g,i}(1)=1,\,S_{g,i}(0)=1]=\eta_g\\
    \text{and CIA)}%
  }%
}{=}
\E\Bigl[\E\bigl[Y^*_{g,i}(1)\mid X_{g,i},\,S_{g,i}(1)=1,\,S_{g,i}(0)=1\bigr]\;\bigm|\;
        S_{g,i}(1)=1,\,S_{g,i}(0)=1\Bigr]\nonumber\\[1mm]
&\overset{(\mathrm{LIE})}{=} \mathbb{E}\Bigl[Y^*_{g,i}(1)\,\Big\vert\,S_{g,i}(1)=1,S_{g,i}(0)=1\Bigr]\,.
\label{ate-treated-1}
\end{align}
}

Therefore, the Lee-IPW bounds on \(\mathbb{E}\left[Y^*_{s,i}(1) \,\vert\,S_{s,i}(1)=1,S_{s,i}(0)=1\right]\) are identified by computing the trimmed mean of \(\widetilde{Y}_{s,i}\) in the treated group at the trimming proportion \(q\), while the control mean is point–identified via \eqref{eq:iden_Y0}. The difference between these two quantities yields the desired bound on \(\Delta_{ATE}\).

\newpage
\section{Proofs for Inference under Covariate-Based, Heterogeneous-Shares Stratification}
\label{app:pf_thm31_het}

This appendix provides the main technical arguments for the extension to covariate-based stratification with
heterogeneous treatment shares.

Some lemmas and theorems rely on two design conditions, a design LLN for treated/control weighted sums
(Assumption~\ref{ass:CB-LLN}) and a design CLT for the block-imbalance term (Assumption~\ref{ass:CB-IMB}), which are
analogous in purpose to Assumptions~4.1--4.2 in \citet{bai2023efficiency}, adapted to allow stratum-specific treatment
shares $\eta_g$.

\begin{assumption}[Design LLN for treated-weighted sums]\label{ass:CB-LLN}
For any integrable scalar function $\gamma:\mathbb R^{d_x}\to\mathbb R$,
\[
\frac{1}{n}\sum_{i=1}^n D_i\,\gamma(X_i)\ \xrightarrow{P}\ \E\!\big[\eta(X)\,\gamma(X)\big],
\qquad
\eta(x):=\Pr(D=1\mid X=x),
\]
and likewise $\frac{1}{n}\sum_{i=1}^n (1-D_i)\gamma(X_i)\xrightarrow{P}\E[(1-\eta(X))\gamma(X)]$.
Moreover, for any countable collection $\{\gamma_j\}_{j\ge1}$ of integrable scalar functions,
\[
\sup_{j\ge1}\left|
\frac{1}{n}\sum_{i=1}^n D_i\,\gamma_j(X_i)\;-\;\E\!\big[\eta(X)\,\gamma_j(X)\big]
\right|\ \xrightarrow{P}\ 0,
\]
and similarly,
\[
\sup_{j\ge1}\left|
\frac{1}{n}\sum_{i=1}^n (1-D_i)\,\gamma_j(X_i)\;-\;\E\!\big[(1-\eta(X))\,\gamma_j(X)\big]
\right|\ \xrightarrow{P}\ 0.
\]
\end{assumption}

\begin{assumption}[Design CLT for block-imbalance]\label{ass:CB-IMB}
Let $\rho$ be any metric that metrizes weak convergence. For any square-integrable vector function
$\gamma:\mathbb R^{d_x}\to\mathbb R^{k}$ with $\E[\gamma(X)]=0$,
\[
\rho\!\left(\ \frac{1}{\sqrt n}\sum_{i=1}^n (D_i-\eta_i)\,\gamma(X_i)\ ,\ \mathcal N\!\big(0,V_{\mathrm{imb},\gamma}\big)\ \middle|\ X^{(n)}\right)\ \xrightarrow{P}\ 0,
\]
where $\eta_i:=\Pr(D_i=1\mid X^{(n)})=\eta_g$ for $i\in\lambda_g$ (Assumption~\ref{ass:Ahet}), and
\begin{align*}
V_{\mathrm{imb},\gamma}
&:= \plim\ \frac{1}{n}\sum_{g=1}^G N_g\,\eta_g(1-\eta_g)\,S_{\gamma,g},\\[0.25em]
S_{\gamma,g}
&:= \frac{1}{N_g-1}\sum_{i\in\lambda_g}
\big(\gamma(X_i)-\bar\gamma_g\big)\big(\gamma(X_i)-\bar\gamma_g\big)^\prime,\\[0.25em]
\bar\gamma_g
&:= \frac{1}{N_g}\sum_{i\in\lambda_g}\gamma(X_i).
\end{align*}
\end{assumption}

\subsection{Consistency Proof in the Covariate-Based, Heterogeneous-Shares Stratification Setting}

\begin{lemma}[Consistency of $\hat\theta_n$]\label{lem:A8-het-cov}
Suppose Assumptions~\ref{ass:iid}, \ref{ass:B}, \ref{ass:Ahet}, \ref{ass:C}, and \ref{ass:CB-LLN} hold. Then,
$\hat\theta_n \xrightarrow{p} \theta_0$.
\end{lemma}

\begin{proof}
\noindent This proof adapts the argument of Lemma~A.8 in \citet{bai2023efficiency} to the covariate-based heterogeneous-shares setting used here.

By Assumption~\ref{ass:C}(a) and Theorem~5.9 in \citet{van2000asymptotic}, it suffices to show that for each
$1\le s\le d_\theta$,
\begin{equation}\label{eq:A8-ULLN-het-cov-final}
\sup_{\theta\in\Theta}
\left|
\frac{1}{n}\sum_{i=1}^n
\Bigl(m_s(X_i,D_i,Y_i;\theta)-\E[m_s(X_i,D_i,Y_i;\theta)]\Bigr)
\right|
\ \xrightarrow{p}\ 0.
\end{equation}
By Assumption~\ref{ass:C}(d), the supremum may be interpreted as an outer supremum. Let $\Theta^\ast$ be the
countable set in Assumption~\ref{ass:C}(d), and take it to contain $\theta_0$ without loss of generality.
Assumption~\ref{ass:C}(d) and dominated convergence (using the envelope in Assumption~\ref{ass:C}(e)(i))
imply that if $m_s(x,d,y,\theta_m)\to m_s(x,d,y,\theta)$ for $\{\theta_m\}\subset\Theta^\ast$, then
$\E[m_s(X_i,D_i,Y_i;\theta_m)]\to \E[m_s(X_i,D_i,Y_i;\theta)]$ (see Problem~2.4.1 in \citet{vaart1997weak}).
Therefore it suffices to establish \eqref{eq:A8-ULLN-het-cov-final} with $\Theta^\ast$ in place of $\Theta$.

Since $Y_i=Y_i(D_i)$, we have the decomposition
\begin{equation}\label{eq:decomp-het-cov-final}
m_s(X_i,D_i,Y_i;\theta)
=
D_i\,m_s(X_i,1,Y_i(1);\theta)
+(1-D_i)\,m_s(X_i,0,Y_i(0);\theta).
\end{equation}
Let $\eta_i:=\Pr(D_i=1\mid X^{(n)})$; under Assumption~\ref{ass:Ahet}, $\eta_i=\eta_g:=T_g/N_g$ for
$i\in\lambda_g$. For $d\in\{0,1\}$ define
\[
\mu_{d,s}(x,\theta):=\E\!\big[m_s(X_i,d,Y_i(d);\theta)\mid X_i=x\big],
\qquad
\bar m_d(\theta):=\E\!\big[m_s(X_i,d,Y_i(d);\theta)\big],
\]
and write $\Delta_s(x,\theta):=\mu_{1,s}(x,\theta)-\mu_{0,s}(x,\theta)$.

Moreover, by Assumption~\ref{ass:iid},
\begin{equation}\label{eq:cond-swap-het-cov-final}
\E\!\big[m_s(X_i,d,Y_i(d);\theta)\mid X^{(n)}\big]
=
\E\!\big[m_s(X_i,d,Y_i(d);\theta)\mid X_i\big]
=
\mu_{d,s}(X_i,\theta),
\end{equation}
since $(X_i,Y_i(d))$ is independent of $\{X_j\}_{j\neq i}$.

Using \eqref{eq:cond-swap-het-cov-final}, for each $i$,
\[
\E\!\big[m_s(X_i,D_i,Y_i;\theta)\mid X^{(n)}\big]
=
\eta_i\,\mu_{1,s}(X_i,\theta)+(1-\eta_i)\,\mu_{0,s}(X_i,\theta)
=
\mu_{0,s}(X_i,\theta)+\eta_i\Delta_s(X_i,\theta).
\]
Add and subtract $\E[m_s(X_i,D_i,Y_i;\theta)\mid X^{(n)}]$ and apply the triangle inequality to obtain
\[
\sup_{\theta\in\Theta^\ast}
\left|
\frac{1}{n}\sum_{i=1}^n
\Bigl(m_s(X_i,D_i,Y_i;\theta)-\E[m_s(X_i,D_i,Y_i;\theta)]\Bigr)
\right|
\le A_{n,1}+A_{n,2},
\]
where
\[
A_{n,1}:=
\sup_{\theta\in\Theta^\ast}
\left|
\frac{1}{n}\sum_{i=1}^n
\Bigl(m_s(X_i,D_i,Y_i;\theta)-\E[m_s(X_i,D_i,Y_i;\theta)\mid X^{(n)}]\Bigr)
\right|,
\]
\[
A_{n,2}:=
\sup_{\theta\in\Theta^\ast}
\left|
\frac{1}{n}\sum_{i=1}^n
\Bigl(\E[m_s(X_i,D_i,Y_i;\theta)\mid X^{(n)}]-\E[m_s(X_i,D_i,Y_i;\theta)]\Bigr)
\right|.
\]
It suffices to show $A_{n,1}=o_p(1)$ and $A_{n,2}=o_p(1)$.

Define
\[
C_n:=
\sup_{\theta\in\Theta^\ast}
\left|
\frac1n\sum_{i=1}^n (D_i-\eta_i)\Delta_s(X_i,\theta)
\right|.
\]

Note that by the law of iterated expectations and the definition
$\eta(x):=\Pr(D=1\mid X=x)$,
\[
\E[m_s(X,D,Y;\theta)]
=
\E\!\left[\mu_{0,s}(X,\theta)+\eta(X)\Delta_s(X,\theta)\right].
\]

Using the display above and the identity
\[
\eta_i\Delta_s(X_i,\theta)
=
D_i\Delta_s(X_i,\theta)-(D_i-\eta_i)\Delta_s(X_i,\theta),
\]
we have
\begin{align*}
A_{n,2}
&=
\sup_{\theta\in\Theta^\ast}\left|
\frac1n\sum_{i=1}^n\Bigl(\mu_{0,s}(X_i,\theta)+\eta_i\Delta_s(X_i,\theta)\Bigr)
-\E\!\Bigl[\mu_{0,s}(X,\theta)+\eta(X)\Delta_s(X,\theta)\Bigr]
\right|\\
&\le
\sup_{\theta\in\Theta^\ast}\left|
\frac1n\sum_{i=1}^n\bigl(\mu_{0,s}(X_i,\theta)-\E[\mu_{0,s}(X,\theta)]\bigr)
\right|\\
&\quad+
\sup_{\theta\in\Theta^\ast}\left|
\frac1n\sum_{i=1}^n\bigl(D_i\Delta_s(X_i,\theta)-\E[\eta(X)\Delta_s(X,\theta)]\bigr)
\right|
\;+\;
C_n.
\end{align*}
For the first term, write
\begin{align*}
\frac1n\sum_{i=1}^n\bigl(\mu_{0,s}(X_i,\theta)-\E[\mu_{0,s}(X,\theta)]\bigr)
&=
\frac1n\sum_{i=1}^n D_i\,g_\theta(X_i)
+\frac1n\sum_{i=1}^n (1-D_i)\,g_\theta(X_i),\\[0.25em]
g_\theta(x)
&:=\mu_{0,s}(x,\theta)-\E[\mu_{0,s}(X,\theta)].
\end{align*}
Since $\Theta^\ast$ is countable and each $g_\theta$ is integrable,
Assumption~\ref{ass:CB-LLN} applied to the countable collection $\{g_\theta:\theta\in\Theta^\ast\}$ yields
\begin{align*}
\sup_{\theta\in\Theta^\ast}\left|
\frac1n\sum_{i=1}^n D_i\,g_\theta(X_i)-\E[\eta(X)g_\theta(X)]
\right|
&=o_p(1),\\[0.35em]
\sup_{\theta\in\Theta^\ast}\left|
\frac1n\sum_{i=1}^n (1-D_i)\,g_\theta(X_i)-\E[(1-\eta(X))g_\theta(X)]
\right|
&=o_p(1).
\end{align*}
Moreover, $\E[\eta(X)g_\theta(X)]+\E[(1-\eta(X))g_\theta(X)]=\E[g_\theta(X)]=0$ for each $\theta$, so the first term above
is $o_p(1)$. For the second term, since $\Theta^\ast$ is countable and each $\Delta_s(\cdot,\theta)$ is integrable,
Assumption~\ref{ass:CB-LLN} applied to the countable collection
$\{\Delta_s(\cdot,\theta):\theta\in\Theta^\ast\}$ implies it is $o_p(1)$. Therefore,
\begin{equation}\label{eq:A2-bound-het-cov-final}
A_{n,2}\ \le\ o_p(1)\ +\ C_n.
\end{equation}

Combine \eqref{eq:decomp-het-cov-final} and \eqref{eq:cond-swap-het-cov-final} to obtain
\begin{align*}
&m_s(X_i,D_i,Y_i;\theta)-\E[m_s(X_i,D_i,Y_i;\theta)\mid X^{(n)}]\\
&=
D_i\bigl(m_s(X_i,1,Y_i(1);\theta)-\mu_{1,s}(X_i,\theta)\bigr)\\
&\quad+(1-D_i)\bigl(m_s(X_i,0,Y_i(0);\theta)-\mu_{0,s}(X_i,\theta)\bigr)\\
&\quad+(D_i-\eta_i)\Delta_s(X_i,\theta).
\end{align*}
Therefore,
\[
A_{n,1}\le B_{n,1}+B_{n,0}+C_n,
\]
where for $d\in\{0,1\}$,
\[
B_{n,d}:=
\sup_{\theta\in\Theta^\ast}
\left|
\frac1n\sum_{i=1}^n \mathbf 1\{D_i=d\}\Bigl(m_s(X_i,d,Y_i(d);\theta)-\mu_{d,s}(X_i,\theta)\Bigr)
\right|.
\]

It is enough to show that for each $d\in\{0,1\}$,
\[
\sup_{\theta\in\Theta^\ast}\left|\frac1n\sum_{i=1}^n (D_i-\eta_i)\mu_{d,s}(X_i,\theta)\right|=o_p(1),
\]
since $\Delta_s(\cdot,\theta)=\mu_{1,s}(\cdot,\theta)-\mu_{0,s}(\cdot,\theta)$.
Fix $d$ and decompose
\[
\frac{1}{n}\sum_{i=1}^n (D_i-\eta_i)\mu_{d,s}(X_i,\theta)
=
\frac{1}{n}\sum_{i=1}^n (D_i-\eta_i)\mu_{d,s}(X_i,\theta_0)
+\frac{1}{n}\sum_{i=1}^n (D_i-\eta_i)\Delta_d(X_i,\theta),
\]
where $\Delta_d(X_i,\theta):=\mu_{d,s}(X_i,\theta)-\mu_{d,s}(X_i,\theta_0)$.

For the first term, use that for each block $g$, $\sum_{i\in\lambda_g}(D_i-\eta_g)=0$ to write
\[
\frac{1}{n}\sum_{i=1}^n (D_i-\eta_i)\mu_{d,s}(X_i,\theta_0)
=
\frac{1}{n}\sum_{g=1}^G\sum_{i\in\lambda_g}(D_i-\eta_g)
\Bigl(\mu_{d,s}(X_i,\theta_0)-\bar\mu_{d,g,s}(\theta_0)\Bigr),
\]
with $\bar\mu_{d,g,s}(\theta_0)=\frac{1}{N_g}\sum_{j\in\lambda_g}\mu_{d,s}(X_j,\theta_0)$.
Thus,
\begin{align*}
\left|\frac{1}{n}\sum_{i=1}^n (D_i-\eta_i)\mu_{d,s}(X_i,\theta_0)\right|
&\le
\frac{1}{n}\sum_{g=1}^G\sum_{i\in\lambda_g}|D_i-\eta_g|\,
\max_{j,k\in\lambda_g}\big|\mu_{d,s}(X_j,\theta_0)-\mu_{d,s}(X_k,\theta_0)\big|\\
&\le
\frac{1}{n}\sum_{g=1}^G N_g\,
\max_{j,k\in\lambda_g}\big|\mu_{d,s}(X_j,\theta_0)-\mu_{d,s}(X_k,\theta_0)\big|.
\end{align*}
By Assumption~\ref{ass:C}(f) at $\theta_0$,
$\max_{j,k\in\lambda_g}|\mu_{d,s}(X_j,\theta_0)-\mu_{d,s}(X_k,\theta_0)|
\le L_d \max_{j,k\in\lambda_g}\|X_j-X_k\|$,
so using $\max_g N_g\le\bar N$ w.p.a.1 and Cauchy--Schwarz with Assumption~\ref{ass:B} yields $o_p(1)$.

For the second term, again $\sum_{i\in\lambda_g}(D_i-\eta_g)=0$ implies
\[
\sum_{i\in\lambda_g}(D_i-\eta_g)\Delta_d(X_i,\theta)
=
\sum_{i\in\lambda_g}(D_i-\eta_g)\Bigl(\Delta_d(X_i,\theta)-\bar\Delta_{d,g}(\theta)\Bigr),
\quad
\bar\Delta_{d,g}(\theta):=\frac1{N_g}\sum_{j\in\lambda_g}\Delta_d(X_j,\theta).
\]
Hence
\[
\sup_{\theta\in\Theta^\ast}\left|\frac1n\sum_{i=1}^n (D_i-\eta_i)\Delta_d(X_i,\theta)\right|
\le
\frac1n\sum_{g=1}^G N_g\,
\sup_{\theta\in\Theta^\ast}\max_{i,j\in\lambda_g}\big|\Delta_d(X_i,\theta)-\Delta_d(X_j,\theta)\big|.
\]
By Assumption~\ref{ass:C}(f), for each $\theta\in\Theta^\ast$,
$\max_{i,j\in\lambda_g}|\Delta_d(X_i,\theta)-\Delta_d(X_j,\theta)|
\le 2L_d\max_{i,j\in\lambda_g}\|X_i-X_j\|$,
and the same $\max_g N_g\le\bar N$ plus Cauchy--Schwarz and Assumption~\ref{ass:B} gives $o_p(1)$.
Therefore $C_n=o_p(1)$, and \eqref{eq:A2-bound-het-cov-final} implies $A_{n,2}=o_p(1)$.

Fix $d\in\{0,1\}$. Write
\begin{align*}
\mathbf 1\{D_i=d\}\Bigl(m_s(X_i,d,Y_i(d);\theta)-\mu_{d,s}(X_i,\theta)\Bigr)
&=
\mathbf 1\{D_i=d\}\Bigl(m_s(X_i,d,Y_i(d);\theta)-\bar m_d(\theta)\Bigr)\\
&\quad-\mathbf 1\{D_i=d\}\Bigl(\mu_{d,s}(X_i,\theta)-\bar m_d(\theta)\Bigr).
\end{align*}
Hence $B_{n,d}\le B^{(1)}_{n,d}+B^{(2)}_{n,d}$, where
\[
B^{(1)}_{n,d}:=
\sup_{\theta\in\Theta^\ast}
\left|
\frac1n\sum_{i=1}^n \mathbf 1\{D_i=d\}\Bigl(m_s(X_i,d,Y_i(d);\theta)-\bar m_d(\theta)\Bigr)
\right|,
\]
\[
B^{(2)}_{n,d}:=
\sup_{\theta\in\Theta^\ast}
\left|
\frac1n\sum_{i=1}^n \mathbf 1\{D_i=d\}\Bigl(\mu_{d,s}(X_i,\theta)-\bar m_d(\theta)\Bigr)
\right|.
\]

Let $a_i(\theta):=m_s(X_i,d,Y_i(d);\theta)-\bar m_d(\theta)$ and define
\[
\mathcal G_d:=\{(x,y)\mapsto m_s(x,d,y;\theta)-\bar m_d(\theta):\theta\in\Theta^\ast\}.
\]
By Assumption~\ref{ass:C}(e)(ii), $\mathcal G_d$ is $Q$-Donsker, hence
\[
\sup_n \E\!\Bigg[\sup_{\theta\in\Theta^\ast}
\Bigl|\frac{1}{\sqrt n}\sum_{i=1}^n \bigl(a_i(\theta)-\E[a_i(\theta)]\bigr)\Bigr|\Bigg]<\infty,
\]
and $\E[a_i(\theta)]=0$ by definition of $\bar m_d(\theta)$.

Let $\mathcal Z_n:=\sigma\{(X_i,Y_i(0),Y_i(1))_{i=1}^n\}$. Conditional on $\mathcal Z_n$, the array
$\{a_i(\theta)\}$ is nonrandom and treatment is assigned independently across blocks with fixed quotas.
Applying Proposition~C.1 of \citet{han2021complex} block-by-block and aggregating yields a universal constant
$C<\infty$ such that
\[
\E\!\Bigg[\sup_{\theta\in\Theta^\ast}
\Bigl|\frac{1}{\sqrt n}\sum_{i=1}^n \mathbf 1\{D_i=d\}\,a_i(\theta)\Bigr|
\ \Bigm|\ \mathcal Z_n\Bigg]
\le
C\,
\E\!\Bigg[\sup_{\theta\in\Theta^\ast}
\Bigl|\frac{1}{\sqrt n}\sum_{i=1}^n \xi_{i,d}\,a_i(\theta)\Bigr|
\ \Bigm|\ \mathcal Z_n\Bigg],
\]
where $\{\xi_{i,d}\}_{i=1}^n$ are conditionally independent Bernoulli random variables with
$\Pr(\xi_{i,1}=1\mid \mathcal Z_n)=\eta_i$ and $\Pr(\xi_{i,0}=1\mid \mathcal Z_n)=1-\eta_i$, independent of
$\mathcal Z_n$. Let $p_{i,1}:=\eta_i$ and $p_{i,0}:=1-\eta_i$. Decompose $\xi_{i,d}=p_{i,d}+(\xi_{i,d}-p_{i,d})$:
\begin{align*}
\E\!\Bigg[\sup_{\theta\in\Theta^\ast}
\Bigl|\frac{1}{\sqrt n}\sum_{i=1}^n \xi_{i,d}\,a_i(\theta)\Bigr|
\ \Bigm|\ \mathcal Z_n\Bigg]
&\le
\sup_{\theta\in\Theta^\ast}\Bigl|\frac{1}{\sqrt n}\sum_{i=1}^n p_{i,d}\,a_i(\theta)\Bigr|\\
&\quad+
\E\!\Bigg[\sup_{\theta\in\Theta^\ast}
\Bigl|\frac{1}{\sqrt n}\sum_{i=1}^n (\xi_{i,d}-p_{i,d})\,a_i(\theta)\Bigr|
\ \Bigm|\ \mathcal Z_n\Bigg].
\end{align*}
Since $0\le p_{i,d}\le 1$ and $\mathcal G_d$ is Donsker, the first term is bounded in expectation uniformly in $n$.
For the second term, $(\xi_{i,d}-p_{i,d})$ are conditionally independent, mean zero, and uniformly bounded, so a
standard multiplier maximal inequality for Donsker classes implies
\[
\sup_n \E\!\Bigg[\sup_{\theta\in\Theta^\ast}
\Bigl|\frac{1}{\sqrt n}\sum_{i=1}^n (\xi_{i,d}-p_{i,d})\,a_i(\theta)\Bigr|\Bigg]<\infty.
\]
Therefore $\sup_n \E\!\left[\sup_{\theta\in\Theta^\ast}\Bigl|\frac{1}{\sqrt n}\sum_{i=1}^n \mathbf 1\{D_i=d\}a_i(\theta)\Bigr|\right]<\infty$,
and multiplying by $1/\sqrt n$ yields $\E[B^{(1)}_{n,d}]\to 0$. By Markov's inequality, $B^{(1)}_{n,d}=o_p(1)$.

Let $p_{i,d}:=\E[\mathbf 1\{D_i=d\}\mid X^{(n)}]$ so $p_{i,1}=\eta_i$ and $p_{i,0}=1-\eta_i$, and decompose
\[
\frac1n\sum_{i=1}^n \mathbf 1\{D_i=d\}\bigl(\mu_{d,s}(X_i,\theta)-\bar m_d(\theta)\bigr)
=
T_{n,d}(\theta)+U_{n,d}(\theta),
\]
where
\begin{align*}
T_{n,d}(\theta)
&:=\frac1n\sum_{i=1}^n
\bigl(\mathbf 1\{D_i=d\}-p_{i,d}\bigr)\bigl(\mu_{d,s}(X_i,\theta)-\bar m_d(\theta)\bigr),\\
U_{n,d}(\theta)
&:=\frac1n\sum_{i=1}^n
p_{i,d}\bigl(\mu_{d,s}(X_i,\theta)-\bar m_d(\theta)\bigr).
\end{align*}
Because $p_{i,d}=p_{g,d}$ within block and $\sum_{i\in\lambda_g}(\mathbf 1\{D_i=d\}-p_{g,d})=0$,
\begin{align*}
T_{n,d}(\theta)
&=
\frac1n\sum_{g=1}^G\sum_{i\in\lambda_g}
\bigl(\mathbf 1\{D_i=d\}-p_{g,d}\bigr)
\Bigl(\mu_{d,s}(X_i,\theta)-\bar\mu_{d,g,s}(\theta)\Bigr),\\
\bar\mu_{d,g,s}(\theta)
&:=\frac1{N_g}\sum_{j\in\lambda_g}\mu_{d,s}(X_j,\theta).
\end{align*}
Hence
\[
\sup_{\theta\in\Theta^\ast}|T_{n,d}(\theta)|
\le
\frac1n\sum_{g=1}^G N_g\,
\sup_{\theta\in\Theta^\ast}\max_{i,j\in\lambda_g}\big|\mu_{d,s}(X_i,\theta)-\mu_{d,s}(X_j,\theta)\big|
=o_p(1),
\]
by Assumptions~\ref{ass:C}(f) and \ref{ass:B}. For $U_{n,d}$,
\[
U_{n,d}(\theta)
=
\frac1n\sum_{g=1}^G N_g p_{g,d}\Bigl(\bar\mu_{d,g,s}(\theta)-\bar m_d(\theta)\Bigr).
\]
Adding and subtracting $\mu_{d,s}(X_i,\theta)$ within blocks yields
\begin{align*}
\sup_{\theta\in\Theta^\ast}|U_{n,d}(\theta)|
&\le
\sup_{\theta\in\Theta^\ast}\left|
\frac1n\sum_{i=1}^n\bigl(\mu_{d,s}(X_i,\theta)-\bar m_d(\theta)\bigr)
\right|\\
&\quad+
\frac1n\sum_{g=1}^G N_g\,
\sup_{\theta\in\Theta^\ast}\max_{i,j\in\lambda_g}
\big|\mu_{d,s}(X_i,\theta)-\mu_{d,s}(X_j,\theta)\big|.
\end{align*}
The second term is $o_p(1)$ by the same Lipschitz plus balancing bound as above. For the first term, let
$h_\theta(x):=\mu_{d,s}(x,\theta)-\bar m_d(\theta)$, so $\E[h_\theta(X)]=0$ for each $\theta$. Since $\Theta^\ast$ is
countable and each $h_\theta$ is integrable, Assumption~\ref{ass:CB-LLN} applied to the countable collection
$\{h_\theta:\theta\in\Theta^\ast\}$ gives
\begin{align*}
\sup_{\theta\in\Theta^\ast}\left|
\frac1n\sum_{i=1}^n D_i\,h_\theta(X_i)-\E[\eta(X)h_\theta(X)]
\right|
&=o_p(1),\\[0.35em]
\sup_{\theta\in\Theta^\ast}\left|
\frac1n\sum_{i=1}^n (1-D_i)\,h_\theta(X_i)-\E[(1-\eta(X))h_\theta(X)]
\right|
&=o_p(1).
\end{align*}
and $\E[\eta(X)h_\theta(X)]+\E[(1-\eta(X))h_\theta(X)]=\E[h_\theta(X)]=0$. Hence
\[
\sup_{\theta\in\Theta^\ast}\left|\frac1n\sum_{i=1}^n h_\theta(X_i)\right|
=
\sup_{\theta\in\Theta^\ast}\left|\frac1n\sum_{i=1}^n D_i\,h_\theta(X_i)+\frac1n\sum_{i=1}^n (1-D_i)\,h_\theta(X_i)\right|
=o_p(1).
\]
Therefore $\sup_{\theta\in\Theta^\ast}|U_{n,d}(\theta)|=o_p(1)$, so $B^{(2)}_{n,d}=o_p(1)$ and thus $B_{n,d}=o_p(1)$ for
$d\in\{0,1\}$.

Combining $B_{n,1}=o_p(1)$, $B_{n,0}=o_p(1)$, and $C_n=o_p(1)$ yields $A_{n,1}=o_p(1)$. Together with $A_{n,2}=o_p(1)$,
this implies that the left-hand side of \eqref{eq:A8-ULLN-het-cov-final} is $o_p(1)$, proving
\eqref{eq:A8-ULLN-het-cov-final} and completing the proof.
\end{proof}

\subsection{Auxiliary Lemma for the Asymptotic Distribution Proof}

\noindent This next lemma adapts Lemma~A.6 in \citet{bai2023efficiency} to the covariate-based
heterogeneous-shares setting used here.

\begin{lemma}\label{lem:A6-het-cov-correct}
Suppose Assumptions~\ref{ass:iid}, \ref{ass:Ahet}, \ref{ass:B}, \ref{ass:CB-LLN}, and \ref{ass:CB-IMB} hold.
Let $f(x,d,y)\in\mathbb R^{k}$ satisfy
\[
\E\!\big[f(X,D,Y)\big]=0,
\qquad
\E\!\big[\|f(X,d,Y(d))\|^{2+\delta}\big]<\infty\quad \text{for } d\in\{0,1\},
\]
for some $\delta>0$. Define, for $d\in\{0,1\}$,
\[
\mu_d(x):=\E\!\big[f(X,d,Y(d))\mid X=x\big],\qquad
\Sigma_d(x):=\Var\!\big(f(X,d,Y(d))\mid X=x\big),
\]
and $\Delta(x):=\mu_1(x)-\mu_0(x)$. Let $\eta_i=\Pr(D_i=1\mid X^{(n)})$ and define
\begin{align*}
C_{1,n}
&:=\frac{1}{\sqrt n}\sum_{i=1}^n \big(\mu_0(X_i)+\eta_i\Delta(X_i)\big),\\
V_{3,f}
&:=\E\!\Big[\eta(X)\Sigma_1(X)+(1-\eta(X))\Sigma_0(X)\Big].
\end{align*}
Assume that $C_{1,n}\Rightarrow \mathcal N(0,V_{1,f})$ for some deterministic, finite $k\times k$ matrix $V_{1,f}$.
Then
\[
\frac{1}{\sqrt n}\sum_{i=1}^n f(X_i,D_i,Y_i)\ \xRightarrow{d}\ \mathcal N(0,V_f),
\qquad
V_f\ =\ V_{1,f}+V_{\mathrm{imb},\Delta}+V_{3,f}.
\]
\end{lemma}

\begin{proof}
Let
\[
C_n:=\frac{1}{\sqrt n}\sum_{i=1}^n f(X_i,D_i,Y_i).
\]
Using $Y_i=Y_i(D_i)$,
\[
f(X_i,D_i,Y_i)=D_i f(X_i,1,Y_i(1))+(1-D_i)f(X_i,0,Y_i(0)).
\]
Add and subtract $\mu_d(X_i)$ to obtain
\[
C_n=C_{1,n}+C_{2,n}+C_{3,n},
\]
where
\[
C_{1,n}:=\frac{1}{\sqrt n}\sum_{i=1}^n\big(\mu_0(X_i)+\eta_i\Delta(X_i)\big),
\qquad
C_{2,n}:=\frac{1}{\sqrt n}\sum_{i=1}^n (D_i-\eta_i)\Delta(X_i),
\]
and
\[
C_{3,n}
:=\frac{1}{\sqrt n}\sum_{i=1}^n\Big\{
D_i\big(f(X_i,1,Y_i(1))-\mu_1(X_i)\big)
+(1-D_i)\big(f(X_i,0,Y_i(0))-\mu_0(X_i)\big)
\Big\}.
\]
By construction,
\begin{align*}
C_{1,n}&=\E[C_n\mid X^{(n)}],\\
C_{1,n}+C_{2,n}&=\E[C_n\mid X^{(n)},D^{(n)}],\\
C_{3,n}&=C_n-\E[C_n\mid X^{(n)},D^{(n)}].
\end{align*}

We first consider $C_{2,n}$. Note that, under fixed within-block quotas,
\[
\sum_{i=1}^n (D_i-\eta_i)
=
\sum_{g=1}^G \sum_{i\in\lambda_g} (D_i-\eta_g)
=
\sum_{g=1}^G (T_g - N_g\eta_g)
=0
\quad\text{identically.}
\]
Therefore centering does not change $C_{2,n}$:
\[
\frac{1}{\sqrt n}\sum_{i=1}^n (D_i-\eta_i)\big(\Delta(X_i)-\E[\Delta(X)]\big)
=
\frac{1}{\sqrt n}\sum_{i=1}^n (D_i-\eta_i)\Delta(X_i)
=
C_{2,n}.
\]
By Assumption~\ref{ass:CB-IMB} applied with $\gamma=\Delta-\E[\Delta(X)]$, letting
$V_{\mathrm{imb},\Delta}$ denote the corresponding $V_{\mathrm{imb},\gamma}$, we have
\[
\rho\!\left(C_{2,n},\ \mathcal N(0,V_{\mathrm{imb},\Delta})\ \middle|\ X^{(n)}\right)\ \xrightarrow{P}\ 0.
\]

Next consider $C_{3,n}$. Define
\[
\zeta_i
:=
D_i\big(f(X_i,1,Y_i(1))-\mu_1(X_i)\big)
+(1-D_i)\big(f(X_i,0,Y_i(0))-\mu_0(X_i)\big),
\]
so that $C_{3,n}=\frac{1}{\sqrt n}\sum_{i=1}^n \zeta_i$ and
$\E[\zeta_i\mid X^{(n)},D^{(n)}]=0$.
Conditional on $(X^{(n)},D^{(n)})$, the $\zeta_i$ are independent across $i$ by
Assumptions~\ref{ass:iid} and \ref{ass:Ahet}. Moreover,
\[
\Var(\zeta_i\mid X^{(n)},D^{(n)})
=
D_i\,\Sigma_1(X_i)+(1-D_i)\,\Sigma_0(X_i).
\]
Thus
\[
s_{3,n}^2
:=
\Var(C_{3,n}\mid X^{(n)},D^{(n)})
=
\frac{1}{n}\sum_{i=1}^n\Big(D_i\,\Sigma_1(X_i)+(1-D_i)\,\Sigma_0(X_i)\Big).
\]

We now show that $s_{3,n}^2\xrightarrow{P}V_{3,f}$. Rewrite
\[
s_{3,n}^2
=
\frac{1}{n}\sum_{i=1}^n \Sigma_0(X_i)
+\frac{1}{n}\sum_{i=1}^n D_i\Big(\Sigma_1(X_i)-\Sigma_0(X_i)\Big).
\]
By the i.i.d.\ LLN and $\E[\|\Sigma_0(X)\|]<\infty$,
\[
\frac{1}{n}\sum_{i=1}^n \Sigma_0(X_i)\ \xrightarrow{P}\ \E[\Sigma_0(X)].
\]
Moreover, applying Assumption~\ref{ass:CB-LLN} componentwise with
$\gamma(X)=\Sigma_1(X)-\Sigma_0(X)$ yields
\[
\frac{1}{n}\sum_{i=1}^n D_i\Big(\Sigma_1(X_i)-\Sigma_0(X_i)\Big)
\ \xrightarrow{P}\
\E\!\Big[\eta(X)\big(\Sigma_1(X)-\Sigma_0(X)\big)\Big].
\]
Therefore,
\[
s_{3,n}^2\ \xrightarrow{P}\ \E\!\Big[\eta(X)\Sigma_1(X)+(1-\eta(X))\Sigma_0(X)\Big]\;=\;V_{3,f}.
\]

The $(2+\delta)$-moment bound implies the conditional Lyapunov condition for
$\{\zeta_i/\sqrt n\}_{i=1}^n$ given $(X^{(n)},D^{(n)})$, so the conditional Lyapunov CLT yields
\[
\rho\!\left(C_{3,n},\ \mathcal N(0,s_{3,n}^2)\ \middle|\ X^{(n)},D^{(n)}\right)\ \xrightarrow{P}\ 0.
\]

It remains to combine the three terms. Fix $t\in\mathbb R^k$ and write the characteristic function of $C_n$ as
\[
\phi_n(t)
=
\E\!\left[\exp\!\big(i t'C_{1,n}\big)\ \E\!\left[\exp\!\big(i t'(C_{2,n}+C_{3,n})\big)\ \middle|\ X^{(n)}\right]\right].
\]
Since $C_{2,n}$ is measurable with respect to $(X^{(n)},D^{(n)})$,
\begin{align*}
\E\!\left[\exp\!\big(i t'(C_{2,n}+C_{3,n})\big)\ \middle|\ X^{(n)}\right]
&=
\E\!\left[\exp\!\big(i t' C_{2,n}\big)\
\E\!\left[\exp\!\big(i t' C_{3,n}\big)\ \middle|\ X^{(n)},D^{(n)}\right]\ \middle|\ X^{(n)}\right].
\end{align*}

We claim that
\[
\E\!\left[\exp\!\big(i t'(C_{2,n}+C_{3,n})\big)\ \middle|\ X^{(n)}\right]
\ \xrightarrow{P}\
\exp\!\Big(-\tfrac12 t'(V_{\mathrm{imb},\Delta}+V_{3,f})t\Big).
\]
To show this, write
\[
\E\!\left[\exp\!\big(i t'(C_{2,n}+C_{3,n})\big)\ \middle|\ X^{(n)}\right]
=
\E\!\left[\exp\!\big(i t' C_{2,n}\big)\,W_n(t)\ \middle|\ X^{(n)}\right],
\]
where $W_n(t):=\E[\exp(i t' C_{3,n})\mid X^{(n)},D^{(n)}]$, and set
$c(t):=\exp(-\tfrac12 t'V_{3,f}t)$.

By $\rho(C_{3,n},\mathcal N(0,s_{3,n}^2)\mid X^{(n)},D^{(n)})\xrightarrow{P}0$ and the fact that
$x\mapsto \exp(i t'x)$ is bounded and continuous, we have
\[
\Big|W_n(t)-\exp\!\big(-\tfrac12 t's_{3,n}^2 t\big)\Big|\ \xrightarrow{P}\ 0.
\]
Since $s_{3,n}^2\xrightarrow{P}V_{3,f}$, it follows that $W_n(t)\xrightarrow{P}c(t)$.
Moreover, $|W_n(t)-c(t)|\le 2$, so $\{|W_n(t)-c(t)|\}$ is uniformly integrable and therefore
\[
\E\big[|W_n(t)-c(t)|\big]\ \longrightarrow\ 0.
\]
Hence, by Markov's inequality,
\[
\E\!\left[\ |W_n(t)-c(t)|\ \middle|\ X^{(n)}\right]\ \xrightarrow{P}\ 0.
\]
Using $|\exp(i t' C_{2,n})|=1$,
\begin{align*}
&\Bigg|
\E\!\left[\exp\!\big(i t' C_{2,n}\big)\,W_n(t)\ \middle|\ X^{(n)}\right]
-
c(t)\,\E\!\left[\exp\!\big(i t' C_{2,n}\big)\ \middle|\ X^{(n)}\right]
\Bigg|
\\
&\hspace{4em}
\le
\E\!\left[\ |W_n(t)-c(t)|\ \middle|\ X^{(n)}\right]
\ \xrightarrow{P}\ 0.
\end{align*}
By Assumption~\ref{ass:CB-IMB} applied to $\gamma=\Delta-\E[\Delta(X)]$,
\[
\E\!\left[\exp\!\big(i t' C_{2,n}\big)\ \middle|\ X^{(n)}\right]
\ \xrightarrow{P}\
\exp\!\Big(-\tfrac12 t'V_{\mathrm{imb},\Delta}t\Big).
\]
Combining the last two displays yields the claimed convergence.

Therefore,
\[
\E\!\left[\exp\!\big(i t'(C_{2,n}+C_{3,n})\big)\ \middle|\ X^{(n)}\right]
=
\exp\!\Big(-\tfrac12 t'(V_{\mathrm{imb},\Delta}+V_{3,f})t\Big)+o_P(1).
\]
Using $|\exp(i t'C_{1,n})|\le 1$ and iterated expectations, we obtain
\[
\phi_n(t)
=
\E\!\left[\exp\!\big(i t'C_{1,n}\big)\right]\cdot
\exp\!\Big(-\tfrac12 t'(V_{\mathrm{imb},\Delta}+V_{3,f})t\Big)+o(1).
\]
By the assumption $C_{1,n}\Rightarrow \mathcal N(0,V_{1,f})$,
\[
\E\!\left[\exp\!\big(i t'C_{1,n}\big)\right]\ \longrightarrow\ \exp\!\Big(-\tfrac12 t'V_{1,f}t\Big).
\]
Therefore,
\[
\phi_n(t)\ \longrightarrow\ \exp\!\Big(-\tfrac12 t'V_f t\Big),
\qquad
V_f=V_{1,f}+V_{\mathrm{imb},\Delta}+V_{3,f},
\]
which implies $C_n\Rightarrow \mathcal N(0,V_f)$.
\end{proof}

\subsection{Proof of Theorem \ref{thm:clt_het_gmm}}\label{proof_hetshares_asympt}

\noindent This proof adapts the argument of Theorem~3.1 in Appendix A.1 of \citet{bai2023efficiency} to the covariate-based heterogeneous-shares setting used here.

By Lemma~\ref{lem:A8-het-cov}, $\hat\theta_n\xrightarrow{p}\theta_0$. By definition of $\theta_0$,
$\E[m(X,D,Y;\theta_0)]=0$.

Following the proof of Theorem~5.21 in \citet{van2000asymptotic}, to obtain the asymptotic linear representation
it suffices to verify the stochastic equicontinuity condition
\begin{equation}
\label{eq:Ln-hat-theta-het-cov}
L_n(\hat\theta_n)\xrightarrow{p}0,
\end{equation}
where $L_n(\theta)=(L_n^{(1)}(\theta),\dots,L_n^{(d_\theta)}(\theta))'$ and, for each $s=1,\dots,d_\theta$,
\begin{align*}
L_n^{(s)}(\theta)
:=
&\ \frac{1}{\sqrt{n}}\sum_{i=1}^n
\Big(m_s(X_i,D_i,Y_i;\theta)-\E[m_s(X_i,D_i,Y_i;\theta)]\Big)\\
&-\frac{1}{\sqrt{n}}\sum_{i=1}^n
\Big(m_s(X_i,D_i,Y_i;\theta_0)-\E[m_s(X_i,D_i,Y_i;\theta_0)]\Big).
\end{align*}
Once \eqref{eq:Ln-hat-theta-het-cov} holds, the Z-estimator expansion yields
\begin{equation}
\label{eq:linear-het-cov-final}
\sqrt{n}(\hat\theta_n-\theta_0)
=
-\,M^{-1}\frac{1}{\sqrt{n}}\sum_{i=1}^n m(X_i,D_i,Y_i;\theta_0)+o_p(1).
\end{equation}

We verify \eqref{eq:Ln-hat-theta-het-cov} componentwise.
Using Assumption~\ref{ass:C}(c)--(d) and Proposition~8.11 of \citet{kosorok2008introduction}, we may restrict to the
countable set $\Theta^\ast$ in Assumption~\ref{ass:C}(d), so for any sequence $\delta_n\downarrow 0$ it suffices to show
\[
\sup_{\theta\in\Theta^\ast:\ \|\theta-\theta_0\|<\delta_n}\big|L_n^{(s)}(\theta)\big| \xrightarrow{p}0
\qquad (s=1,\dots,d_\theta).
\]
Decomposing by treatment status,
\[
L_n^{(s)}(\theta)=L^{(s)}_{n,1}(\theta)+L^{(s)}_{n,0}(\theta),
\]
where
\begin{align*}
L^{(s)}_{n,1}(\theta)
&=
\frac{1}{\sqrt{n}}\sum_{i=1}^n D_i
\Big(\Delta m_{1,s}(i,\theta)-\E[\Delta m_{1,s}(i,\theta)]\Big),\\
L^{(s)}_{n,0}(\theta)
&=
\frac{1}{\sqrt{n}}\sum_{i=1}^n (1-D_i)
\Big(\Delta m_{0,s}(i,\theta)-\E[\Delta m_{0,s}(i,\theta)]\Big).
\end{align*}
with $\Delta m_{d,s}(i,\theta):=m_s(X_i,d,Y_i(d);\theta)-m_s(X_i,d,Y_i(d);\theta_0)$.

Fix $(d,s)$ and define the $L_2(Q)$ semimetric
\[
\rho_{Q,d,s}(\theta,\theta_0)
:=
\Big\{\E\big[(\Delta m_{d,s}(X,d,Y(d);\theta)-\E[\Delta m_{d,s}(X,d,Y(d);\theta)])^2\big]\Big\}^{1/2},
\]
which is continuous at $\theta_0$ by Assumption~\ref{ass:C}(c). Let
\[
\mathcal F_n := \sigma\{(X_i,Y_i(0),Y_i(1))_{i=1}^n\}.
\]
For any sequence $\tilde\delta_n\downarrow 0$, Proposition~C.1 of \citet{han2021complex}, applied conditional on
$\mathcal F_n$ and using the blockwise complete-randomization assignment in Assumption~\ref{ass:Ahet}, yields
\[
\E\Big[\sup_{\rho_{Q,d,s}(\theta,\theta_0)<\tilde\delta_n} |L^{(s)}_{n,d}(\theta)|\ \Big|\ \mathcal F_n\Big]
\ \le\
C\,\E\Big[\sup_{\rho_{Q,d,s}(\theta,\theta_0)<\tilde\delta_n}
\Big|\frac{1}{\sqrt{n}}\sum_{i=1}^n (\xi_{i,d}-p_{i,d})\,\Delta m_{d,s}(i,\theta)\Big|\ \Big|\ \mathcal F_n\Big],
\]
for a universal constant $C<\infty$, where $\{\xi_{i,d}\}_{i=1}^n$ are conditionally independent Bernoulli draws with
success probabilities $p_{i,1}=\eta_i$ and $p_{i,0}=1-\eta_i$, constructed on an auxiliary probability space as in
\citet{han2021complex}.
Assumption~\ref{ass:C}(e)(ii) (Donsker) and Corollary~2.3.12 in \citet{vaart1997weak} imply the right-hand side converges
to $0$ as $\tilde\delta_n\downarrow 0$, hence
$\sup_{\|\theta-\theta_0\|<\delta_n,\ \theta\in\Theta^\ast}|L_n^{(s)}(\theta)|\xrightarrow{p}0$ by Markov’s inequality.
Therefore $L_n(\hat\theta_n)\xrightarrow{p}0$, which establishes \eqref{eq:Ln-hat-theta-het-cov} and hence
\eqref{eq:linear-het-cov-final}.

\medskip
Define
\[
m_n^\ast(X_i,D_i,Y_i;\theta_0)
:=
m(X_i,D_i,Y_i;\theta_0)-(D_i-\eta_i)\Delta(X_i).
\]
Using $Y_i=Y_i(D_i)$ and adding/subtracting $\mu_d(X_i)$, we have the decomposition
\begin{equation}
\label{eq:m-decomp-het-cov-final}
m(X_i,D_i,Y_i;\theta_0)=m_n^\ast(X_i,D_i,Y_i;\theta_0) + (D_i-\eta_i)\Delta(X_i).
\end{equation}
Thus, combining \eqref{eq:linear-het-cov-final} with \eqref{eq:m-decomp-het-cov-final}, it suffices to show that
\begin{equation}
\label{eq:imb-op1-final}
\frac{1}{\sqrt{n}}\sum_{i=1}^n (D_i-\eta_i)\Delta(X_i)=o_p(1).
\end{equation}

Fix $a\in\mathbb{R}^{d_\theta}$ with $\|a\|=1$ and define the scalar $\Delta_a(x):=a'\Delta(x)$.
By Cram\'er--Wold, it suffices to prove
\[
\frac{1}{\sqrt{n}}\sum_{i=1}^n (D_i-\eta_i)\Delta_a(X_i)=o_p(1).
\]
Fix a block $g$ and write $\bar\Delta_{a,g}:=N_g^{-1}\sum_{i\in\lambda_g}\Delta_a(X_i)$. Since
$\sum_{i\in\lambda_g}(D_i-\eta_g)=0$ under fixed within-block quotas and $T_g\in\{1,\dots,N_g-1\}$
(Assumption~\ref{ass:Ahet}), we have $N_g\ge 2$ and
\[
\sum_{i\in\lambda_g}(D_i-\eta_g)\Delta_a(X_i)
=
\sum_{i\in\lambda_g}(D_i-\eta_g)\big(\Delta_a(X_i)-\bar\Delta_{a,g}\big).
\]
Under complete randomization with exactly $T_g$ treated units in block $g$,
\[
\Var\!\left(\sum_{i\in\lambda_g}(D_i-\eta_g)\big(\Delta_a(X_i)-\bar\Delta_{a,g}\big)\ \Bigm|\ X^{(n)}\right)
=
\frac{N_g\,\eta_g(1-\eta_g)}{N_g-1}\sum_{i\in\lambda_g}\big(\Delta_a(X_i)-\bar\Delta_{a,g}\big)^2.
\]
Therefore, using independence across blocks given $X^{(n)}$,
\begin{align*}
\Var\!\left(\frac{1}{\sqrt{n}}\sum_{i=1}^n (D_i-\eta_i)\Delta_a(X_i)\ \Bigm|\ X^{(n)}\right)
&=
\frac{1}{n}\sum_{g=1}^G
\frac{N_g\,\eta_g(1-\eta_g)}{N_g-1}\sum_{i\in\lambda_g}\big(\Delta_a(X_i)-\bar\Delta_{a,g}\big)^2\\
&\le
\frac{1}{n}\sum_{g=1}^G C_0\,N_g\,\max_{i,j\in\lambda_g}\big|\Delta_a(X_i)-\Delta_a(X_j)\big|^2,
\end{align*}
for a constant $C_0<\infty$ because $\eta_g(1-\eta_g)\le 1/4$ and $N_g/(N_g-1)\le 2$ when $N_g\ge 2$.

By Assumption~\ref{ass:C}(f) at $\theta_0$ and $|\Delta_a(x)-\Delta_a(x')|\le \|\Delta(x)-\Delta(x')\|$,
$\Delta_a(\cdot)$ is Lipschitz with constant $L_\Delta\le L_1+L_0$, hence
$\max_{i,j\in\lambda_g}|\Delta_a(X_i)-\Delta_a(X_j)|^2 \le L_\Delta^2\max_{i,j\in\lambda_g}\|X_i-X_j\|^2$.
Using $\max_g N_g\le \bar N$ w.p.a.1 (Assumption~\ref{ass:B}) and the covariate balancing condition in
Assumption~\ref{ass:B} gives
\[
\Var\!\left(\frac{1}{\sqrt{n}}\sum_{i=1}^n (D_i-\eta_i)\Delta_a(X_i)\ \Bigm|\ X^{(n)}\right)\xrightarrow{p}0,
\]
so Chebyshev’s inequality yields $n^{-1/2}\sum_{i=1}^n (D_i-\eta_i)\Delta_a(X_i)=o_p(1)$.
Since $a$ was arbitrary, \eqref{eq:imb-op1-final} follows by Cram\'er--Wold.
Combining \eqref{eq:linear-het-cov-final}, \eqref{eq:m-decomp-het-cov-final}, and \eqref{eq:imb-op1-final} yields
\begin{equation}\label{eq:IF-het-cov-final}
\sqrt{n}\,(\hat\theta_n-\theta_0)
=
-\,M^{-1}\frac{1}{\sqrt{n}}\sum_{i=1}^n m_n^\ast(X_i,D_i,Y_i;\theta_0)
+o_p(1).
\end{equation}

\medskip
By \eqref{eq:IF-het-cov-final} and Cram\'er--Wold, it suffices to obtain a CLT for
$n^{-1/2}\sum_{i=1}^n m_n^\ast(X_i,D_i,Y_i;\theta_0)$.

We first show that the imbalance component associated with $\Delta(\cdot)$ is asymptotically degenerate in the present
fine-stratification regime. In the notation of Assumption~\ref{ass:CB-IMB}, define $\gamma_\Delta(x):=\Delta(x)-\E[\Delta(X)]$,
so that $\E[\gamma_\Delta(X)]=0$. Let $S_{\Delta,g}:=S_{\gamma_\Delta,g}$ denote the within-block covariance matrix of $\gamma_\Delta(X)$.
Since $\Delta(\cdot)$ is Lipschitz at $\theta_0$ by Assumption~\ref{ass:C}(f), the covariate balancing condition in Assumption~\ref{ass:B}
implies $\max_{i,j\in\lambda_g}\|\Delta(X_i)-\Delta(X_j)\|^2\le L_\Delta^2\max_{i,j\in\lambda_g}\|X_i-X_j\|^2$ and hence, for any matrix norm $\|\cdot\|$,
\[
\left\|
\frac{1}{n}\sum_{g=1}^G N_g\,\eta_g(1-\eta_g)\,S_{\Delta,g}
\right\|
\ \xrightarrow{P}\ 0.
\]
Therefore, the corresponding imbalance variance is degenerate:
\begin{equation}
\label{eq:VimbDelta-degenerate}
V_{\mathrm{imb},\Delta}
=\plim\ \frac{1}{n}\sum_{g=1}^G N_g\,\eta_g(1-\eta_g)\,S_{\Delta,g}
=0.
\end{equation}

Lemma~\ref{lem:A6-het-cov-correct}, applied with $f(x,d,y)=m(x,d,y;\theta_0)$, yields
\[
\frac{1}{\sqrt{n}}\sum_{i=1}^n m(X_i,D_i,Y_i;\theta_0)
\ \xRightarrow{d}\
\mathcal N\!\big(0,\ V_{1,m}+V_{\mathrm{imb},\Delta}+V_{3,m}\big),
\]
where $V_{3,m}=\E[\eta(X)\Sigma_1(X)+(1-\eta(X))\Sigma_0(X)]$ provided that
\[
C_{1,n}\ \dto\ \mathcal N(0,V_{1,m}).
\]
This is exactly the lemma's requirement on the ``superpopulation'' component.
By \eqref{eq:VimbDelta-degenerate}, this reduces to
\begin{equation}
\label{eq:CLT-sum-m-reduced}
\frac{1}{\sqrt{n}}\sum_{i=1}^n m(X_i,D_i,Y_i;\theta_0)
\ \xRightarrow{d}\
\mathcal N\!\big(0,\ V_{1,m}+V_{3,m}\big).
\end{equation}

Next, by \eqref{eq:m-decomp-het-cov-final} and \eqref{eq:imb-op1-final},
\begin{align*}
\frac{1}{\sqrt{n}}\sum_{i=1}^n m(X_i,D_i,Y_i;\theta_0)
&=
\frac{1}{\sqrt{n}}\sum_{i=1}^n m_n^\ast(X_i,D_i,Y_i;\theta_0)
+\frac{1}{\sqrt{n}}\sum_{i=1}^n (D_i-\eta_i)\Delta(X_i)\\
&=
\frac{1}{\sqrt{n}}\sum_{i=1}^n m_n^\ast(X_i,D_i,Y_i;\theta_0)+o_p(1).
\end{align*}
Combining this with \eqref{eq:CLT-sum-m-reduced} and Slutsky’s theorem yields
\[
\frac{1}{\sqrt{n}}\sum_{i=1}^n m_n^\ast(X_i,D_i,Y_i;\theta_0)
\ \xRightarrow{d}\
\mathcal N\!\big(0,\ V_{1,m}+V_{3,m}\big).
\]
Finally, premultiplying and postmultiplying by $M^{-1}$ in \eqref{eq:IF-het-cov-final} yields
\begin{equation}\label{eq:CLT-het-cov-final}
\sqrt{n}\,(\hat\theta_n-\theta_0)
\ \dto\
\mathcal N\!\bigl(0,\ V_\ast\bigr),
\qquad
V_\ast = M^{-1}\Omega_{\eta(\cdot)}M^{-1\prime}.
\end{equation}

\subsection{Consistency of the Variance Estimator in the Covariate-Based, Hetero\-geneous-Shares Stratification Setting}

\subsubsection{Deterministic-sequence LLN for arm-specific moment means}

\noindent This next lemma adapts Lemma~A.9 in \citet{bai2023efficiency} to the covariate-based
heterogeneous-shares setting used here. To state the lemma, we first introduce arm-specific assignment probabilities and the corresponding IPW moment means.

For $d\in\{0,1\}$ and $i\in\lambda_g$, define the arm-$d$ assignment probability
\[
\eta_{d,b_i}
:=
\Pr(D_i=d\mid X^{(n)})
=
\begin{cases}
\eta_g, & d=1,\\
1-\eta_g, & d=0.
\end{cases}
\]
For $1\le s\le d_\theta$ and $\theta\in\Theta$, define
\begin{equation}\label{eq:muhatA9_het_final}
\widehat\mu^{(s)}_{d,n}(\theta)
:=
\frac{1}{n}\sum_{i=1}^n
\frac{\mathbf 1\{D_i=d\}}{\eta_{d,b_i}}\,
m_s\!\left(X_i,d,Y_i(d);\theta\right),
\qquad
\mu^{(s)}_{d}(\theta)
:=
\E\!\left[m_s\!\left(X,d,Y(d);\theta\right)\right],
\end{equation}
and write $\mu^{(s)}_{d}:=\mu^{(s)}_{d}(\theta_0)$.

\begin{lemma}[Deterministic-sequence LLN for arm-$d$ moment means]\label{lem:A9_het_final}
Suppose Assumptions~\ref{ass:iid}, \ref{ass:B}, \ref{ass:Ahet}, and \ref{ass:C} hold.
In addition, suppose Assumption~\ref{ass:CB-var-regularity} holds.
Then, for each $d\in\{0,1\}$ and each deterministic sequence $\theta_n\to\theta_0$,
\[
\widehat\mu^{(s)}_{d,n}(\theta_n)\ \xrightarrow{P}\ \mu^{(s)}_{d},
\qquad\text{for every }1\le s\le d_\theta.
\]
\end{lemma}

\begin{proof}
Fix $d\in\{0,1\}$ and $1\le s\le d_\theta$, and let $\theta_n\to\theta_0$ be deterministic.
Define $M_i(\theta):=m_s(X_i,d,Y_i(d);\theta)$.
It suffices to show
\begin{align}
\widehat\mu^{(s)}_{d,n}(\theta_n)
-\E_n\!\left[\widehat\mu^{(s)}_{d,n}(\theta_n)\mid X^{(n)},Y^{(n)}(d)\right]
&\xrightarrow{P}0,
\label{eq:A9_final_51}\\
\E_n\!\left[\widehat\mu^{(s)}_{d,n}(\theta_n)\mid X^{(n)},Y^{(n)}(d)\right]
&\xrightarrow{P}\mu^{(s)}_{d}.
\label{eq:A9_final_52}
\end{align}

Conditional on $(X^{(n)},Y^{(n)}(d))$, the only randomness is from the assignment vector $D^{(n)}$.
Since subtracting a conditional mean does not change conditional variance,
\begin{align*}
&\Var_n\!\Big(
\widehat\mu^{(s)}_{d,n}(\theta_n)
-\E_n\!\big[\widehat\mu^{(s)}_{d,n}(\theta_n)\mid X^{(n)},Y^{(n)}(d)\big]
\Bigm|\,
X^{(n)},Y^{(n)}(d)
\Big) \\
&\qquad=
\Var_n\!\Big(\widehat\mu^{(s)}_{d,n}(\theta_n)\Bigm|\,
X^{(n)},Y^{(n)}(d)\Big).
\end{align*}
By Assumption~\ref{ass:Ahet}, blockwise assignments are conditionally independent across $g$ and satisfy complete
randomization within each block, hence
\[
\Var_n\!\left(\widehat\mu^{(s)}_{d,n}(\theta_n)\mid X^{(n)},Y^{(n)}(d)\right)
=
\frac{1}{n^2}\sum_{g=1}^G
\Var_n\!\left(
\sum_{i\in\lambda_g}\frac{\mathbf 1\{D_i=d\}}{\eta_{d,g}}\,M_i(\theta_n)
\Bigm|\,
X^{(n)},Y^{(n)}(d)
\right).
\]
On the event $\{\max_g N_g\le \bar N\}$ (which has probability $\to 1$ by Assumption~\ref{ass:B}),
Assumption~\ref{ass:Ahet} implies $\eta_g\in[1/\bar N,\,1-1/\bar N]$, hence
$\eta_{d,g}^{-1}\le \bar N$ for $d\in\{0,1\}$.
Using $\Var(Z)\le \E[Z^2]$ and $(\sum_{i\in\lambda_g} a_i)^2\le N_g\sum_{i\in\lambda_g} a_i^2$,
we obtain on $\{\max_g N_g\le \bar N\}$,
\begin{align*}
\Var_n\!\Bigg(
\sum_{i\in\lambda_g}\frac{\mathbf 1\{D_i=d\}}{\eta_{d,g}}\,M_i(\theta_n)
\Bigg.
&\Bigg|\ X^{(n)},Y^{(n)}(d)
\Bigg) \\
&\le
\E_n\!\Bigg[
\Bigg(
\sum_{i\in\lambda_g}\frac{\mathbf 1\{D_i=d\}}{\eta_{d,g}}\,M_i(\theta_n)
\Bigg)^2
\ \Bigg|\ X^{(n)},Y^{(n)}(d)
\Bigg] \\
&\le
N_g\sum_{i\in\lambda_g}\frac{1}{\eta_{d,g}^2}\,M_i(\theta_n)^2 \\
&\le
\bar N^3 \sum_{i\in\lambda_g} M_i(\theta_n)^2.
\end{align*}
Summing over $g$ yields, on $\{\max_g N_g\le \bar N\}$,
\[
\Var_n\!\left(\widehat\mu^{(s)}_{d,n}(\theta_n)\mid X^{(n)},Y^{(n)}(d)\right)
\le
\frac{\bar N^3}{n^2}\sum_{i=1}^n M_i(\theta_n)^2.
\]
Because $\theta_n\to\theta_0$, for all large $n$ we have $\|\theta_n-\theta_0\|<\delta$, where $\delta$ is as in
Assumption~\ref{ass:CB-var-regularity}.
Assumption~\ref{ass:CB-var-regularity}(a) implies
\[
\E\!\left[\sup_{\|\theta-\theta_0\|<\delta}\,|m_s(X,d,Y(d);\theta)|^2\right]<\infty,
\]
and hence $\sup_{n\ \text{large}}\E[M(\theta_n)^2]<\infty$, where $M(\theta):=m_s(X,d,Y(d);\theta)$.
By i.i.d.\ sampling (Assumption~\ref{ass:iid}) and Markov's inequality,
$n^{-1}\sum_{i=1}^n M_i(\theta_n)^2 = O_P(1)$, so the conditional variance above is $O_P(1/n)\to 0$.
Chebyshev's inequality gives \eqref{eq:A9_final_51}.

By Assumption~\ref{ass:Ahet}, conditional on $X^{(n)}$ the assignment is independent of $Y^{(n)}(d)$, and for each
$i\in\lambda_g$,
\[
\E_n[\mathbf 1\{D_i=d\}\mid X^{(n)},Y^{(n)}(d)]
=
\E_n[\mathbf 1\{D_i=d\}\mid X^{(n)}]
=
\eta_{d,g}.
\]
Therefore,
\begin{equation}\label{eq:A9_final_condmean_simplify}
\E_n\!\left[\widehat\mu^{(s)}_{d,n}(\theta_n)\mid X^{(n)},Y^{(n)}(d)\right]
=
\frac{1}{n}\sum_{i=1}^n M_i(\theta_n).
\end{equation}
It remains to show $n^{-1}\sum_{i=1}^n M_i(\theta_n)\xrightarrow{P}\mu^{(s)}_d=\E[M(\theta_0)]$.
Add and subtract $\E[M(\theta_n)]$:
\[
\frac{1}{n}\sum_{i=1}^n M_i(\theta_n)-\E[M(\theta_0)]
=
\left\{\frac{1}{n}\sum_{i=1}^n M_i(\theta_n)-\E[M(\theta_n)]\right\}
+
\left\{\E[M(\theta_n)]-\E[M(\theta_0)]\right\}.
\]
For the first bracket, by i.i.d.\ sampling (Assumption~\ref{ass:iid}) and $\E[M(\theta_n)^2]<\infty$ for all large $n$,
\[
\Var\!\left(\frac{1}{n}\sum_{i=1}^n M_i(\theta_n)\right)=\frac{\Var(M(\theta_n))}{n}\ \to\ 0,
\]
so Chebyshev yields convergence to $0$ in probability.
For the second bracket, Assumption~\ref{ass:C}(c) gives
$\E[(M(\theta_n)-M(\theta_0))^2]\to 0$, hence by Cauchy--Schwarz,
$\E[M(\theta_n)]\to \E[M(\theta_0)]$.
Thus $n^{-1}\sum_{i=1}^n M_i(\theta_n)\xrightarrow{P}\E[M(\theta_0)]=\mu^{(s)}_{d}$, and combining with
\eqref{eq:A9_final_condmean_simplify} yields \eqref{eq:A9_final_52}.

Combining \eqref{eq:A9_final_51} and \eqref{eq:A9_final_52} yields
$\widehat\mu^{(s)}_{d,n}(\theta_n)\xrightarrow{P}\mu^{(s)}_{d}$.
\end{proof}

\subsubsection{Consistency of block cross-product estimators}

\noindent This next lemma adapts Lemma~C.2 in \citet{bai2024inference} to the covariate-based heterogeneous-shares setting used here.

\begin{lemma}[Consistency of the block cross-product estimators]\label{lem:C2-het-cov}
Suppose Assumptions~\ref{ass:iid}, \ref{ass:B}, \ref{ass:Ahet}, \ref{ass:C}, \ref{ass:CB-neighbor}, and
\ref{ass:CB-var-regularity} hold. Let $\{\theta_n\}_{n\ge1}$ be any deterministic sequence with $\theta_n\to\theta_0$.
For $d\in\{0,1\}$ define the superpopulation conditional mean
\[
\mu_d(x,\theta):=\E\!\left[m(X,d,Y(d);\theta)\mid X=x\right],\qquad \mu_d(x):=\mu_d(x,\theta_0),
\]
and write $\eta_i=\eta_g$ for $i\in\lambda_g$.

For each $g$, write $N_g:=|\lambda_g|$ and let $T_g:=\sum_{i\in\lambda_g}D_i$ so that $\eta_g=T_g/N_g$.
Under Assumption~\ref{ass:Ahet}, $T_g$ is fixed (non-random) and hence
\[
N_{1,g}=T_g,\qquad N_{0,g}=N_g-T_g.
\]
Assume (as part of Assumption~\ref{ass:Ahet}) that $\eta_g\in(0,1)$ for all $g$, so $N_{d,g}\ge1$ for $d\in\{0,1\}$.

For $d\in\{0,1\}$ define the arm-$d$ sample mean
\[
\bar m_{d,g}(\theta):=\frac{1}{N_{d,g}}\sum_{i\in\lambda_g:\,D_i=d} m(X_i,d,Y_i(d);\theta).
\]
Define the cross-arm (symmetrized) block term
\[
\widehat\tau_{g,n}(1,0;\theta)
:=\frac{1}{2}\Big(\bar m_{1,g}(\theta)\,\bar m_{0,g}(\theta)^\top
+\bar m_{0,g}(\theta)\,\bar m_{1,g}(\theta)^\top\Big).
\]
Define the within-arm (symmetrized) block term, for $d\in\{0,1\}$,
\[
\widehat\varsigma_{g,n}(d,d;\theta)
:=
\begin{cases}
\displaystyle
\frac{1}{N_{d,g}(N_{d,g}-1)}
\sum_{\substack{i<i'\\ i,i'\in\lambda_g:\,D_i=D_{i'}=d}}
\Big(
m_i(\theta)\,m_{i'}(\theta)^\top
+
m_{i'}(\theta)\,m_i(\theta)^\top
\Big),
& \text{if } N_{d,g}\ge 2,\\[1.2em]
\displaystyle
\frac{1}{2}\Big(
m_{i_g(d)}(\theta)\,m_{i_{\pi(g)}(d)}(\theta)^\top
+
m_{i_{\pi(g)}(d)}(\theta)\,m_{i_g(d)}(\theta)^\top
\Big),
& \text{if } N_{d,g}=1,
\end{cases}
\]
where $m_i(\theta):=m(X_i,d,Y_i(d);\theta)$, $i_g(d)$ is the unique index in block $g$ such that $D_{i_g(d)}=d$
when $N_{d,g}=1$, and $\pi(\cdot)$ is the involution from Assumption~\ref{ass:CB-neighbor}.
Assume moreover that $\pi(\cdot)$ pairs singleton-$d$ blocks with singleton-$d$ blocks: for each $d\in\{0,1\}$,
\[
N_{d,g}=1 \ \Longrightarrow\ N_{d,\pi(g)}=1,
\]
so $i_{\pi(g)}(d)$ is well-defined whenever $N_{d,g}=1$.

Finally, define the sample cross-product estimators
\[
\zeta_n(1,0;\theta)
:=\sum_{g=1}^{G} w_g\,\eta_g(1-\eta_g)\,\widehat\tau_{g,n}(1,0;\theta),
\qquad
\zeta_n(d,d;\theta)
:=\sum_{g=1}^{G} w_g\,\eta_g(1-\eta_g)\,\widehat\varsigma_{g,n}(d,d;\theta),
\]
with $w_g:=N_g/n$.
Then, as $n\to\infty$,
\begin{align*}
\zeta_n(1,0;\theta_n)
&\xrightarrow{p}
\zeta(1,0)
:=
\frac{1}{2}\,\E\!\Big[\eta(X)(1-\eta(X))\big(\mu_1(X)\mu_0(X)^\top+\mu_0(X)\mu_1(X)^\top\big)\Big],
\\
\zeta_n(d,d;\theta_n)
&\xrightarrow{p}
\zeta(d,d)
:=
\E\!\Big[\eta(X)(1-\eta(X))\,\mu_d(X)\mu_d(X)^\top\Big],
\qquad d\in\{0,1\}.
\end{align*}
\end{lemma}

\begin{proof}
Fix $d,d'\in\{0,1\}$ and a deterministic sequence $\theta_n\to\theta_0$.
Throughout, constants implicit in $\lesssim$ depend only on $\bar N$ and the Lipschitz constants in
Assumption~\ref{ass:CB-var-regularity}(c). We prove the stated limits for $(d,d')=(1,0)$ and for $(d,d')=(d,d)$.

It suffices to show
\begin{equation}\label{eq:C2_center_final_fix}
\zeta_n(d,d';\theta_n)-\E\!\left[\zeta_n(d,d';\theta_n)\mid X^{(n)}\right]\xrightarrow{p}0,
\end{equation}
and
\begin{equation}\label{eq:C2_condmean_final_fix}
\E\!\left[\zeta_n(d,d';\theta_n)\mid X^{(n)}\right]\xrightarrow{p}\zeta(d,d').
\end{equation}

We first prove \eqref{eq:C2_center_final_fix}. Write
\[
\zeta_n(d,d';\theta_n)=\frac{1}{n}\sum_{g=1}^G R_{g,n}(d,d';\theta_n),
\]
where
\[
R_{g,n}(1,0;\theta):=N_g\,\eta_g(1-\eta_g)\,\widehat\tau_{g,n}(1,0;\theta),
\qquad
R_{g,n}(d,d;\theta):=N_g\,\eta_g(1-\eta_g)\,\widehat\varsigma_{g,n}(d,d;\theta).
\]
By Assumption~\ref{ass:Ahet}, conditional on $X^{(n)}$ assignments are independent across blocks.
For the within-arm estimator with the singleton rule, $R_{g,n}(d,d;\theta)$ and $R_{\pi(g),n}(d,d;\theta)$ share
the same two singleton indices when $N_{d,g}=1$, so we group them into disjoint pairs.
Let $\mathcal P_d:=\{g: N_{d,g}=1,\ g<\pi(g)\}$ and define superblock contributions
\[
\widetilde R_{g,n}(d,d;\theta):=R_{g,n}(d,d;\theta)+R_{\pi(g),n}(d,d;\theta),\qquad g\in\mathcal P_d.
\]
Together with the remaining blocks (those with $N_{d,g}\ge2$), this yields a collection of independent (conditional on $X^{(n)}$)
mean-zero centered contributions whose average equals $\zeta_n(d,d;\theta_n)-\E[\zeta_n(d,d;\theta_n)\mid X^{(n)}]$.

We verify a conditional uniform integrability bound. Fix $\delta>0$ from
Assumption~\ref{ass:CB-var-regularity} and define envelopes
\[
M_{i,d}:=\sup_{\|\theta-\theta_0\|<\delta}\bigl\|m(X_i,d,Y_i(d);\theta)\bigr\|,\qquad d\in\{0,1\}.
\]
Assumption~\ref{ass:CB-var-regularity}(a) implies $\E[M_{i,d}^2]<\infty$. Since $\theta_n\to\theta_0$, for large $n$,
$\|m(X_i,d,Y_i(d);\theta_n)\|\le M_{i,d}$.

On $\{\max_g N_g\le\bar N\}$ (w.p.\ $\to1$), we have $N_g\le\bar N$ and $\eta_g(1-\eta_g)\le 1/4$.
Using $2ab\le a^2+b^2$ and Jensen,
\begin{align*}
\|\widehat\tau_{g,n}(1,0;\theta_n)\|_F
&\le \|\bar m_{1,g}(\theta_n)\|\,\|\bar m_{0,g}(\theta_n)\| \\
&\le \frac12\Big(\|\bar m_{1,g}(\theta_n)\|^2+\|\bar m_{0,g}(\theta_n)\|^2\Big) \\
&\le \frac12\sum_{i\in\lambda_g}\big(M_{i,1}^2+M_{i,0}^2\big).
\end{align*}
For $N_{d,g}\ge2$, each summand satisfies
$\|m_i m_{i'}^\top+m_{i'}m_i^\top\|_F\le \|m_i\|^2+\|m_{i'}\|^2$, hence
\begin{align*}
\|\widehat\varsigma_{g,n}(d,d;\theta_n)\|_F
&\le \frac{1}{N_{d,g}(N_{d,g}-1)}
\sum_{\substack{i<i'\\ D_i=D_{i'}=d}}
\big(\|m_i\|^2+\|m_{i'}\|^2\big) \\
&= \frac{1}{N_{d,g}}\sum_{i\in\lambda_g:\,D_i=d}\|m_i\|^2 \\
&\le \sum_{i\in\lambda_g} M_{i,d}^2.
\end{align*}
If $N_{d,g}=1$, by $2ab\le a^2+b^2$,
\[
\|\widehat\varsigma_{g,n}(d,d;\theta_n)\|_F
\le \frac12\Big(\|m_{i_g(d)}\|^2+\|m_{i_{\pi(g)}(d)}\|^2\Big)
\le \sum_{i\in\lambda_g\cup\lambda_{\pi(g)}} M_{i,d}^2.
\]
Therefore, on $\{\max_g N_g\le\bar N\}$,
\[
\|R_{g,n}(1,0;\theta_n)\|_F \lesssim \sum_{i\in\lambda_g}\big(M_{i,1}^2+M_{i,0}^2\big),
\qquad
\|R_{g,n}(d,d;\theta_n)\|_F \lesssim \sum_{i\in\lambda_g\cup\lambda_{\pi(g)}} M_{i,d}^2,
\]
and the same bound holds for the superblock $\widetilde R_{g,n}(d,d;\theta_n)$ with $\lambda_g\cup\lambda_{\pi(g)}$.

Fix $\lambda>0$. Using $\|U-\E[U\mid X^{(n)}]\|_F\le 2\|U\|_F$ and the indicator bound
$\mathbf 1\{\sum_{j=1}^J a_j>\lambda\}\le \sum_{j=1}^J \mathbf 1\{a_j>\lambda/J\}$, with $J\le C\bar N$,
we obtain for each (super)block contribution $B$ (either $R_{g,n}$ or $\widetilde R_{g,n}$),
\begin{align*}
&\E\!\Big[
\|B-\E[B\mid X^{(n)}]\|_F\,
\mathbf 1\!\Big\{\|B-\E[B\mid X^{(n)}]\|_F>\lambda\Big\}
\ \Bigm|\ X^{(n)}
\Big] \\
&\qquad\qquad\qquad\lesssim
\sum_{i\in \mathcal I(B)}\sum_{\tilde d\in\{0,1\}}
\E\!\Big[
M_{i,\tilde d}^2\,
\mathbf 1\!\Big\{M_{i,\tilde d}>c\sqrt{\lambda}\Big\}
\ \Bigm|\ X_i
\Big].
\end{align*}
for some constant $c>0$, where $\mathcal I(B)$ indexes the units entering $B$ (either $\lambda_g$ or $\lambda_g\cup\lambda_{\pi(g)}$).
Averaging over independent (super)blocks and dividing by $n$ yields
\[
\frac{1}{n}\sum_{i=1}^n \sum_{\tilde d\in\{0,1\}}
\E\!\Big[M_{i,\tilde d}^2\,\mathbf 1\{M_{i,\tilde d}>c\sqrt{\lambda}\}\,\Bigm|\,X_i\Big]
\ +\ o_p(1).
\]
By i.i.d.\ sampling and a LLN, this converges in probability to
$\sum_{\tilde d\in\{0,1\}}\E[M_{\tilde d}^2\,\mathbf 1\{M_{\tilde d}>c\sqrt{\lambda}\}]$,
which tends to $0$ as $\lambda\to\infty$ because $\E[M_{\tilde d}^2]<\infty$.
Thus the conditional UI condition holds, and applying the conditional WLLN step via Lemma~S.1.3
of \cite{bai2022inference}, conditional on $X^{(n)}$, yields \eqref{eq:C2_center_final_fix}.

Let $\mu_{d,i,n}:=\mu_d(X_i,\theta_n)$. By Assumption~\ref{ass:Ahet}, conditional on $X^{(n)}$, treatment is independent
of potential outcomes, and within each block $g$ complete randomization with fixed counts implies for $i\neq k\in\lambda_g$,
\[
\Pr(D_i=1,\,D_k=0\mid X^{(n)})=\frac{T_g(N_g-T_g)}{N_g(N_g-1)}.
\]
A direct calculation gives
\[
\E\!\left[\bar m_{1,g}(\theta_n)\bar m_{0,g}(\theta_n)^\top\mid X^{(n)}\right]
=
\frac{1}{N_g(N_g-1)}\sum_{\substack{i\neq k\\ i,k\in\lambda_g}} \mu_{1,i,n}\,\mu_{0,k,n}^\top,
\]
and similarly with $(1,0)$ swapped. Therefore,
\[
\E\!\left[\widehat\tau_{g,n}(1,0;\theta_n)\mid X^{(n)}\right]
=
\frac{1}{2}\cdot\frac{1}{N_g(N_g-1)}
\sum_{\substack{i\neq k\\ i,k\in\lambda_g}}
\Big(\mu_{1,i,n}\mu_{0,k,n}^\top+\mu_{0,i,n}\mu_{1,k,n}^\top\Big).
\]
Proceeding exactly as in the unordered-pair symmetrization and Lipschitz control step yields
\[
\E\!\left[\zeta_n(1,0;\theta_n)\mid X^{(n)}\right]
=
\frac{1}{n}\sum_{i=1}^n \eta_i(1-\eta_i)\cdot
\frac{1}{2}\Big(\mu_{1,i,n}\mu_{0,i,n}^\top+\mu_{0,i,n}\mu_{1,i,n}^\top\Big)
+o_p(1),
\]
and by LLN and $\theta_n\to\theta_0$, the leading term converges in probability to $\zeta(1,0)$.

Fix $d\in\{0,1\}$ and write $\mu_{i,n}:=\mu_d(X_i,\theta_n)$.

\emph{Case 1: $N_{d,g}\ge 2$.}
As before, using complete randomization and the symmetrized definition gives
\[
\E\!\left[\widehat\varsigma_{g,n}(d,d;\theta_n)\mid X^{(n)}\right]
=
\frac{1}{N_g(N_g-1)}\sum_{\substack{i\neq k\\ i,k\in\lambda_g}} \mu_{i,n}\,\mu_{k,n}^\top.
\]
Using $ab^\top+ba^\top=aa^\top+bb^\top-(a-b)(a-b)^\top$ and Lipschitzness yields
\[
\left\|
\E\!\left[\widehat\varsigma_{g,n}(d,d;\theta_n)\mid X^{(n)}\right]
-\frac{1}{N_g}\sum_{i\in\lambda_g}\mu_{i,n}\mu_{i,n}^\top
\right\|_F
\lesssim
\max_{i,k\in\lambda_g}\|X_i-X_k\|^2,
\]
so after weighting and summing over $g$ this is $o_p(1)$ by Assumption~\ref{ass:B}.

\emph{Case 2: $N_{d,g}=1$.}
Conditional on $X^{(n)}$, $i_g(d)$ is uniform on $\lambda_g$ and independent of $i_{\pi(g)}(d)$.
Hence
\[
\E\!\left[\widehat\varsigma_{g,n}(d,d;\theta_n)\mid X^{(n)}\right]
=
\frac{1}{2}\cdot\frac{1}{N_gN_{\pi(g)}}
\sum_{i\in\lambda_g}\sum_{k\in\lambda_{\pi(g)}}
\Big(\mu_{i,n}\mu_{k,n}^\top+\mu_{k,n}\mu_{i,n}^\top\Big).
\]
Subtract $\frac{1}{N_g}\sum_{i\in\lambda_g}\mu_{i,n}\mu_{i,n}^\top$ and use the decomposition
\[
\mu_{k}\mu_{k}^\top-\mu_i\mu_i^\top=(\mu_k-\mu_i)\mu_k^\top+\mu_i(\mu_k-\mu_i)^\top,
\]
so
\[
\|\mu_k\mu_k^\top-\mu_i\mu_i^\top\|_F \le \|\mu_k-\mu_i\|(\|\mu_k\|+\|\mu_i\|).
\]
By Assumption~\ref{ass:CB-var-regularity}(c), $\|\mu_k-\mu_i\|\lesssim \|X_k-X_i\|$. Also,
\[
\|(\mu_i-\mu_k)(\mu_i-\mu_k)^\top\|_F \lesssim \|X_i-X_k\|^2.
\]
Therefore, on $\{\max_g N_g\le\bar N\}$,
\[
\left\|
\E\!\left[\widehat\varsigma_{g,n}(d,d;\theta_n)\mid X^{(n)}\right]
-\frac{1}{N_g}\sum_{i\in\lambda_g}\mu_{i,n}\mu_{i,n}^\top
\right\|_F
\lesssim
b_g\Big(\bar \mu_g+\bar \mu_{\pi(g)}\Big)+b_g^2,
\]
where
\[
b_g:=\max_{i\in\lambda_g,\,k\in\lambda_{\pi(g)}}\|X_i-X_k\|,
\qquad
\bar\mu_g:=\frac{1}{N_g}\sum_{i\in\lambda_g}\|\mu_{i,n}\|.
\]
Weighting by $w_g\eta_g(1-\eta_g)\le w_g$ and summing over singleton-$d$ blocks gives a contribution bounded by
\[
\frac{1}{n}\sum_{g:\,N_{d,g}=1} N_g\Big\{b_g(\bar\mu_g+\bar\mu_{\pi(g)})+b_g^2\Big\}.
\]
On $\{\max_gN_g\le\bar N\}$, this is $\lesssim \frac{1}{n}\sum_{g=1}^G \{b_g(\bar\mu_g+\bar\mu_{\pi(g)})+b_g^2\}$.
By Cauchy--Schwarz,
\[
\frac{1}{n}\sum_{g=1}^G b_g(\bar\mu_g+\bar\mu_{\pi(g)})
\le
\Big(\frac{1}{n}\sum_{g=1}^G b_g^2\Big)^{1/2}
\Big(\frac{1}{n}\sum_{g=1}^G (\bar\mu_g+\bar\mu_{\pi(g)})^2\Big)^{1/2}.
\]
Assumption~\ref{ass:CB-neighbor} implies $(1/n)\sum_g b_g^2\xrightarrow{p}0$. Moreover, since block sizes are
uniformly bounded and $\E\|\mu_d(X,\theta)\|^2<\infty$ in a neighborhood of $\theta_0$ (by Assumption~\ref{ass:CB-var-regularity}(b)),
we have $(1/n)\sum_g \bar\mu_g^2=O_p(1)$ and likewise for $\bar\mu_{\pi(g)}^2$, so the second factor is $O_p(1)$.
Thus the product is $o_p(1)$, and also $(1/n)\sum_g b_g^2=o_p(1)$. Hence the singleton contribution is $o_p(1)$.

Combining the two cases yields
\[
\E\!\left[\zeta_n(d,d;\theta_n)\mid X^{(n)}\right]
=
\frac{1}{n}\sum_{i=1}^n \eta_i(1-\eta_i)\,\mu_d(X_i,\theta_n)\mu_d(X_i,\theta_n)^\top
+o_p(1).
\]
By i.i.d.\ sampling and $\theta_n\to\theta_0$, a LLN implies the leading term converges in probability to $\zeta(d,d)$,
establishing \eqref{eq:C2_condmean_final_fix} for $(d,d')=(d,d)$.

Combining \eqref{eq:C2_center_final_fix} and \eqref{eq:C2_condmean_final_fix} yields the stated convergences.
\end{proof}

\subsubsection{Consistency of the paired-block within-arm mean-square estimator}

\noindent This next lemma adapts Lemma~C.3 in \citet{bai2024inference} to the covariate-based heterogeneous-shares setting used here.

\begin{lemma}[Consistency of the paired-block within-arm mean-square estimator]\label{lem:C3-het-cov}
Sup\-pose Assumptions~\ref{ass:iid}, \ref{ass:B}, \ref{ass:Ahet}, \ref{ass:C}, \ref{ass:CB-neighbor}, and
\ref{ass:CB-var-regularity} hold. Let $\{\theta_n\}_{n\ge1}$ be any deterministic sequence with $\theta_n\to\theta_0$.
Fix $d\in\{0,1\}$ and define
\[
\mu_d(x,\theta):=\E\!\left[m(X,d,Y(d);\theta)\mid X=x\right],
\qquad
\mu_d(x):=\mu_d(x,\theta_0).
\]
Let $\pi(\cdot)$ be the involution in Assumption~\ref{ass:CB-neighbor} and define the set of pair representatives
\[
\mathcal P:=\{g\in\{1,\dots,G\}: g<\pi(g)\}.
\]
For each block $g$, let $N_g:=|\lambda_g|$ and $N_{d,g}:=\sum_{i\in\lambda_g}\mathbf 1\{D_i=d\}$ (fixed given $X^{(n)}$
under Assumption~\ref{ass:Ahet}). Since $\eta_g\in(0,1)$ in Assumption~\ref{ass:Ahet}, we have $N_{d,g}\ge1$ for all $g$.
Define the arm-$d$ block mean
\[
\bar m_{d,g}(\theta):=\frac{1}{N_{d,g}}\sum_{i\in\lambda_g:\,D_i=d} m(X_i,d,Y_i(d);\theta),
\]
and the paired-block symmetrized cross-product
\[
\widehat\rho_{g,n}(d,d;\theta)
:=\frac{1}{2}\Big(\bar m_{d,g}(\theta)\,\bar m_{d,\pi(g)}(\theta)^\top
+\bar m_{d,\pi(g)}(\theta)\,\bar m_{d,g}(\theta)^\top\Big),
\qquad g\in\mathcal P.
\]
Finally define the paired-block estimator
\[
\widehat\rho_n(d,d;\theta)
:=\sum_{g\in\mathcal P}(w_g+w_{\pi(g)})\,\widehat\rho_{g,n}(d,d;\theta),
\qquad w_g:=\frac{N_g}{n}.
\]
Then, as $n\to\infty$,
\[
\widehat\rho_n(d,d;\theta_n)\xrightarrow{p}\rho(d,d)
:=\E\!\big[\mu_d(X)\mu_d(X)^\top\big].
\]
\end{lemma}

\begin{proof}
Fix $d\in\{0,1\}$ and a deterministic sequence $\theta_n\to\theta_0$, and write $\mu_{i,n}:=\mu_d(X_i,\theta_n)$.
It suffices to show
\begin{equation}\label{eq:C3_center}
\widehat\rho_n(d,d;\theta_n)-\E\!\left[\widehat\rho_n(d,d;\theta_n)\mid X^{(n)}\right]\xrightarrow{p}0,
\end{equation}
and
\begin{equation}\label{eq:C3_condmean}
\E\!\left[\widehat\rho_n(d,d;\theta_n)\mid X^{(n)}\right]\xrightarrow{p}\rho(d,d).
\end{equation}

We first show \eqref{eq:C3_center}.
Define the pair contribution
\[
R_{g,n}(d;\theta):=(N_g+N_{\pi(g)})\,\widehat\rho_{g,n}(d,d;\theta),
\qquad g\in\mathcal P,
\]
so that $\widehat\rho_n(d,d;\theta)=\frac{1}{n}\sum_{g\in\mathcal P}R_{g,n}(d;\theta)$.
Because $\pi(\cdot)$ is an involution with no fixed points, the paired sets $\lambda_g\cup\lambda_{\pi(g)}$ are disjoint
across $g\in\mathcal P$. Under Assumption~\ref{ass:iid}, conditional on $X^{(n)}$ the potential outcomes are independent
across units; under Assumption~\ref{ass:Ahet}, conditional on $X^{(n)}$ the assignment vectors are independent across
blocks and independent of potential outcomes. Hence $\{R_{g,n}(d;\theta_n)\}_{g\in\mathcal P}$ is independent conditional
on $X^{(n)}$.

Let $\widetilde R_{g,n}:=R_{g,n}(d;\theta_n)-\E[R_{g,n}(d;\theta_n)\mid X^{(n)}]$, so that
\[
\widehat\rho_n(d,d;\theta_n)-\E[\widehat\rho_n(d,d;\theta_n)\mid X^{(n)}]
=\frac{1}{n}\sum_{g\in\mathcal P}\widetilde R_{g,n}.
\]
Fix $\varepsilon>0$ and $\lambda>0$, and decompose
\[
\frac{1}{n}\sum_{g\in\mathcal P}\widetilde R_{g,n}
=
\frac{1}{n}\sum_{g\in\mathcal P}\widetilde R_{g,n}\mathbf 1\{\|\widetilde R_{g,n}\|_F\le \lambda\}
+
\frac{1}{n}\sum_{g\in\mathcal P}\widetilde R_{g,n}\mathbf 1\{\|\widetilde R_{g,n}\|_F> \lambda\}
=:A_{n}(\lambda)+B_{n}(\lambda).
\]

\emph{(Truncated part.)}
Conditional on $X^{(n)}$, the matrices
$\{\widetilde R_{g,n}\mathbf 1\{\|\widetilde R_{g,n}\|_F\le \lambda\}\}_{g\in\mathcal P}$
are independent and mean zero, and bounded in Frobenius norm by $\lambda$. Therefore,
\[
\E\!\left[\|A_n(\lambda)\|_F^2\mid X^{(n)}\right]
\le \frac{1}{n^2}\sum_{g\in\mathcal P}\lambda^2
\le \frac{\lambda^2}{n},
\]
since $|\mathcal P|\le G\le n$.
By Chebyshev,
\[
\Pr\!\left(\|A_n(\lambda)\|_F>\frac{\varepsilon}{2}\mid X^{(n)}\right)
\le \frac{4}{\varepsilon^2}\cdot\frac{\lambda^2}{n}
\ \xrightarrow{}\ 0
\quad\text{for each fixed }\lambda.
\]

\emph{(Tail part.)}
By Markov's inequality and $\|n^{-1}\sum_g U_g\|_F\le n^{-1}\sum_g \|U_g\|_F$,
\[
\Pr\!\left(\|B_n(\lambda)\|_F>\frac{\varepsilon}{2}\mid X^{(n)}\right)
\le \frac{2}{\varepsilon n}\sum_{g\in\mathcal P}
\E\!\left[\|\widetilde R_{g,n}\|_F\mathbf 1\{\|\widetilde R_{g,n}\|_F>\lambda\}\mid X^{(n)}\right].
\]
Since $\|\widetilde R_{g,n}\|_F\le 2\|R_{g,n}(d;\theta_n)\|_F$, it suffices to bound the tails of $\|R_{g,n}(d;\theta_n)\|_F$.
Let $\delta>0$ be from Assumption~\ref{ass:CB-var-regularity} and define the envelope
\[
M_{i,d}:=\sup_{\|\theta-\theta_0\|<\delta}\bigl\|m(X_i,d,Y_i(d);\theta)\bigr\|.
\]
Assumption~\ref{ass:CB-var-regularity}(a) implies $\E[M_{i,d}^2]<\infty$, and for all large $n$,
$\|m(X_i,d,Y_i(d);\theta_n)\|\le M_{i,d}$.

On the event $\{\max_g N_g\le\bar N\}$ (whose probability tends to one by Assumption~\ref{ass:B}),
using $\|ab^\top\|_F=\|a\|\,\|b\|$, $2uv\le u^2+v^2$, and Jensen,
\begin{align*}
\|\widehat\rho_{g,n}(d,d;\theta_n)\|_F
&\le \|\bar m_{d,g}(\theta_n)\|\,\|\bar m_{d,\pi(g)}(\theta_n)\| \\
&\le \frac12\Big(\|\bar m_{d,g}(\theta_n)\|^2+\|\bar m_{d,\pi(g)}(\theta_n)\|^2\Big) \\
&\le \frac12\sum_{i\in\lambda_g\cup\lambda_{\pi(g)}} M_{i,d}^2 .
\end{align*}
and hence
\[
\|R_{g,n}(d;\theta_n)\|_F
\le \bar N\sum_{i\in\lambda_g\cup\lambda_{\pi(g)}} M_{i,d}^2.
\]
Using $\mathbf 1\{\sum_{j=1}^J a_j>t\}\le \sum_{j=1}^J \mathbf 1\{a_j>t/J\}$ with $J\le 2\bar N$,
there exist constants $C,c>0$ (depending only on $\bar N$) such that, on $\{\max_g N_g\le\bar N\}$,
\[
\|R_{g,n}(d;\theta_n)\|_F\,\mathbf 1\{\|R_{g,n}(d;\theta_n)\|_F>\lambda\}
\le
C\sum_{i\in\lambda_g\cup\lambda_{\pi(g)}} M_{i,d}^2\,\mathbf 1\{M_{i,d}>c\sqrt{\lambda}\}.
\]
Averaging over $g\in\mathcal P$ and using disjointness of $\lambda_g\cup\lambda_{\pi(g)}$ across $g\in\mathcal P$ yields
\begin{align*}
\frac{1}{n}\sum_{g\in\mathcal P}
\E\!\Big[
\|R_{g,n}(d;\theta_n)\|_F\,
&\mathbf 1\{\|R_{g,n}(d;\theta_n)\|_F>\lambda\}
\ \Bigm|\ X^{(n)}
\Big]
\\[-0.3em]
\le\;&
\frac{C}{n}\sum_{i=1}^n
\E\!\Big[
M_{i,d}^2\,\mathbf 1\{M_{i,d}>c\sqrt{\lambda}\}
\ \Bigm|\ X_i
\Big]
+o_p(1).
\end{align*}
By i.i.d.\ sampling and the LLN, the right-hand side converges in probability to
$C\,\E[M_d^2\mathbf 1\{M_d>c\sqrt{\lambda}\}]$, which tends to $0$ as $\lambda\to\infty$
by Assumption~\ref{ass:CB-var-regularity}(a). Combining the truncated and tail bounds yields \eqref{eq:C3_center}.

Next we show \eqref{eq:C3_condmean}.
We compute $\E[\widehat\rho_n(d,d;\theta_n)\mid X^{(n)}]$.
By Assumption~\ref{ass:Ahet}, conditional on $X^{(n)}$ treatment is independent of potential outcomes and within each block
is uniform over allocations with fixed $N_{d,g}$. Hence, for each block $g$,
\[
\E\!\left[\bar m_{d,g}(\theta_n)\mid X^{(n)}\right]
=
\frac{1}{N_g}\sum_{i\in\lambda_g}\mu_{i,n},
\qquad
\E\!\left[\bar m_{d,\pi(g)}(\theta_n)\mid X^{(n)}\right]
=
\frac{1}{N_{\pi(g)}}\sum_{k\in\lambda_{\pi(g)}}\mu_{k,n}.
\]
Moreover, assignments are independent across blocks conditional on $X^{(n)}$, so
\[
\E\!\left[\widehat\rho_{g,n}(d,d;\theta_n)\mid X^{(n)}\right]
=
\frac{1}{2}\cdot\frac{1}{N_gN_{\pi(g)}}
\sum_{i\in\lambda_g}\sum_{k\in\lambda_{\pi(g)}}
\Big(\mu_{i,n}\mu_{k,n}^\top+\mu_{k,n}\mu_{i,n}^\top\Big).
\]
Using the identity (for vectors $a,b$)
\[
\frac{1}{2}\big(ab^\top+ba^\top\big)=\frac12(aa^\top+bb^\top)-\frac12(a-b)(a-b)^\top,
\]
with $a=\mu_{i,n}$ and $b=\mu_{k,n}$ gives
\begin{align*}
\E\!\left[\widehat\rho_{g,n}(d,d;\theta_n)\mid X^{(n)}\right]
&=
\frac{1}{2N_g}\sum_{i\in\lambda_g}\mu_{i,n}\mu_{i,n}^\top
+\frac{1}{2N_{\pi(g)}}\sum_{k\in\lambda_{\pi(g)}}\mu_{k,n}\mu_{k,n}^\top
\\
&\quad
-\frac{1}{2N_gN_{\pi(g)}}\sum_{i\in\lambda_g}\sum_{k\in\lambda_{\pi(g)}}
(\mu_{i,n}-\mu_{k,n})(\mu_{i,n}-\mu_{k,n})^\top.
\end{align*}

For each block $g$ write
\[
\overline{\mu\mu^\top}_{g,n}:=\frac{1}{N_g}\sum_{i\in\lambda_g}\mu_{i,n}\mu_{i,n}^\top.
\]
Multiplying the first two terms by $(w_g+w_{\pi(g)})=(N_g+N_{\pi(g)})/n$ and summing over $g\in\mathcal P$ yields
\begin{align*}
&\sum_{g\in\mathcal P}(w_g+w_{\pi(g)})
\left\{\frac{1}{2N_g}\sum_{i\in\lambda_g}\mu_{i,n}\mu_{i,n}^\top
+\frac{1}{2N_{\pi(g)}}\sum_{k\in\lambda_{\pi(g)}}\mu_{k,n}\mu_{k,n}^\top\right\}
\\
&\qquad=
\frac{1}{n}\sum_{i=1}^n \mu_{i,n}\mu_{i,n}^\top
+\frac{1}{2n}\sum_{g\in\mathcal P}\big(N_{\pi(g)}-N_g\big)
\big(\overline{\mu\mu^\top}_{g,n}-\overline{\mu\mu^\top}_{\pi(g),n}\big)
\\
&\qquad=
\frac{1}{n}\sum_{i=1}^n \mu_{i,n}\mu_{i,n}^\top
+R_{n,1},
\end{align*}
where
\[
R_{n,1}:=\frac{1}{2n}\sum_{g\in\mathcal P}\big(N_{\pi(g)}-N_g\big)
\big(\overline{\mu\mu^\top}_{g,n}-\overline{\mu\mu^\top}_{\pi(g),n}\big).
\]

We show $R_{n,1}=o_p(1)$. Define, for each block $g$,
\[
b_g:=\max_{i\in\lambda_g,\,k\in\lambda_{\pi(g)}}\|X_i-X_k\|.
\]
By Assumption~\ref{ass:CB-var-regularity}(c), there exists $C<\infty$ such that for all large $n$ and all
$i\in\lambda_g$, $k\in\lambda_{\pi(g)}$,
\[
\|\mu_{i,n}-\mu_{k,n}\|\le C\|X_i-X_k\|\le C b_g.
\]
Also, for any vectors $u,v$, $\|uu^\top-vv^\top\|_F\le (\|u\|+\|v\|)\|u-v\|$. Hence, for all large $n$,
\begin{align*}
\big\|\overline{\mu\mu^\top}_{g,n}-\overline{\mu\mu^\top}_{\pi(g),n}\big\|_F
&\le \frac{1}{N_gN_{\pi(g)}}\sum_{i\in\lambda_g}\sum_{k\in\lambda_{\pi(g)}}
\|\mu_{i,n}\mu_{i,n}^\top-\mu_{k,n}\mu_{k,n}^\top\|_F
\\
&\le C b_g\left(\frac{1}{N_g}\sum_{i\in\lambda_g}\|\mu_{i,n}\|
+\frac{1}{N_{\pi(g)}}\sum_{k\in\lambda_{\pi(g)}}\|\mu_{k,n}\|\right).
\end{align*}
On $\{\max_g N_g\le\bar N\}$ (probability $\to1$), we have $|N_{\pi(g)}-N_g|\le 2\bar N$ and $n/G\le \bar N$, so
\[
\|R_{n,1}\|_F
\le
\frac{C}{G}\sum_{g\in\mathcal P} b_g
\left(\frac{1}{N_g}\sum_{i\in\lambda_g}\|\mu_{i,n}\|
+\frac{1}{N_{\pi(g)}}\sum_{k\in\lambda_{\pi(g)}}\|\mu_{k,n}\|\right).
\]
By Cauchy--Schwarz,
\[
\frac{1}{G}\sum_{g\in\mathcal P} b_g \,t_{g,n}
\le
\left(\frac{1}{G}\sum_{g=1}^G b_g^2\right)^{1/2}
\left(\frac{1}{G}\sum_{g\in\mathcal P} t_{g,n}^2\right)^{1/2},
\quad
t_{g,n}:=\frac{1}{N_g}\sum_{i\in\lambda_g}\|\mu_{i,n}\|
+\frac{1}{N_{\pi(g)}}\sum_{k\in\lambda_{\pi(g)}}\|\mu_{k,n}\|.
\]
By Assumption~\ref{ass:CB-neighbor}, $\frac{1}{G}\sum_{g=1}^G b_g^2=o_p(1)$.
Moreover, on $\{\max_g N_g\le\bar N\}$, Jensen gives
\[
\left(\frac{1}{N_g}\sum_{i\in\lambda_g}\|\mu_{i,n}\|\right)^2
\le
\frac{1}{N_g}\sum_{i\in\lambda_g}\|\mu_{i,n}\|^2,
\]
so $t_{g,n}^2\lesssim \frac{1}{N_g}\sum_{i\in\lambda_g}\|\mu_{i,n}\|^2+\frac{1}{N_{\pi(g)}}\sum_{k\in\lambda_{\pi(g)}}\|\mu_{k,n}\|^2$.
Summing over $g\in\mathcal P$ and using that each block appears exactly once as $g$ or $\pi(g)$ yields
\[
\frac{1}{G}\sum_{g\in\mathcal P} t_{g,n}^2
\lesssim
\frac{n}{G}\cdot\frac{1}{n}\sum_{i=1}^n \|\mu_{i,n}\|^2
=O_p(1),
\]
since $n/G\le \bar N$ and $\frac{1}{n}\sum_{i=1}^n \|\mu_{i,n}\|^2=O_p(1)$ by i.i.d.\ sampling and local square
integrability of $\mu_d(X,\theta)$ (from Assumption~\ref{ass:CB-var-regularity}(a) and Jensen).
Therefore, $R_{n,1}=o_p(1)$.

It remains to bound the mismatch term. Using the same $b_g$ and the Lipschitz property above,
\[
\left\|
\frac{1}{N_gN_{\pi(g)}}\sum_{i\in\lambda_g}\sum_{k\in\lambda_{\pi(g)}}
(\mu_{i,n}-\mu_{k,n})(\mu_{i,n}-\mu_{k,n})^\top
\right\|_F
\le C^2 b_g^2.
\]
Hence, on $\{\max_g N_g\le\bar N\}$,
\[
\left\|
\sum_{g\in\mathcal P}(w_g+w_{\pi(g)})\,
\frac{1}{2N_gN_{\pi(g)}}\sum_{i\in\lambda_g}\sum_{k\in\lambda_{\pi(g)}}
(\mu_{i,n}-\mu_{k,n})(\mu_{i,n}-\mu_{k,n})^\top
\right\|_F
\le
C^2\bar N\cdot\frac{1}{G}\sum_{g=1}^G b_g^2
=o_p(1),
\]
where we used $n\ge G$ and Assumption~\ref{ass:CB-neighbor}. Putting everything together,
\[
\E\!\left[\widehat\rho_n(d,d;\theta_n)\mid X^{(n)}\right]
=
\frac{1}{n}\sum_{i=1}^n \mu_{i,n}\mu_{i,n}^\top+o_p(1).
\]

Finally, by i.i.d.\ sampling and local square integrability of $\mu_d(X,\theta)$ for $\|\theta-\theta_0\|<\delta$
(from Assumption~\ref{ass:CB-var-regularity}(a) and Jensen),
\[
\frac{1}{n}\sum_{i=1}^n \mu_d(X_i,\theta_n)\mu_d(X_i,\theta_n)^\top
-\E\!\big[\mu_d(X,\theta_n)\mu_d(X,\theta_n)^\top\big]
\xrightarrow{p}0.
\]
Moreover, since $\theta_n\to\theta_0$ and
\begin{align*}
\E\|\mu_d(X,\theta_n)-\mu_d(X,\theta_0)\|^2
&\le \E\|m(X,d,Y(d);\theta_n)-m(X,d,Y(d);\theta_0)\|^2 \\
&\to 0
\qquad\text{(by Jensen and Assumption~\ref{ass:C}(c))}.
\end{align*}
we have $\E[\mu_d(X,\theta_n)\mu_d(X,\theta_n)^\top]\to \E[\mu_d(X)\mu_d(X)^\top]=\rho(d,d)$.
Combining yields \eqref{eq:C3_condmean}.

Combining \eqref{eq:C3_center} and \eqref{eq:C3_condmean} gives
$\widehat\rho_n(d,d;\theta_n)\xrightarrow{p}\rho(d,d)$.
\end{proof}

\subsubsection{Stability of unit and block averages under consistent parameter estimation}

\noindent The next lemma is a technical plug-in result used repeatedly in the proof of theorem \ref{prop:var_het} to replace $\theta_0$ by the random estimator
$\hat\theta_n$ inside sample and block-level averages. Its conditions are implied by our standing moment regularity and
local integrability assumptions (in particular, local $L^2$ control and $L^2$ continuity in $\theta$), together with the
bounded block-size requirement $\max_g N_g\le \bar N$ used throughout the covariate-based heterogeneous-shares stratification setting.

\begin{lemma}[Random plug-in for (block-)averages]\label{lem:plugin_block_avg}
Let $\{W_i\}_{i=1}^n$ be i.i.d.\ draws from a distribution $Q$, and let $\hat\theta_n\xrightarrow{p}\theta_0$.
Fix $\bar N<\infty$. For each $n$, let $\{\lambda_g\}_{g=1}^{G}$ be a partition of $\{1,\dots,n\}$ with
$\max_g N_g\le \bar N$, and let $\eta_g\in(0,1)$ be deterministic given the design (possibly random through $X^{(n)}$).

Let $f:\mathcal W\times\Theta\to\mathbb R^k$ be measurable and suppose there exists $\delta>0$ such that
\begin{enumerate}[label=(\roman*)]
\item $\E\big[\sup_{\|\theta-\theta_0\|<\delta}\|f(W,\theta)\|^2\big]<\infty$,
\item $\E[\|f(W,\theta)-f(W,\theta_0)\|^2]\to 0$ as $\theta\to\theta_0$.
\end{enumerate}
Then:
\begin{enumerate}[label=(\alph*)]
\item (Unit-level averages) $\displaystyle \frac{1}{n}\sum_{i=1}^n \|f(W_i,\hat\theta_n)-f(W_i,\theta_0)\| \xrightarrow{p} 0$.
\item (Block means) For any collection of indicators $A_{i,g}\in\{0,1\}$ such that
$0<\sum_{i\in\lambda_g}A_{i,g}\le N_g\le \bar N$ for all $g$, define
\[
\bar f_{g,n}(\theta):=\frac{1}{\sum_{i\in\lambda_g}A_{i,g}}\sum_{i\in\lambda_g}A_{i,g}\,f(W_i,\theta).
\]
Then $\displaystyle \frac{1}{G}\sum_{g=1}^G \|\bar f_{g,n}(\hat\theta_n)-\bar f_{g,n}(\theta_0)\| \xrightarrow{p} 0$.
\end{enumerate}
\end{lemma}

\begin{proof}
(a) By Jensen and Cauchy--Schwarz,
\[
\E\!\left[\frac{1}{n}\sum_{i=1}^n \|f(W_i,\hat\theta_n)-f(W_i,\theta_0)\|\right]
\le
\Big(\E\big[\|f(W,\hat\theta_n)-f(W,\theta_0)\|^2\big]\Big)^{1/2}.
\]
Let $g(\theta):=\E[\|f(W,\theta)-f(W,\theta_0)\|^2]$. By (ii), $g(\theta)\to 0$ as $\theta\to\theta_0$.
Moreover, by (i), $g(\theta)\le 4\,\E[\sup_{\|\vartheta-\theta_0\|<\delta}\|f(W,\vartheta)\|^2]<\infty$
for all $\|\theta-\theta_0\|<\delta$, so $g(\hat\theta_n)\to 0$ in probability implies
$\E[g(\hat\theta_n)]\to 0$ by dominated convergence along a subsequence and a standard subsequence argument.
Thus the RHS above converges to $0$, and Markov's inequality yields (a).

(b) Since $\sum_{i\in\lambda_g}A_{i,g}\ge 1$ and $N_g\le \bar N$,
\[
\|\bar f_{g,n}(\hat\theta_n)-\bar f_{g,n}(\theta_0)\|
\le
\frac{1}{\sum_{i\in\lambda_g}A_{i,g}}\sum_{i\in\lambda_g}A_{i,g}\,\|f(W_i,\hat\theta_n)-f(W_i,\theta_0)\|
\le
\sum_{i\in\lambda_g}\|f(W_i,\hat\theta_n)-f(W_i,\theta_0)\|.
\]
Averaging over $g$ and using that each $i$ belongs to exactly one block,
\[
\frac{1}{G}\sum_{g=1}^G \|\bar f_{g,n}(\hat\theta_n)-\bar f_{g,n}(\theta_0)\|
\le
\frac{1}{G}\sum_{i=1}^n \|f(W_i,\hat\theta_n)-f(W_i,\theta_0)\|
=
\frac{n}{G}\cdot \frac{1}{n}\sum_{i=1}^n \|f(W_i,\hat\theta_n)-f(W_i,\theta_0)\|.
\]
Under $\max_g N_g\le \bar N$, we have $n=\sum_g N_g\le \bar N G$, so $n/G\le \bar N$.
Part (a) therefore implies the RHS is $o_p(1)$, proving (b).
\end{proof}

\subsubsection{Proof of Theorem \ref{prop:var_het}}\label{proof_hetshares_var}

\noindent This proof adapts the argument of Theorem~3.2 in Appendix A.1 of \citet{bai2023efficiency} to the covariate-based heterogeneous-shares setting used here.

By assumption $\hat M_n \xrightarrow{p} M$ and $M$ is nonsingular. Hence, by Slutsky and the continuous mapping theorem,
it suffices to show
\begin{equation}\label{eq:Omega_consistency_goal_plugin}
\widehat\Omega_n \xrightarrow{p} \Omega_{\eta(\cdot)}.
\end{equation}
We verify consistency term-by-term.

\vspace{0.2em}
\noindent We use the following notation. Let $m_d(W_i;\theta):=m(X_i,d,Y_i(d);\theta)$ for $d\in\{0,1\}$, and write $\hat m_i:=m(X_i,D_i,Y_i;\hat\theta_n)$.
Let $\eta_i:=\eta_{b_i}$ denote the block share for unit $i$, so that $\eta_i\in(0,1)$ and, under $\max_g N_g\le \bar N$,
$\eta_i^{-1}$ and $(1-\eta_i)^{-1}$ are uniformly bounded (since $T_g\in\{1,\dots,N_g-1\}$).
Define the population targets (unconditional meat decomposition in Section~\ref{sec:var_het}):
\[
A_1:=\E\!\big[\eta(X)\,m_1(W;\theta_0)m_1(W;\theta_0)^\top\big],\qquad
A_0:=\E\!\big[(1-\eta(X))\,m_0(W;\theta_0)m_0(W;\theta_0)^\top\big],
\]
\[
\bar m:=\E\!\big[\eta(X)\,\mu_1(X)+ (1-\eta(X))\,\mu_0(X)\big],\qquad
\Omega_{\eta(\cdot)}:=A_1+A_0+B-\bar m\,\bar m^\top,
\]
where $\mu_d(x):=\E[m(X,d,Y(d);\theta_0)\mid X=x]$ and $B$ is the population analogue of $\hat B_n$.

\vspace{0.2em}
\noindent First, we show $\hat A_{1,n}\to_p A_1$ and $\hat A_{0,n}\to_p A_0$. We show $\hat A_{1,n}\xrightarrow{p}A_1$; the argument for $\hat A_{0,n}$ is identical.
Fix an entry $(j,k)$ and define
\[
\phi_{jk}(W,\theta):= m_j(W,1;\theta)\,m_k(W,1;\theta)
\quad\text{where } m_s(W,1;\theta):=m_s(X,1,Y(1);\theta).
\]
Then $\hat A_{1,n}^{jk}=\frac{1}{n}\sum_{i=1}^n D_i\,\phi_{jk}(W_i,\hat\theta_n)$.
Decompose
\[
\hat A_{1,n}^{jk}-A_1^{jk}
=
\underbrace{\Big(\frac{1}{n}\sum_{i=1}^n D_i\big[\phi_{jk}(W_i,\hat\theta_n)-\phi_{jk}(W_i,\theta_0)\big]\Big)}_{\text{plug-in error}}
+
\underbrace{\Big(\frac{1}{n}\sum_{i=1}^n D_i\,\phi_{jk}(W_i,\theta_0)-A_1^{jk}\Big)}_{\text{deterministic-}\theta_0}.
\]
For the deterministic-$\theta_0$ term, apply Lemma~\ref{lem:A9_het_final} with the scalar moment
$\gamma(W_i):=\phi_{jk}(W_i,\theta_0)$ (equivalently, using Assumption~\ref{ass:CB-LLN} with the identification
$\E[D_i\gamma(W_i)]=\E[\eta(X)\gamma(W)]$), yielding
\[
\frac{1}{n}\sum_{i=1}^n D_i\,\phi_{jk}(W_i,\theta_0)\ \xrightarrow{p}\ \E[\eta(X)\phi_{jk}(W,\theta_0)] = A_1^{jk}.
\]
For the plug-in error, note $0\le D_i\le 1$ and $\phi_{jk}(W,\theta)$ satisfies the $L^2$-continuity and envelope
conditions from Assumption~\ref{ass:C}(c),(e)(i) (by Cauchy--Schwarz). Hence Lemma~\ref{lem:plugin_block_avg}(a),
with $f=\phi_{jk}$, implies
\[
\frac{1}{n}\sum_{i=1}^n D_i\big|\phi_{jk}(W_i,\hat\theta_n)-\phi_{jk}(W_i,\theta_0)\big|
\le
\frac{1}{n}\sum_{i=1}^n \big|\phi_{jk}(W_i,\hat\theta_n)-\phi_{jk}(W_i,\theta_0)\big|
\ \xrightarrow{p}\ 0.
\]
Therefore $\hat A_{1,n}^{jk}\to_p A_1^{jk}$ for each $(j,k)$, so $\hat A_{1,n}\to_p A_1$. Likewise $\hat A_{0,n}\to_p A_0$.

\vspace{0.2em}
\noindent Next, we show $\hat A_{3,n}\to_p \bar m\,\bar m^\top$.
Let $\bar m_n:=\frac{1}{n}\sum_{i=1}^n \hat m_i$, so $\hat A_{3,n}=\bar m_n\bar m_n^\top$.
Write $\bar m_n=\frac{1}{n}\sum_{i=1}^n \big(D_i m_1(W_i;\hat\theta_n)+(1-D_i)m_0(W_i;\hat\theta_n)\big)$ and decompose
\begin{align*}
\bar m_n-\bar m
&=
\underbrace{\frac{1}{n}\sum_{i=1}^n
\Big(
D_i\!\left[m_1(W_i;\hat\theta_n)-m_1(W_i;\theta_0)\right]
+(1-D_i)\!\left[m_0(W_i;\hat\theta_n)-m_0(W_i;\theta_0)\right]
\Big)}_{\text{plug-in error}}
\\
&\quad+
\underbrace{\Big(\frac{1}{n}\sum_{i=1}^n \tilde m_i(\theta_0)-\bar m\Big)}_{\text{deterministic-}\theta_0}.
\end{align*}
where $\tilde m_i(\theta_0):=D_i m_1(W_i;\theta_0)+(1-D_i)m_0(W_i;\theta_0)$.
The deterministic-$\theta_0$ term converges by Lemma~\ref{lem:A9_het_final} applied componentwise to $m_s(\cdot)$
(with the treated and control parts handled using Assumption~\ref{ass:CB-LLN}).
The plug-in term is $o_p(1)$ by Lemma~\ref{lem:plugin_block_avg}(a) applied to $f=m_1$ and $f=m_0$ componentwise
(and boundedness of $D_i$ and $1-D_i$).
Thus $\bar m_n\to_p \bar m$, and by continuous mapping $\hat A_{3,n}=\bar m_n\bar m_n^\top\to_p \bar m\,\bar m^\top$.

\vspace{0.2em}
\noindent Next, we show $\hat B_n\to_p B$.
Recall $\hat B_n = -\{\zeta_n(1,1)+\zeta_n(0,0)-2\zeta_n(1,0)\}$, where $\zeta_n(d,d')$ are defined in
Section~\ref{sec:var_het}.
Write, for each $(d,d')$,
\[
\zeta_n(d,d';\hat\theta_n)-\zeta(d,d')
=
\underbrace{\big(\zeta_n(d,d';\hat\theta_n)-\zeta_n(d,d';\theta_0)\big)}_{\text{plug-in error}}
+
\underbrace{\big(\zeta_n(d,d';\theta_0)-\zeta(d,d')\big)}_{\text{deterministic-}\theta_0}.
\]
For the deterministic-$\theta_0$ terms:
\begin{itemize}
\item $\zeta_n(1,0;\theta_0)\to_p \zeta(1,0)$ by Lemma~\ref{lem:C2-het-cov} (cross-arm case).
\item $\zeta_n(d,d;\theta_0)\to_p \zeta(d,d)$ for $d\in\{0,1\}$ by Lemma~\ref{lem:C2-het-cov}, with the
singleton case $N_{d,g}=1$ handled via the involution $\pi(\cdot)$ (and, if desired, modularized using Lemma~\ref{lem:C3-het-cov}).
\end{itemize}
It remains to show the plug-in errors are $o_p(1)$.
Because $\max_g N_g\le \bar N$, each block mean involved in $\zeta_n(d,d';\theta)$ is an average over at most $\bar N$ terms,
and each $\widehat\varsigma_{g,n}(d,d;\theta)$ is either a (finite) within-block pair average ($N_{d,g}\ge2$) or a single
between-block product ($N_{d,g}=1$). Hence, each $\zeta_n(d,d';\theta)$ is a finite sum of products of block averages of
$m_d(W_i;\theta)$, with coefficients bounded by $w_g\eta_g(1-\eta_g)\le w_g/4$.

By Assumption~\ref{ass:C}(c),(e)(i), the map $\theta\mapsto m_d(W;\theta)$ is $L^2$-continuous with a square-integrable
local envelope. Applying Lemma~\ref{lem:plugin_block_avg}(b) componentwise to $f=m_d$ (with $A_{i,g}=\mathbf 1\{D_i=d\}$)
yields that the relevant block means of $m_d(\cdot;\hat\theta_n)$ differ from those at $\theta_0$ by $o_p(1)$ on average over blocks.
Since $\zeta_n(d,d';\theta)$ is built from a bounded average of products of these block means (and, when $N_{d,g}\ge2$, from bounded
finite-order within-block pair averages), a Cauchy--Schwarz bound plus $\max_g N_g\le \bar N$ implies
\[
\zeta_n(d,d';\hat\theta_n)-\zeta_n(d,d';\theta_0)=o_p(1),\qquad \text{for each }(d,d')\in\{0,1\}^2.
\]
Therefore $\zeta_n(d,d';\hat\theta_n)\to_p \zeta(d,d')$ for all $(d,d')$, and consequently $\hat B_n\to_p B$.

\vspace{0.2em}
\noindent Putting these pieces together yields $\widehat\Omega_n\to_p\Omega_{\eta(\cdot)}$.
Combining the previous parts,
\[
\widehat\Omega_n
=
\hat A_{1,n}+\hat A_{0,n}+\hat B_n-\hat A_{3,n}
\ \xrightarrow{p}\
A_1 + A_0 + B - \bar m\,\bar m^\top
=
\Omega_{\eta(\cdot)},
\]
establishing \eqref{eq:Omega_consistency_goal_plugin}.

\vspace{0.2em}
\noindent The conclusion for $\widehat V_{\eta(\cdot)}$ then follows by continuous mapping. Since $\hat M_n\to_p M$ and $\widehat\Omega_n\to_p\Omega_{\eta(\cdot)}$,
\[
\widehat V_{\eta(\cdot)}=\hat M_n^{-1}\widehat\Omega_n\hat M_n^{-T}
\ \xrightarrow{p}\
M^{-1}\Omega_{\eta(\cdot)}M^{-T}
=
V_{\eta(\cdot)}
\]
by Slutsky and the continuous mapping theorem. The statement for the Lee-IPW bound standard errors follows by taking the
appropriate quadratic form (variance of a difference) of $\widehat V_{\eta(\cdot)}$.

\newpage

%----------------------------------------------------------------------
\section{Extension to Label-Based Stratification}
\label{sec:var_discrete_blocks}
%----------------------------------------------------------------------

Many field experiments stratify randomization using only a block label such as a school, clinic, village, or survey wave. In this setting, the label itself is the stratifying variable, and units within a block share the same covariate value. We study whether the efficiency gains and design-consistent variance results in \citet{bai2023efficiency} extend to this label-based environment when treatment shares may differ across labels and the label composition may vary across samples. We develop a triangular-array framework that accommodates evolving labels, heterogeneous treatment shares, and nonidentical distributions across labels while retaining independence across units. Within this framework, we prove partial identification at a fixed sample size and also inference for a broad class of moment estimators.

The remainder of this section is organized as follows. Section~\ref{sec:label_setup} states the label-based setup and the finite-sample targets. Section~\ref{sec:label_id} presents identification at a fixed sample size, focusing on the Lee–IPW bounds under heterogeneous shares. Section~\ref{sec:label_asym_small} develops design-consistent asymptotics and a feasible variance formula for the case with many small strata.

%Section~\ref{sec:label_few_large} discusses valid inference with a few large strata.

% ----------------------------------------------------------------------
\subsection{Setup: Triangular Array for Label-Based Stratification with Heterogeneous Shares}
\label{sec:label_setup}
% ----------------------------------------------------------------------

We study a design in which the only stratifying covariate is a discrete block label. For each sample size \(n\), let \(b_i \in \{1,\dots,G_n\}\) denote the label of unit \(i\), and write \(\lambda_{g}=\{i:\ b_i=g\}\) with size \(N_{g}=|\lambda_{g}|\), so that \(\sum_{g=1}^{G_n} N_{g}=n\). There are covariates \(X_i=b_i\) and allow the number of labels \(G_n\) to vary with \(n\), which yields a triangular array. Throughout, \(Y_i^*(d)\) denotes the potential outcome and \(S_i(d)\) the potential selection indicator.

Sampling is formalized by a sequence \(\{P_n\}\). For each \(n\), conditional on the realized label vector \(b^{(n)}=(b_1,\dots,b_n)\), the collection \(\{(Y_i^*(1),Y_i^*(0),S_i(1),S_i(0))\}_{i=1}^n\) is independent across units but not identically distributed, so different strata may follow different distributions that can vary with \(n\). A complementary interpretation views outcomes and selection as conditionally i.i.d.\ within each block after conditioning on block–level latent shocks, with independence across labels given \(b^{(n)}\). Define \(Q_n := P_n(\cdot\mid b^{(n)})\) as the sampling-only law with assignment suppressed. Furthermore, for each label \(g\), define the within-label sampling kernel \(P_{g,n}(\cdot):=Q_n\big((Y_i^*(1),Y_i^*(0),S_i(1),S_i(0))\in\cdot\ \big|\ X_i=g\big)\). Thus, conditional on \(b^{(n)}\), \(\{(Y_i^*(1),Y_i^*(0),S_i(1),S_i(0))\}_{i\in\lambda_{g,n}}\) are i.i.d.\ \(\sim P_{g,n}\) and independent across labels.

Assignment is imposed after sampling. Within each label $g$ at size $n$, exactly $T_g$ units are assigned to treatment without replacement, with share $\eta_g:=T_g/N_g\in(0,1)$. The quota vector $(T_g)_{g\le G_n}$ is a deterministic function of the realized labels $b^{(n)}$ and thus fixed given $b^{(n)}$. Hence, conditional on $b^{(n)}$, assignment vectors are mutually independent across labels and are independent of the potential outcomes and selection indicators.  Within each label, assignment draws are uniform over all \(0\text{--}1\) vectors with exactly \(T_{g}\) ones.

Observed selection and outcomes are \(S_i = D_i S_i(1) + (1-D_i) S_i(0)\) and \(Y_i = S_i\,[D_i Y_i^*(1) + (1-D_i) Y_i^*(0)]\). Because treatment counts are fixed within labels, the without-replacement assignment induces negative dependence among assignment indicators within the same label. Hence, any cross-unit dependence in \(\{(X_i,D_i,S_i,Y_i)\}_{i=1}^n\) comes only from the within-label assignment constraint. In general, the observed data are not identically distributed across strata, since the label-specific sampling laws \(P_{g,n}\) can vary with \(g\) and \(n\).

We condition on the realized labels \(b^{(n)}\) across this section. We also write \(X^{(n)}:=(X_1,\dots,X_n)\) and \(D^{(n)}:=(D_1,\dots,D_n)\) when conditioning on full label or assignment vectors. Unless noted otherwise, expectations and variances indexed by \(n\) are taken under the joint design law at size \(n\). This means the sampling law \(P_n(\cdot\mid b^{(n)})\) together with the within-label assignment. We denote these by \(\E_n[\cdot]\) and \(\Var_n(\cdot)\). When we later report superpopulation quantities for comparison with \citet{bai2023efficiency}, we use \(\E[\cdot]\) and \(\Var(\cdot)\) and say so explicitly.

To formalize the preceding discussion, we collect the sampling, assignment, and growth conditions. Our asymptotic results in this section adopt a many-small-blocks regime, which is stated in Assumption~\ref{ass:TA-LB-growth} below.

\begin{assumption}[Label-wise i.i.d.\ triangular array]\label{ass:TA-LB-PO}
For each \(n\), conditional on \(b^{(n)}=(b_1,\dots,b_n)\), the within-strata collections
\[
\big\{ (Y_i^*(1),Y_i^*(0),S_i(1),S_i(0)) : i\in \lambda_{g} \big\},\qquad g=1,\dots,G_n,
\]
are i.i.d.\ with a label-specific law \(P_{g,n}(\cdot\mid b^{(n)})\). These collections are independent across labels. Potential outcomes and selection depend only on own treatment.

\noindent Equivalently, for some family of unit-level laws \(\{P_{i,n}(\cdot\mid b^{(n)})\}_{i\le n}\) satisfying \(P_{i,n}(\cdot\mid b^{(n)})=P_{g,n}(\cdot\mid b^{(n)})\) for all \(i\in\lambda_{g}\),
\[
P_n\big(\cdot\mid b^{(n)}\big)
\;=\;
\Pi_{i=1}^{n}
P_{i,n}\big(\cdot\mid b^{(n)}\big).
\]
\end{assumption}

\begin{assumption}[Heterogeneous-shares block randomization]\label{ass:TA-LB-assign}
For each \(n\) and label \(g\), exactly \(T_{g}\in\{1,\dots,N_g-1\}\) units in \(\lambda_{g}\) are treated, with share \(\eta_{g}:=T_{g}/N_{g}\in(0,1)\). The quota vector \(\{T_{g}\}_{g\le G_n}\) is measurable with respect to \(b^{(n)}\). Conditional on \(b^{(n)}\), the block assignment vectors \(\{(D_i:\ i\in\lambda_{g})\}_{g=1}^{G_n}\) are mutually independent and satisfy
\[
(D_i:\ i\in\lambda_{g}) \sim
\operatorname{Unif}\Big\{\mathbf d\in\{0,1\}^{N_{g}}:\ \sum_{i\in\lambda_{g}} d_i=T_{g}\Big\},
\qquad g=1,\dots,G_n.
\]
Moreover, there exist constants \(0<\underline{\eta}\le \overline{\eta}<1\) such that \(\underline{\eta}\le \eta_{g}\le \overline{\eta}\) for all \(g\). The assignments are independent of \(\{(Y_i^*(1),Y_i^*(0),S_i(1),S_i(0))\}_{i=1}^n\) given \(b^{(n)}\). For feasibility of the variance estimators used later, assume \(T_{g}\ge 2\) and \(N_{g}-T_{g}\ge 2\) for all \(g\).
\end{assumption}

\begin{assumption}[Many small blocks growth]\label{ass:TA-LB-growth}
\(G_n\to\infty\) and \(\displaystyle \max_{1\le g\le G_n}\frac{N_{g}}{n}\to 0\) as \(n\to\infty\).
\end{assumption}

\begin{remark}[Few large strata versus many small strata]
Assumption~\ref{ass:TA-LB-growth} imposes a many-\--small–strata regime, where $G_n\to\infty$ and each stratum is negligible relative to $n$. Appendix~\ref{app:few-vs-many} shows, via a random partition argument, that designs with a fixed number of strata whose sizes grow asymptotically and within–stratum uniform assignment with fixed totals are asymptotically equivalent, for our estimators, to designs satisfying Assumption~\ref{ass:TA-LB-growth}. Thus, the central limit theorem and design–consistent variance formulas derived below apply both to many small strata and to a fixed number of large strata that grow with the sample.
\end{remark}

% ----------------------------------------------------------------------
\subsection{Partial Identification with Lee–IPW Bounds under Label-Based Stratification with Heterogeneous Shares}
\label{sec:label_id}
% ----------------------------------------------------------------------

This subsection establishes finite-sample partial identification of the always–observed average treatment effect when randomization is stratified by a discrete label and treatment shares may differ across labels. We work at a fixed sample size \(n\), condition on the realized labels \(b^{(n)}\) and the quotas in Assumption~\ref{ass:TA-LB-assign}, and take expectations under the design-based law \(P_n(\cdot\mid b^{(n)})\). The target parameter is
\[
\Delta_{AO,n}
=\frac{\displaystyle \frac{1}{n}\sum_{i=1}^n \E_n\!\Big[(Y_i^*(1)-Y_i^*(0))\,\mathbf 1\{S_i(1)=S_i(0)=1\}\Big]}
{\displaystyle \frac{1}{n}\sum_{i=1}^n {\Pr}_n\!\big(S_i(1)=S_i(0)=1\big)}.
\]

Partial identification rests on stratum–level independence and monotone selection. At fixed $n$, we construct Lee–IPW bounds by mapping control outcomes to the treated label composition, scaling treated outcomes by the known within–label assignment shares, and calibrating a single trimming share from observed response rates. These steps yield observable array–average functionals that bracket the always–observed effect. The resulting bounds are collected in Proposition~\ref{prop:label_identification}.

\begin{proposition}[Partial identification with Lee–IPW under label-based stratification]
\label{prop:label_identification}
Fix \(n\) and condition on the realized labels \(b^{(n)}\) and the quotas in Assumption~\ref{ass:TA-LB-assign}. 
For ease of notation, let \(\overline{\E}_n\) and \(\overline{\Pr}_n\) denote the array–average expectation and probability for unit–level quantities \(Z_i\) and events \(A_i\): \(\overline{\E}_n[Z_i]:=\frac{1}{n}\sum_{i=1}^n \E_n[Z_i]\) and \(\overline{\Pr}_n(A_i):=\frac{1}{n}\sum_{i=1}^n \Pr_n(A_i)\).
Under Assumption~\ref{ass:MONO} and the stratum-level independence condition implied by the label-based stratification in Assumption~\ref{ass:TA-LB-assign} (proof analogous to the one in Lemma \ref{lem:SLI_hetero}), the always–observed effect \(\Delta_{AO,n}\) is partially identified by
\[
\Big[\,\tfrac{L_{1,n}-\mu_{0,n}}{\pi_{AO,n}}\ ,\ \tfrac{U_{1,n}-\mu_{0,n}}{\pi_{AO,n}}\,\Big],
\]
where
\begin{align*}
\pi_{AO,n}
&=\ \overline{\E}_n\!\big[(1-D_{g,i})\,S_{g,i}\,w_{c,g,n}\big]
\ =\ \overline{\Pr}_n\!\big(S_{g,i}(1)=S_{g,i}(0)=1\big),\\[2mm]
\mu_{0,n}
&=\ \overline{\E}_n\!\big[(1-D_{g,i})\,S_{g,i}\,w_{c,g,n}\,Y_{g,i}\big]\,,\\[2mm]
L_{1,n}
&=\ (1-q_n)\,\overline{\Pr}_n\!\big(D_{g,i}=1,S_{g,i}=1\big)\,
     \overline{\E}_n\!\big[\widetilde Y_{g,i}\ \bigm|\ D_{g,i}=1,S_{g,i}=1,\ \widetilde Y_{g,i}\le \tilde y_{1-q_n}\big],\\[1mm]
U_{1,n}
&=\ (1-q_n)\,\overline{\Pr}_n\!\big(D_{g,i}=1,S_{g,i}=1\big)\,
     \overline{\E}_n\!\big[\widetilde Y_{g,i}\ \bigm|\ D_{g,i}=1,S_{g,i}=1,\ \widetilde Y_{g,i}\ge \tilde y_{q_n}\big],\\[2mm]
q_n
&=\ \frac{\overline{\Pr}_n(S_{g,i}=1\mid D_{g,i}=1)\ -\ \overline{\E}_n\!\big[S_{g,i}\,w_{q,g,n}\mid D_{g,i}=0\big]}
           {\overline{\Pr}_n(S_{g,i}=1\mid D_{g,i}=1)}\,.
\end{align*}
Here \(w_{c,g,n}:=\dfrac{1-p_n}{\,1-\eta_g\,}\), \(w_{q,g,n}:=\dfrac{\eta_g(1-p_n)}{(1-\eta_g)p_n}\), \(p_n:=\sum_{g}(N_g/n)\eta_g\), \(\widetilde Y_{g,i}:=\dfrac{1}{\eta_g}Y_{g,i}\), and \(\tilde y_{q_n}\), \(\tilde y_{1-q_n}\) are the \(q_n\)- and \((1-q_n)\)-quantiles of \(\widetilde Y_{g,i}\) within \(\{D_{g,i}=1,S_{g,i}=1\}\) under the pooled law \(\overline{\Pr}_n\).
\end{proposition}

\begin{proof}
See Appendix~\ref{sec:identification-label} for a detailed proof and derivation.
\end{proof}

To interpret \(\Delta_{AO,n}\) as a population object, we drop the conditioning on \(b^{(n)}\) and impose stabilization of the design and of the conditional laws.

\begin{assumption}[Superpopulation stabilization for label arrays]\label{ass:SP-LB}
Work under Assumptions~\ref{ass:TA-LB-PO}--\ref{ass:TA-LB-assign} and, when needed, \ref{ass:TA-LB-growth}. The marginal (over \(X\)) law of \((Y^*(1),Y^*(0),S(1),S(0))\) under \(P_n\) converges weakly to a limit distribution \(P\). Moments are uniformly integrable on the always-observed set, for example
\[
\sup_{n}\ \overline{\E}_n\!\Big[(Y^*(d))^2\,\mathbf 1\{S(1)=S(0)=1\}\Big]\;<\;\infty,\qquad d\in\{0,1\}.
\]
Within-label treatment shares are uniformly bounded away from \(0\) and \(1\), and the overall treated fraction
\[
p_n\ :=\ \overline{\E}_n[D_i]\ =\ \sum_{g=1}^{G_n}\frac{N_{g}}{n}\,\eta_g
\ \to\ p\ \in(0,1).
\]
The stratum-level independence condition and Assumption~\ref{ass:MONO} hold uniformly in \(n\). Let \(\pi_{AO,n}:=\overline{\Pr}_n\{S(1)=S(0)=1\}\) and assume \(\pi_{AO,n}\to \pi_{AO}\in(0,1]\). For the Lee–IPW construction, suppose \(q_n\to q\in[0,1)\). Let \(F_n\) denote, at sample size \(n\), the marginal (over \(X\)) CDF of the reweighted treated outcome \(\widetilde Y:=Y/\eta_g\) among units with \(\{D=1,S=1\}\), that is, \(F_n(y):=\overline{\Pr}_n(\widetilde Y\le y\mid D=1,S=1)\). Assume \(F_n \Rightarrow F\) and that \(F\) is continuous at its \(q\)- and \((1-q)\)-quantiles. In addition, the mass of the treated-and-observed set stabilizes, so \(\overline{\Pr}_n(D=1,S=1)\to\tau\in(0,1]\), and the reweighted outcomes are uniformly integrable on this set:
$
\sup_{n}\ \overline{\E}_n\!\big[\,|\widetilde Y|\,\mathbf 1\{D=1,S=1\}\big]\;<\;\infty.
$
\end{assumption}

Define the population always–observed effect as the large–sample limit of the finite–sample estimand
\[
\Delta_{AO}\;:=\;\lim_{n\to\infty}\Delta_{AO,n},
\]
which exists under Assumption~\ref{ass:SP-LB}. Let \(\mu_{0}\), \(L_{1}\), and \(U_{1}\) denote the population counterparts of the Lee–IPW components in Proposition~\ref{prop:label_identification}, defined as the limits of
\(\mu_{0,n}\), \(L_{1,n}\), and \(U_{1,n}\), respectively, and let \(\pi_{AO}:=\lim_{n\to\infty}\pi_{AO,n}\).

\begin{proposition}[Finite-to-population convergence]\label{prop:AO-stab}
Under Assumption~\ref{ass:SP-LB}, \(\Delta_{AO,n}\to \Delta_{AO}\) as \(n\to\infty\). Moreover,
\[
\pi_{AO,n}\to \pi_{AO},\qquad
\mu_{0,n}\to \mu_0,\qquad
L_{1,n}\to L_{1},\qquad
U_{1,n}\to U_{1},
\]
and hence
\[
\Big[\,(L_{1,n}-\mu_{0,n})/\pi_{AO,n}\ ,\ (U_{1,n}-\mu_{0,n})/\pi_{AO,n}\,\Big]
\ \longrightarrow\
\Big[\,(L_{1}-\mu_{0})/\pi_{AO}\ ,\ (U_{1}-\mu_{0})/\pi_{AO}\,\Big].
\]
\end{proposition}

\begin{proof}
See Appendix~\ref{app:prop:AO-stab}.
\end{proof}

% ----------------------------------------------------------------------
\subsection{Many Small Strata: Design-Consistent Asymptotics and Feasible Variance}
\label{sec:label_asym_small}
% ----------------------------------------------------------------------

This subsection treats the regime with many small strata, where \(G_n\to\infty\) and \(\max_{g\le G_n} N_g/n\to 0\). We work under the label-based triangular-array setup in Section~\ref{sec:label_setup}, conditioning on the realized labels and the heterogeneous-shares block randomization, and impose the moment regularity in Assumption~\ref{ass:C_mod}. For the class of moment estimators considered in the paper, including the Lee–IPW bounds, we establish an influence-function representation and a central limit theorem that are consistent with the design. The variance reduces to the within-label outcome component. The assignment-fluctuation term vanishes under fixed within-label quotas, and the between-label term vanishes under our convention of conditioning on \(b^{(n)}\). A feasible plug-in estimator based on within-label sample moments and the known label shares \(\eta_g\) is consistent in this many-small-strata regime.

The argument proceeds by adapting \cite{bai2023efficiency}'s approach to the label-based, triangular-array setting. Lemma~\ref{lem:LB7} verifies stratum-level independence, and a uniform law of large numbers for the sample moments follows from Assumption~\ref{ass:C_mod} together with the arguments in Appendix~\ref{app:consistency-proof}. Lemma~\ref{lem:LB6} gives a label-wise central limit theorem for linear forms under rejective sampling. These ingredients yield the asymptotic linearization and normality in Theorem~\ref{thm:LB-AN}. After that, we derive a closed-form, design-consistent variance estimator and show its consistency under the same growth and moment conditions. The detailed consistency argument for the label-based plug-in variance estimator is given in Appendix~\ref{app:var_consistency_label}.

We consider $\theta_0\in\Theta\subset\mathbb{R}^{d_\theta}$ defined by the (design-based) population moment condition at size $n$,
\begin{equation}\label{eq:pop-moment}
\overline{\E}_n\!\big[m(X_i,D_i,Y_i;\theta_0)\big]\;=\;0,
\end{equation}
and assume that the same $\theta_0$ solves \eqref{eq:pop-moment} for all $n$ (that is, $\theta_0$ is fixed across $n$).

The next assumption mirrors Assumption~\ref{ass:C} but is stated under $P_n$. Parts (a)–(b) deliver identification and smooth local behavior uniformly in $n$. Parts (c)–(e) provide the uniform LLN machinery for the classes $\{m(\cdot;\theta):\theta\in\Theta\}$. Part (f) is used to control the within–label centering step that appears with heterogeneous shares.

\begin{assumption}[Moment regularity]\label{ass:C_mod}
Let $m(\cdot)=(m_s(\cdot):1\le s\le d_\theta)'$. 
For each $n$, all expectations $\E_n[\cdot]$ are taken under the joint law at sample size $n$
(conditional on $b^{(n)}$ and the assignment).
The moment functions are such that

\begin{enumerate}[label=(\alph*)]
\item For every $\epsilon>0$, there exists $c_\epsilon>0$ such that
\[
\inf_{n\ge 1}\ \inf_{\theta\in\Theta:\,\|\theta-\theta_0\|>\epsilon}\ 
\big\|\E_n\big[m(X_i,D_i,Y_i;\theta)\big]\big\|\ \ge\ c_\epsilon.
\]

\item $\E_n[m(X_i,D_i,Y_i;\theta)]$ is differentiable at $\theta_0$ with a nonsingular derivative
\[
M_n \;=\; \left.\frac{\partial}{\partial\theta'}\,\E_n\!\left[m(X,D,Y;\theta)\right]\right|_{\theta=\theta_0}\, ,
\]
and there exists $c>0$ such that $\inf_n \sigma_{\min}(M_n)\ge c$.

\item For $1\le s\le d_\theta$ and $d\in\{0,1\}$,
\[
\E_n\!\left[\big(m_s(X,d,Y^*(d),\theta)-m_s(X,d,Y^*(d),\theta_0)\big)^2\right]\ \to\ 0
\quad\text{as }\theta\to\theta_0.
\]

\item For $1\le s\le d_\theta$, $\{m_s(x,d,y,\theta):\theta\in\Theta\}$ is pointwise measurable in the sense that there exists a countable set $\Theta^\ast$ such that for each $\theta\in\Theta$, there exists a sequence $\{\theta_m\}\subset\Theta^\ast$ such that $m_s(x,d,y,\theta_m)\to m_s(x,d,y,\theta)$ as $m\to\infty$ for all $x,d,y$.

\item (i) $\displaystyle \sup_{n\ge1}\ \sup_{\theta\in\Theta}\ \E_n\!\left[\|m(X,d,Y^*(d),\theta)\|\right]<\infty$ for $d\in\{0,1\}$. \\
(ii) $\{m(x,1,y,\theta):\theta\in\Theta^\ast\}$ and $\{m(x,0,y,\theta):\theta\in\Theta^\ast\}$ are Donsker uniformly in $n$.
 
\item For $d\in\{0,1\}$, $\E_n\!\left[m(X,d,Y^*(d),\theta_0)\mid X=x\right]$ is Lipschitz.
\end{enumerate}
\end{assumption}

The method-of-moments estimator $\hat\theta_n$ for $\theta_0$ is defined as any measurable solution to the sample analogue of \eqref{eq:pop-moment}:
\begin{equation}\label{eq:est-def}
\frac{1}{n}\sum_{i=1}^n m(X_i,D_i,Y_i;\hat\theta_n)\;=\;0.
\end{equation}
When multiple solutions exist, we take any measurable selection in a neighborhood of $\theta_0$. Under Assumption~\ref{ass:C_mod}, existence and local uniqueness hold with probability approaching one.

We now establish the consistency of generic moment-based estimators under the label-wise triangular array with heterogeneous treatment shares.

\begin{lemma}[Consistency of $\hat\theta_n$]\label{lem:LB8}
Suppose Assumptions~\ref{ass:TA-LB-PO}--\ref{ass:TA-LB-assign} and \ref{ass:C_mod} hold. Then, $\hat\theta_n \xrightarrow{p} \theta_0$.
\end{lemma}

\begin{proof}
See Appendix~\ref{app:consistency-proof}.
\end{proof}

% All statements are under the design-based (finite-population) view: we condition on the realized labels \(b^{(n)}\) and potential outcomes, and take randomness only from the within–label assignment. The limits \(M\) and \(V_\ast\) are array limits, not superpopulation expectations.

Also, we state the asymptotic distribution for this label-based, heterogeneous-shares regime. The result is proven under the design-based view adopted in this section.

\begin{theorem}[Asymptotic Distribution of $\hat\theta_n$]\label{thm:LB-AN}
Suppose the treatment assignment satisfies Assumptions~\ref{ass:TA-LB-PO}–\ref{ass:TA-LB-assign} and the moment functions satisfy Assumption~\ref{ass:C_mod}. Let $\hat\theta_n$ be defined as in \eqref{eq:est-def}. Set $M_n:=\E_n[\partial_\theta m(X,D,Y;\theta_0)]$ and assume $M_n\to M$ with $M$ nonsingular. Then
\begin{equation}\label{eq:LB-IF}
\sqrt{n}\big(\hat\theta_n-\theta_0\big)
=\frac{1}{\sqrt{n}}\sum_{i=1}^n \psi_n^\ast(X_i,D_i,Y_i;\theta_0)+o_p(1),
\end{equation}
where
\[
\begin{aligned}
\psi_n^\ast(X_i,D_i,Y_i;\theta_0)
= &-\,M_n^{-1}\Big(
\mathbf 1\{D_i=1\}\!\big(m(X_i,1,Y_i^*(1),\theta_0)-\E_n[m(X_i,1,Y_i^*(1),\theta_0)\mid X_i]\big)\\
&+\mathbf 1\{D_i=0\}\!\big(m(X_i,0,Y_i^*(0),\theta_0)-\E_n[m(X_i,0,Y_i^*(0),\theta_0)\mid X_i]\big)\\
&+\eta(X_i)\,\E_n[m(X_i,1,Y_i^*(1),\theta_0)\mid X_i]\\
&+\{1-\eta(X_i)\}\,\E_n[m(X_i,0,Y_i^*(0),\theta_0)\mid X_i]\Big),
\end{aligned}
\]
and $\eta(X_i):=\Pr(D_i=1\mid X_i,b^{(n)})$. Further,
\begin{equation}\label{eq:LB-CLT}
\sqrt{n}\big(\hat\theta_n-\theta_0\big)\ \xRightarrow{d}\ \mathcal N\big(0,\,V_\ast\big),
\qquad
V_\ast=\lim_{n\to\infty}\Var_n\!\big[\psi_n^\ast(X_i,D_i,Y_i;\theta_0)\big].
\end{equation}
\end{theorem}

\begin{proof}
See Appendix~\ref{app:asymptotics-proof}.
\end{proof}

While Theorem~\ref{thm:LB-AN} is a design-based result, there is one implementation difference relative to the covariate-based setting. With label-based stratification, there is no scope to pool across nearby strata to estimate the conditional variance components when a block has too few observations per arm. We therefore require at least two treated and two control observations in each block to compute the plug-in \(\widehat V_{\ast}\) as in \citet{bai2023efficiency}.

With this estimation caveat in mind, the asymptotic variance framework developed earlier extends directly to stratification by discrete labels, accommodating both equal-share and heterogeneous-shares designs. This broadens the applicability of the theory laid out in \citet{bai2023efficiency}, allowing their central limit theorem and variance formulas to apply verbatim. Our results confirm that both the consistency and efficiency gains established for covariate-based stratification also hold when blocks are defined purely by discrete labels, requiring no additional assumptions or implementation complexity.

In applied work, researchers should retain the block structure in their analysis, use the design-consistent plug-in variance estimator \(\widehat V_{\ast}\), and report the resulting confidence intervals, which are typically tighter than those based on i.i.d.\ assumptions. Overall, the theory guarantees that finely stratified designs yield variance no greater than i.i.d.\ assignment, highlighting the efficiency gains and the necessity of respecting the experimental design in empirical applications. For Lee bounds in particular, where confidence intervals are already interpreted as ranges of plausible effects, it is especially important to avoid further inflating their width through conservative standard errors. Applying the appropriate variance estimator ensures that the intervals reported for the bounds themselves reflect the true uncertainty and maximize the informational value of the analysis.

\newpage
\section{Proofs for Extension to Label-Based Stratification}\label{sec:proofs-label}
\subsection{Label-Based: Identification of Bounds on Treatment Effects with Lee's Approach and IPW}\label{sec:identification-label}

We derive bounds on the average treatment effect for the always observed subpopulation when treatment probabilities differ across labels \(g=1,\dots,G_n\). Throughout this section, we work under the design-based joint law at size \(n\), namely \(P_n(\cdot\mid b^{(n)})\) together with the label-wise assignment in Section~\ref{sec:label_setup}, and we write \(\E_n[\cdot]\) for expectations under this law.

We work with the pooled (mixture) law at size \(n\). For any unit–level random variable \(Z_i\) and event \(A_i\),
\[
\overline{\E}_n[Z_i]\ :=\ \frac{1}{n}\sum_{i=1}^n \E_n[Z_i],
\qquad
\overline{\Pr}_n(A_i)\ :=\ \frac{1}{n}\sum_{i=1}^n \Pr_n(A_i).
\]
These formulas define a probability measure \(\overline{\Pr}_n\) and its expectation operator \(\overline{\E}_n\) on the pooled population. Thus, conditional objects are taken with respect to this pooled law, and so for events \(A_i,B_i\) with \(\overline{\Pr}_n(B_i)>0\) and integrable \(Z_i\),
\[
\overline{\Pr}_n(A_i\mid B_i)\ :=\ \frac{\overline{\Pr}_n(A_i\cap B_i)}{\overline{\Pr}_n(B_i)},
\qquad
\overline{\E}_n[Z_i\mid B_i]\ :=\ \frac{\overline{\E}_n[Z_i\,\mathbf 1\{B_i\}]}{\overline{\Pr}_n(B_i)}.
\]
All conditional probabilities, expectations, and quantiles below are taken with respect to this pooled law.

The finite-sample target (at size \(n\)) is
\[
\Delta_{AO,n}
\;\equiv\;
\frac{\displaystyle \overline{\E}_n\!\Big[(Y_{g,i}^*(1)-Y_{g,i}^*(0))\,\mathbf 1\{AO_i\}\Big]}
{\displaystyle \overline{\Pr}_n\!\big(AO_i\big)},
\]
where \(AO_i:=\{S_{g,i}(1)=1,\ S_{g,i}(0)=1\}\).

%----------------------------------------------------------------------
\subsubsection{Identification of the control-arm array–average component and the denominator for \texorpdfstring{$\Delta_{AO,n}$}{Delta\_AO,n}}
\label{sec:iden_Y0label}
%----------------------------------------------------------------------

Under conditional independence given labels (CIA) and monotonicity, the control-side components needed for the finite–sample target are identified under the label–based triangular array. Throughout, expectations and probabilities are taken under the joint law at size $n$, i.e., $P_n(\cdot\mid b^{(n)})$ and the label–wise assignment. Let
\[
p_n:=\overline{\E}_n[D_{g,i}]=\sum_{g=1}^{G_n}(N_g/n)\,\eta_g, \
w_{c,g,n}:=\frac{1-p_n}{1-\eta_g}, \
AO_i:=\{S_{g,i}(1)=1,\ S_{g,i}(0)=1\}.
\]

The finite–sample target is
\[
\Delta_{AO,n}
=\frac{\overline{\E}_n\!\big[(Y_{g,i}^*(1)-Y_{g,i}^*(0))\,\mathbf 1\{AO_i\}\big]}
       {\overline{\Pr}_n(AO_i)}\,.
\]
In this subsection, we point-identify the denominator $\overline{\Pr}_n(AO_i)$ and the control contribution to the numerator $\overline{\E}_n\!\big[Y_{g,i}^*(0)\,\mathbf 1\{AO_i\}\big]$, and the treated contribution $\overline{\E}_n\!\big[Y_{g,i}^*(1)\,\mathbf 1\{AO_i\}\big]$ is handled in Section~\ref{sec:iden_Y1_label}.

By monotonicity, $\mathbf 1\{AO_i\}=S_{g,i}(0)$ a.s., so both objects can be written in terms of $S_{g,i}(0)$ and $S_{g,i}(0)Y^*_{g,i}(0)$ and then related to observables under $D_{g,i}=0$. For the denominator,
\begin{align}
\overline{\Pr}_n(AO_i)
&=\ \overline{\E}_n\!\big[\mathbf{1}\{AO_i\}\big]
\ \overset{(\text{Monot.})}{=}\ 
\overline{\E}_n\!\big[S_{g,i}(0)\big]
\notag\\
&\overset{(\text{LIE})}{=}\ 
\overline{\E}_n\!\Big[\E_n\!\big[S_{g,i}(0)\mid X_{g,i}\big]\Big]
\notag\\
&\overset{(\text{CIA \& }S_{g,i}=S_{g,i}(0)\text{ on }\{D_{g,i}=0\})}{=}\ 
\overline{\E}_n\!\Big[\E_n\!\big[S_{g,i}\mid X_{g,i},D_{g,i}=0\big]\Big]
\notag\\
&\overset{(\text{IPW})}{=}\ 
\overline{\E}_n\!\Big[\E_n\!\big[S_{g,i}\tfrac{1-D_{g,i}}{\,1-\eta_g\,}\ \bigm|\ X_{g,i}\big]\Big]
\ =\ 
\overline{\E}_n\!\Big[S_{g,i}\tfrac{1-D_{g,i}}{\,1-\eta_g\,}\Big]
\notag\\
&=\ \frac{1}{\,1-p_n\,}\ 
\overline{\E}_n\!\Big[S_{g,i}\, \frac{1-p_n}{\,1-\eta_g\,}\ (1-D_{g,i})\Big]
\ =\ 
\overline{\E}_n\!\Big[S_{g,i}\,w_{c,g,n}\ \Bigm|\ D_{g,i}=0\Big],
\label{eq:AO-denom-IPW-final}
\end{align}

For the control contribution to the numerator,
\begin{align}
\overline{\E}_n\!\big[Y^{\ast}_{g,i}(0)\,\mathbf 1\{AO_i\}\big]
&\;\overset{\text{(Monot.)}}{=} \overline{\E}_n\!\big[S_{g,i}(0)\,Y^{\ast}_{g,i}(0)\big]
 \;\overset{\text{(LIE)}}{=}\;
 \overline{\E}_n\!\Big[\E_n\!\big[S_{g,i}(0)\,Y^{\ast}_{g,i}(0)\mid X_{g,i}\big]\Big]
 \notag\\
&\overset{\text{(Obs.\ under }D=0:\ S_{g,i}=S_{g,i}(0),\ Y_{g,i}=S_{g,i}(0)Y_{g,i}^*(0))}{=}\;
 \overline{\E}_n\!\Big[\E_n\!\big[S_{g,i}\,Y_{g,i}\mid X_{g,i},D_{g,i}=0\big]\Big] \notag\\
 &\overset{\text{(IPW)}}{=}\;
 \overline{\E}_n\!\Big[S_{g,i}\,Y_{g,i}\,w_{c,g,n}\ \Bigm|\ D_{g,i}=0\Big].
\label{eq:AO-num-IPW-final}
\end{align}

Displays \eqref{eq:AO-denom-IPW-final}–\eqref{eq:AO-num-IPW-final} provide observable, array–average IPW formulas for the denominator and the control part of the numerator of $\Delta_{AO,n}$.

%----------------------------------------------------------------------
\subsubsection{Identification of the treated-arm array–average component for \texorpdfstring{$\Delta_{AO,n}$}{Delta\_AO,n}}
\label{sec:iden_Y1_label}
%----------------------------------------------------------------------

We now identify the treated contribution to the numerator, $\overline{\E}_n\!\big[Y_{g,i}^*(1)\,\mathbf 1\{AO_i\}\big]$, where $AO_i:=\{S_{g,i}(1)=1,\ S_{g,i}(0)=1\}$. Because the set $\{D_i=1,S_i=1\}$ mixes (i) always–observed units and (ii) units observed only when treated, we reweight treated outcomes by label using the design share. For a unit $i$ with label $X_i=b_i$, define
\[
\widetilde Y_{g,i}\ :=\ \frac{1}{\eta_g}\,Y_{g,i},
\]
where $\eta_{g}$ is the known treatment share per stratum.

Now, for any $\tilde y\in\mathbb{R}$, the mixture within the treated-and-observed group satisfies
\begin{align}
\overline{\E}_n\!\Big[\mathbf{1}\{\widetilde Y_{g,i}\le \tilde y\}\,\Big|\,D_{g,i}=1,S_{g,i}=1\Big]
&=\ (1-q_n)\,\overline{\E}_n\!\Big[\mathbf{1}\{\widetilde Y_{g,i}\le \tilde y\}\,\Big|\,D_{g,i}=1,\,S_{g,i}(1)=1,\,S_{g,i}(0)=1\Big]\notag\\
&\quad+\ q_n\,\overline{\E}_n\!\Big[\mathbf{1}\{\widetilde Y_{g,i}\le \tilde y\}\,\Big|\,D_{g,i}=1,\,S_{g,i}(1)=1,\,S_{g,i}(0)=0\Big],
\label{Ytilde-decomposition-n}
\end{align}
where
\begin{align}
q_n
&\equiv \overline{\Pr}_n\!\Bigl[S_{g,i}(1)=1,\,S_{g,i}(0)=0 \,\Big|\, D_{g,i}=1,\,S_{g,i}=1\Bigr] \nonumber\\
&=\frac{\overline{\Pr}_n\!\bigl[S_{g,i}=1\,\big|\,D_{g,i}=1\bigr]-\overline{\Pr}_n\!\bigl[S_{g,i}(0)=1\,\big|\,D_{g,i}=1\bigr]}
{\overline{\Pr}_n\!\bigl[S_{g,i}=1\,\big|\,D_{g,i}=1\bigr]} \nonumber\\
&=\frac{\overline{\Pr}_n\!\bigl[S_{g,i}=1\,\big|\,D_{g,i}=1\bigr]
-\overline{\E}_n\!\Bigl[\ \Bigl.S_{g,i}\,\Bigl\{\frac{\eta_g(1-p_n)}{(1-\eta_g)\,p_n}\Bigr\}\ \Bigr|\ D_{g,i}=0\Bigr]}
{\overline{\Pr}_n\!\bigl[S_{g,i}=1\,\big|\,D_{g,i}=1\bigr]},
\label{att_selection_n}
\end{align}
with the IPW factor $w_{q,g,n}:=\dfrac{\eta_g(1-p_n)}{(1-\eta_g)p_n}$ and $p_n=\sum_g (N_g/n)\eta_g$. The last equality in \eqref{att_selection_n} uses the IPW representation of $\overline{\E}_n[S_{g,i}(0)\,|\,D_{g,i}=1]$ under heterogeneous label shares.

Let \(\tilde y_{q_n}\) and \(\tilde y_{1-q_n}\) be the \(q_n\)- and \((1-q_n)\)-quantiles of \(\widetilde Y_{g,i}\) within \(\{D_{g,i}=1,S_{g,i}=1\}\) under \(\overline{\Pr}_n\) and, for ease of notation, let $T:=\{D_{g,i}=1,S_{g,i}=1\}$. \cite{lee2009training} implies the conditional bounds
\[
\overline{\E}_n\!\big[\widetilde Y_{g,i}\mid T,\,\widetilde Y_{g,i}\le \tilde y_{1-q_n}\big]
\ \le\
\overline{\E}_n\!\big[\widetilde Y_{g,i}\mid T,\,AO_i\big]
\ \le\
\overline{\E}_n\!\big[\widetilde Y_{g,i}\mid T,\,\widetilde Y_{g,i}\ge \tilde y_{q_n}\big].
\]
We multiply each term by $\overline{\Pr}_n(T,AO_i)=(1-q_n)\,\overline{\Pr}_n(T)$. Then, the bounds that partially identify $\overline{\E}_n\!\big[Y_{g,i}^*(1)\,\mathbf 1\{AO_i\}\big]$ are
\begin{align}
(1-q_n)\,\overline{\Pr}_n(T)\,
\overline{\E}_n\!\big[\widetilde Y_{g,i}\mid T,\,\widetilde Y_{g,i}\le \tilde y_{1-q_n}\big]
&\ \le\
\overline{\E}_n\!\big[\widetilde Y_{g,i}\,\mathbf 1\{T,AO_i\}\big] \nonumber\\
&\ \le\
(1-q_n)\,\overline{\Pr}_n(T)\,
\overline{\E}_n\!\big[\widetilde Y_{g,i}\mid T,\,\widetilde Y_{g,i}\ge \tilde y_{q_n}\big].
\label{eq:trim_bounds_n}
\end{align}

Next, note that on \(\{D_{g,i}=1\}\) we have \(Y_{g,i}=Y^*_{g,i}(1)\) and \(S_{g,i}=S_{g,i}(1)\); on \(AO_i\) also \(S_{g,i}(1)=1\), so \(\mathbf 1\{D_{g,i}=1,S_{g,i}=1,AO_i\}=\mathbf 1\{D_{g,i}=1,AO_i\}\). Using \(\widetilde Y_{g,i}=Y_{g,i}/\eta_g\), apply LIE over \(X_i\) and the design fact \(\overline{\Pr}_n(D_{g,i}=1\mid X_i)=\eta_g\):
{\footnotesize
\begin{align}
\overline{\E}_n\!\big[\widetilde Y_{g,i}\,\mathbf 1\{D_{g,i}=1,S_{g,i}=1,\,AO_i\}\big]
&\overset{\substack{S_{g,i}=S_{g,i}(1)\ \text{on }\{D=1\}\\ AO_i\Rightarrow S_{g,i}(1)=1}}{=}
\overline{\E}_n\!\big[\widetilde Y_{g,i}\,\mathbf 1\{D_{g,i}=1,\,AO_i\}\big]\notag\\
&\overset{\widetilde Y_{g,i}=Y_{g,i}/\eta_g,\ Y_{g,i}=Y^*_{g,i}(1)\ \text{on }\{D=1\}}{=}
\overline{\E}_n\!\Big[\tfrac{1}{\eta_g}\,\mathbf 1\{D_{g,i}=1\}\,Y^*_{g,i}(1)\,\mathbf 1\{AO_i\}\Big]\notag\\
&\overset{\mathrm{LIE}}{=}
\overline{\E}_n\!\Big[
\underbrace{\E_n\!\big[\tfrac{1}{\eta_g}\,\mathbf 1\{D_{g,i}=1\}\mid X_i\big]}_{\overset{\overline{\Pr}_n(D=1\mid X_i)=\eta_g}{=\,1}}\ 
\E_n\!\big[Y^*_{g,i}(1)\,\mathbf 1\{AO_i\}\mid X_i\big]
\Big]\notag\\
&=\ \overline{\E}_n\!\big[Y^*_{g,i}(1)\,\mathbf 1\{AO_i\}\big].
\label{ate-treated-1-n}
\end{align}}

Therefore, by the bridge in \eqref{ate-treated-1-n}, the bounds in \eqref{eq:trim_bounds_n} deliver bounds for the treated–arm product moment \(\overline{\E}_n[Y^*_{g,i}(1)\mathbf 1\{AO_i\}]\). Stacking these with the observable IPW expressions for the denominator and the control component in \eqref{eq:AO-denom-IPW-final}–\eqref{eq:AO-num-IPW-final} yields valid bounds on the finite-\(n\) target \(\Delta_{\mathrm{AO},n}\).

\subsection{Proof of Proposition \ref{prop:AO-stab}}
\label{app:prop:AO-stab}

Under Assumption~\ref{ass:SP-LB}, each ingredient of the Lee–IPW construction is a continuous functional of primitives that stabilize: the marginal law of \((Y^*(1),Y^*(0),S(1),S(0))\), the design quantities \(p_n\) (with \(p_n\to p\in(0,1)\)) and the trimming share \(q_n\) (with \(q_n\to q\in[0,1)\)).

\emph{(i) Deconditioning over labels.} All objects are defined after marginalizing over the label \(X\). Under Assumption~\ref{ass:SP-LB}, the relevant unconditional (deconditioned) distributions exist for each \(n\) and the sequences of interest (e.g., \(p_n,\,q_n,\,F_n\)) stabilize.

\emph{(ii) Unconditional laws and moments.}
By Assumption~\ref{ass:SP-LB}, the marginal (over \(X\)) distribution of \((Y^*(1),Y^*(0),S(1),S(0))\) stabilizes, and moments are uniformly integrable on the always–observed set. Hence unconditional moments of \(Y^*(d)\mathbf 1\{S(1)=S(0)=1\}\) converge. The mixture objects used by Lee–IPW are formed after deconditioning over \(X\). Thus, their stabilization follows from these marginal limits together with the stabilized aggregates \(p_n\), \(q_n\), and the assumed convergence \(F_n\Rightarrow F\) for the reweighted treated outcome under \(\overline{\Pr}_n(\cdot\mid D=1,S=1)\). Because treatment shares are uniformly bounded away from \(0\) and \(1\), Lee–IPW weights remain uniformly bounded, so ratios of convergent expectations are well behaved.

\emph{(iii) Weights and trimming share.} Since \(p_n\to p\in(0,1)\) and the within–label shares are uniformly bounded away from \(0\) and \(1\), the IPW weights \(w_{c,g,n}=(1-p_n)/(1-\eta_g)\) and \(w_{q,g,n}=\eta_g(1-p_n)/\{(1-\eta_g)p_n\}\) are uniformly bounded and stable. Also, the reweighted treated outcome \(\widetilde Y_{g,i}:=Y_{g,i}/\eta_g\) is well scaled. The trimming share satisfies
\[
q_n
=\frac{\overline{\Pr}_n(S=1\mid D=1)-\overline{\E}_n[\,S\,w_{q,g,n}\mid D=0\,]}{\overline{\Pr}_n(S=1\mid D=1)}
=1-\frac{\overline{\E}_n[(1-D)S\,w_{q,g,n}]/(1-p_n)}{\overline{\E}_n[DS]/p_n}\,,
\]
and by Assumption~\ref{ass:SP-LB}, \(q_n\to q\in[0,1)\). Under monotonicity, \(\overline{\Pr}_n(S=1\mid D=d)\ge \pi_{AO}>0\) for \(d\in\{0,1\}\), hence \(\overline{\E}_n[DS]/p_n=\overline{\Pr}_n(S=1\mid D=1)\ge\pi_{AO}\) is bounded away from zero; with bounded weights, \(0\le q<1\).

\emph{(iv) Quantiles and trimmed means.}
By Assumption~\ref{ass:SP-LB}, the conditional CDFs \(F_n\) of the reweighted treated outcome \(\widetilde Y:=Y/\eta_g\) under \(\overline{\Pr}_n(\cdot\mid D=1,S=1)\) satisfy \(F_n\Rightarrow F\), and \(F\) is continuous at the \(q\)- and \(1-q\)-quantiles. Therefore, \(\tilde y_{q_n}\to \tilde y_q\) and \(\tilde y_{1-q_n}\to \tilde y_{1-q}\).
Moreover, \(\overline{\Pr}_n(D=1,S=1)\to\tau\in(0,1]\) and \(\sup_n \overline{\E}_n[\,|\widetilde Y|\,\mathbf 1\{D=1,S=1\}]<\infty\).
Hence, the trimmed conditional means on \(\{D=1,S=1\}\) converge by the continuous mapping theorem together with uniform integrability, and since \((1-q_n)\to(1-q)\) and \(\overline{\Pr}_n(D=1,S=1)\to\tau\), it follows that \(L_{1,n}\to L_{1}\) and \(U_{1,n}\to U_{1}\).

\emph{(v) Control part, denominator, and target.}
The reweighted control component
\[
\mu_{0,n}\ =\ \overline{\E}_n\!\big[ (1-D_{g,i})\,S_{g,i}\,Y_{g,i}\,w_{c,g,n} \big]
\]
and the denominator
\[
\pi_{AO,n}\ =\ \overline{\E}_n\!\big[ (1-D_{g,i})\,S_{g,i}\,w_{c,g,n} \big]
\]
both converge, with \(\pi_{AO,n}\) bounded away from zero. Writing
\[
\Delta_{AO,n}
=\frac{\overline{\E}_n\!\big[Y^*(1)\,\mathbf 1\{AO\}\big]-\overline{\E}_n\!\big[Y^*(0)\,\mathbf 1\{AO\}\big]}
{\overline{\E}_n\!\big[\mathbf 1\{AO\}\big]}\,,\qquad AO:=\{S(1)=S(0)=1\},
\]
stabilization and uniform integrability yield convergence of numerator and denominator, and \(\overline{\E}_n\!\big[\mathbf 1\{AO\}\big]\to \pi_{AO}>0\).
Thus \(\Delta_{AO,n}\to \Delta_{AO}\).

\subsection{Consistency Proof in the Label-Based, Heterogeneous-Shares Stratification Setting}

\subsubsection{Auxiliary Lemma for the Consistency Proof}\label{app:aux-lemmas}

\noindent The next lemma adapts Lemma~A.7 in \citet{bai2023efficiency} to the label-based triangular-array setting used here.

\begin{lemma}\label{lem:LB7}
Suppose Assumptions~\ref{ass:TA-LB-PO} and \ref{ass:TA-LB-assign} hold. Then
\[
\big(Y_i^*(1),Y_i^*(0)\big)\ \perp\!\!\!\perp\ D_i\ \big|\ X_i,\ b^{(n)}\,.
\]
\end{lemma}

\begin{proof}
All expectations are taken under the joint law at size $n$ (i.e., $P_n(\cdot\mid b^{(n)})$ and the assignment).
Fix $d\in\{0,1\}$ and Borel sets $B\subset\mathbb{R}^{d_Y}\times\mathbb{R}^{d_Y}$ and $C\subset\{1,\dots,G_n\}$. Then
\begin{align*}
&\E_n\!\Big[{\Pr}_n\!\big\{(Y_i^*(1),Y_i^*(0))\in B,\ D_i=d\ \big|\ X_i,b^{(n)}\big\}\,\mathbf 1\{X_i\in C\}\Big] \\
&\quad=\ \E_n\!\Big[{\Pr}_n\!\big\{(Y_i^*(1),Y_i^*(0))\in B\ \big|\ X_i,b^{(n)}\big\}\,
                               {\Pr}_n\!\big\{D_i=d\ \big|\ X_i,b^{(n)}\big\}\,
                               \mathbf 1\{X_i\in C\}\Big],
\end{align*}
where the equality uses Assumption~\ref{ass:TA-LB-assign}, which implies that $D_i$ is independent of the potential outcomes conditional on $b^{(n)}$ (and hence conditional on $(X_i,b^{(n)})$). Under the uniform without-replacement design with fixed totals within labels,
\[
{\Pr}_n\!\{D_i=1\mid X_i=g,b^{(n)}\}=\eta_g\,.
\]
Thus the integrand factorizes into a function of $(X_i,b^{(n)})$ times another function of $(X_i,b^{(n)})$, which is the defining property of
\(
(Y_i^*(1),Y_i^*(0))\perp\!\!\!\perp D_i\mid X_i,b^{(n)}.
\)
\end{proof}

\subsubsection{Proof of Lemma~\ref{lem:LB8}}\label{app:consistency-proof}

\noindent This proof adapts the argument of Lemma~A.8 in \citet{bai2023efficiency} to the label-based triangular-array setting used here.

By Assumption~\ref{ass:C_mod}(a) and Theorem 5.9 in \cite{van2000asymptotic}, it suffices to show for each $1\le s\le d_\theta$,
\begin{equation}\label{eq:LB8-ULLN}
\sup_{\theta\in\Theta}
\Bigl|
\frac{1}{n}\sum_{i=1}^n \bigl(m_s(X_i,D_i,Y_i;\theta)-\E_n[m_s(X_i,D_i,Y_i;\theta)]\bigr)
\Bigr|
\ \xrightarrow{p}\ 0.
\end{equation}
Notice Assumption~\ref{ass:C_mod}(d) and the dominated convergence theorem (under $P_n$ using the envelope in Assumption~\ref{ass:C_mod}(e\textnormal{-}i)) imply that if
$m_s(x,d,y,\theta_m)\to m_s(x,d,y,\theta)$ as $m\to\infty$ for $\{\theta_m\}\subset\Theta^\ast$,
then $\E_n[m_s(X_i,D_i,Y_i;\theta_m)]\to \E_n[m_s(X_i,D_i,Y_i;\theta)]$. Here, the required dominating function exists by Problem 2.4.1 in \cite{vaart1997weak}. 
Then, Assumption~\ref{ass:C_mod}(c) implies
\begin{align}\label{eq:meas-LB}
\sup_{\theta\in\Theta}
\Bigg|
\frac{1}{n}\sum_{i=1}^n\!\Big(m_s(X_i,D_i,Y_i;&\theta)-\E_n[m_s(X_i,D_i,Y_i;\theta)]\Big)
\Bigg| \notag\\
&=
\sup_{\theta\in\Theta^\ast}
\Bigg|
\frac{1}{n}\sum_{i=1}^n\!\Big(m_s(X_i,D_i,Y_i;\theta)-\E_n[m_s(X_i,D_i,Y_i;\theta)]\Big)
\Bigg| ,
\end{align}
which is measurable. Next, note
\begin{equation}\label{eq:decomp-LB}
m(X_i,D_i,Y_i;\theta)
=
D_i\,m(X_i,1,Y_i^*(1);\theta)
+\{1-D_i\}\,m(X_i,0,Y_i^*(0);\theta).
\end{equation}
By Lemma~\ref{lem:LB7},
\begin{align}\label{eq:mix-LB}
\E_n\!\big[m(X_i,D_i,Y_i;\theta)\big]
&= \E_n\!\Big[\E_n\!\big[m(X_i,D_i,Y_i;\theta)\mid X_i\big]\Big]\notag \\[2mm]
&\overset{(\text{by }\eqref{eq:decomp-LB})}{=}
\E_n\!\Big[\E_n\!\big[D_i\,m(X_i,1,Y_i^*(1);\theta)\mid X_i\big]\Big] \notag\\
&\qquad\qquad+\E_n\!\Big[\E_n\!\big[(1-D_i)\,m(X_i,0,Y_i^*(0);\theta)\mid X_i\big]\Big]\notag \\[2mm]
&\overset{\text{(Lemma \ref{lem:LB7})}}{=}
\E_n\!\Big[\eta(X_i)\,\E_n\!\big[m(X_i,1,Y_i^*(1);\theta)\mid X_i\big]\notag\\
&\qquad\qquad+\{1-\eta(X_i)\}\,\E_n\!\big[m(X_i,0,Y_i^*(0);\theta)\mid X_i\big]\Big],
\end{align}
where $\eta(X_i):=\Pr(D_i=1\mid X_i,b^{(n)})=\eta_{b_i}$. 
Notice that $\mu_{d,n}(x,\theta):=\E_n[m_s(X,d,Y^*(d),\theta)\mid X=x]$ is constant within each stratum (depends on $x$ only through the label).

From \eqref{eq:mix-LB}, we obtain
\begin{align*}
&\E_n\!\Bigg[\sup_{\theta\in\Theta}\Bigl|\frac{1}{n}\sum_{i=1}^n
\bigl(m_s(X_i,D_i,Y_i;\theta)-\E_n[m_s(X_i,D_i,Y_i;\theta)]\bigr)\Bigr|\Bigg] \\[1mm]
&=\E_n\!\Bigg[\sup_{\theta\in\Theta^\ast}\Bigl|\frac{1}{n}\sum_{i=1}^n
\bigl(m_s(X_i,D_i,Y_i;\theta)-\E_n[m_s(X_i,D_i,Y_i;\theta)]\bigr)\Bigr|\Bigg] \\[1mm]
&\le \E_n\!\Bigg[\sup_{\theta\in\Theta^\ast}\Bigl|\frac{1}{n}\sum_{i=1}^n
D_i\bigl(m_s(X_i,1,Y_i^*(1),\theta)-\eta(X_i)\mu_{1,n}(X_i,\theta)\bigr)\Bigr|\Bigg] \\
&\quad + \E_n\!\Bigg[\sup_{\theta\in\Theta^\ast}\Bigl|\frac{1}{n}\sum_{i=1}^n
(1-D_i)\bigl(m_s(X_i,0,Y_i^*(0),\theta)-(1-\eta(X_i))\mu_{0,n}(X_i,\theta)\bigr)\Bigr|\Bigg] \\[1mm]
&= \E_n\!\Bigg[\sup_{\theta\in\Theta^\ast}\Bigl|\frac{1}{n}\sum_{i=1}^n
\Big\{D_i\bigl(m_s(X_i,1,Y_i^*(1),\theta)-\mu_{1,n}(X_i,\theta)\bigr)
+(D_i-\eta(X_i))\,\mu_{1,n}(X_i,\theta)\Big\}\Bigr|\Bigg] \\
&\quad + \E_n\!\Bigg[\sup_{\theta\in\Theta^\ast}\Bigl|\frac{1}{n}\sum_{i=1}^n
\Big\{(1-D_i)\bigl(m_s(X_i,0,Y_i^*(0),\theta)-\mu_{0,n}(X_i,\theta)\bigr)
-(D_i-\eta(X_i))\,\mu_{0,n}(X_i,\theta)\Big\}\Bigr|\Bigg] \\[1mm]
&= \E_n\!\Bigg[\sup_{\theta\in\Theta^\ast}\Bigl|\frac{1}{n}\sum_{i=1}^n
D_i\bigl(m_s(X_i,1,Y_i^*(1),\theta)-\mu_{1,n}(X_i,\theta)\bigr)\Bigr|\Bigg] \\
&\quad + \E_n\!\Bigg[\sup_{\theta\in\Theta^\ast}\Bigl|\frac{1}{n}\sum_{i=1}^n
(1-D_i)\bigl(m_s(X_i,0,Y_i^*(0),\theta)-\mu_{0,n}(X_i,\theta)\bigr)\Bigr|\Bigg] \\[-1mm]
&\qquad\text{since }\frac{1}{n}\sum_{i=1}^n (D_i-\eta(X_i))\,\mu_{d,n}(X_i,\theta)
=\frac{1}{n}\sum_{g=1}^{G_n}\mu_{d,n}(g,\theta)\!\sum_{i\in\lambda_{g,n}}\!(D_i-\eta_g)=0 \\[1mm]
&\lesssim \E_n\!\Bigg[\sup_{\theta\in\Theta^\ast}\Bigl|\frac{1}{n}\sum_{i=1}^n
\bigl(m_s(X_i,1,Y_i^*(1),\theta)-\mu_{1,n}(X_i,\theta)\bigr)\Bigr|\Bigg]
\;+\;\\
&\qquad\E_n\!\Bigg[\sup_{\theta\in\Theta^\ast}\Bigl|\frac{1}{n}\sum_{i=1}^n
\bigl(m_s(X_i,0,Y_i^*(0),\theta)-\mu_{0,n}(X_i,\theta)\bigr)\Bigr|\Bigg] \\[1mm]
&\longrightarrow\;0.
\end{align*}
The first equality uses \eqref{eq:meas-LB}. The first inequality is the triangle inequality splitting treated and control parts after inserting the conditional expectations from \eqref{eq:mix-LB}. The next equality adds and subtracts the label-wise means \(\mu_{d,n}(X_i,\theta)\). The following equality uses that these means are constant within labels and that fixed quotas imply \(\sum_{i\in\lambda_{g,n}}(D_i-\eta_g)=0\) for every label \(g\), so all \((D_i-\eta(X_i))\mu_{d,n}(X_i,\theta)\) terms cancel exactly. The \(\lesssim\) step applies Proposition~C.1 of \citet{han2021complex} within each label and then sums over labels, reducing to the corresponding sums without the assignment indicators up to a universal constant. Finally, convergence is delivered by Assumption~\ref{ass:C_mod}(e) and an application of the backward submartingale convergence theorem (e.g., \citet[Theorem~12.30]{le2022measure}), as detailed in \citet[Theorem~3.1]{han2021complex}. Markov’s inequality then yields \eqref{eq:LB8-ULLN}.

\subsection{Asymptotic Distribution Proof in the Label-Based, Heterogeneous-Shares Stratification Setting}

\subsubsection{Auxiliary Lemma for the Asymptotic Distribution Proof}

\noindent The next lemma adapts Lemma~A.6 in \citet{bai2023efficiency} to the label-based triangular-array setting used here.

We begin with two assignment properties used in the lemma. Throughout, \(X_i=b_i\in\{1,\dots,G_n\}=: \mathcal X\) and we condition on the realized labels \(b^{(n)}\).

\begin{assumption}\label{ass:LB-41}
The treatment assignment mechanism is such that for any integrable $\gamma:\mathcal X\to\mathbb R$,
\[
\frac{1}{n}\sum_{i=1}^n D_i\,\gamma(X_i)\ \xrightarrow{P}\ \E\!\big[\eta(X)\,\gamma(X)\big],
\qquad \eta(x):=\Pr(D=1\mid X=x).
\]
\end{assumption}

\begin{assumption}\label{ass:LB-42}
Let $\rho$ be any metric that metrizes weak convergence. The treatment assignment mechanism is such that for any square-integrable $\gamma:\mathcal X\to\mathbb R^{d_\theta}$ with $\E[\gamma(X)]=0$,
\[
\rho\!\left(\ \frac{1}{\sqrt{n}}\sum_{i=1}^n \big(D_i-\eta(X_i)\big)\,\gamma(X_i)\ ,\ \mathcal N\!\big(0,\ V_{\mathrm{imb},\gamma}\big)\ \middle|\ X^{(n)}\right)\ \xrightarrow{P}\ 0,
\]
for some deterministic variance $V_{\mathrm{imb},\gamma}$.
\end{assumption}

Under the heterogeneous-shares, label-based block randomization in Assumptions~\ref{ass:TA-LB-PO} and \ref{ass:TA-LB-assign}, each label $g$ has exactly $T_g=\eta_g N_g$ treated units. For any integrable $\gamma$,
\begin{align*}
\frac{1}{n}\sum_{i=1}^n D_i\,\gamma(b_i)
&= \frac{1}{n}\sum_{g=1}^{G_n}\ \sum_{i\in\lambda_g} D_i\,\gamma(b_i)\\
&= \frac{1}{n}\sum_{g=1}^{G_n} T_g\,\gamma(g)
= \frac{1}{n}\sum_{i=1}^n \eta_{b_i}\,\gamma(b_i).
\end{align*}
Thus
\(
\frac{1}{n}\sum_{i=1}^n D_i\,\gamma(X_i)
=\frac{1}{n}\sum_{i=1}^n \eta(X_i)\,\gamma(X_i)
\)
holds identically. Under Assumption~\ref{ass:SP-LB}, the right-hand side converges to $\E[\eta(X)\gamma(X)]$, so Assumption~\ref{ass:LB-41} holds.

Furthermore, for any square-integrable $\gamma$,
\begin{align*}
\frac{1}{\sqrt{n}}\sum_{i=1}^n \big(D_i-\eta_{b_i}\big)\,\gamma(b_i)
&= \frac{1}{\sqrt{n}}\sum_{g=1}^{G_n}\ \sum_{i\in\lambda_g}\big(D_i-\eta_g\big)\,\gamma(b_i)\\
&= \frac{1}{\sqrt{n}}\sum_{g=1}^{G_n}\gamma(g)\sum_{i\in\lambda_g}\big(D_i-\eta_g\big)\\
&= \frac{1}{\sqrt{n}}\sum_{g=1}^{G_n}\gamma(g)\,\big(T_g-\eta_g N_g\big)\\
&= 0,
\end{align*}
since $T_g=\eta_g N_g$ for all $g$. Hence Assumption~\ref{ass:LB-42} holds with $V_{\mathrm{imb},\gamma}=0$.

\begin{lemma}
\label{lem:LB6}
Suppose the treatment assignment mechanism satisfies Assumptions \ref{ass:TA-LB-assign}, \ref{ass:TA-LB-PO}, \ref{ass:LB-41} and \ref{ass:LB-42}.
Let $f(x,d,y)$ be a vector-valued function such that
$\overline{\E}_n\!\big[f(X,D,Y)\big]=0$ and there exists $\delta>0$ with
$\sup_{n\ge1}\ \max_{d\in\{0,1\}}\ \overline{\E}_n\!\big[\|f(X,d,Y^*(d))\|^{\,2+\delta}\big]\;<\;\infty.$
Then
\[
\frac{1}{\sqrt{n}}\sum_{i=1}^n f(X_i,D_i,Y_i)\ \xRightarrow{d}\ \mathcal N(0,V_f),
\]
where $V_f=V_{3,f}$ with
\begin{align*}
V_{3,f}
&=\overline{\E}_n\!\Big[
\eta(X)\,\Var_n\!\big(f(X,1,Y^*(1))\mid X\big)
+\big(1-\eta(X)\big)\,\Var_n\!\big(f(X,0,Y^*(0))\mid X\big)
\Big].
\end{align*}
\end{lemma}

\begin{proof}
Notice
\[
C_n\;:=\;\frac{1}{\sqrt{n}}\sum_{i=1}^{n} f(X_i,D_i,Y_i)
\;=\; C_{1,n}+C_{2,n}+C_{3,n},
\]
where
\begin{align*}
C_{1,n}
&=\frac{1}{\sqrt{n}}\sum_{i=1}^{n}\E_n\!\big[f(X_i,D_i,Y_i)\mid X_i\big]  \\
&=\frac{1}{\sqrt{n}}\sum_{i=1}^{n}\Big(\eta(X_i)\,\E_n\!\big[f(X_i,1,Y_i^*(1))\mid X_i\big]
+\big(1-\eta(X_i)\big)\,\E_n\!\big[f(X_i,0,Y_i^*(0))\mid X_i\big]\Big), \\[1ex]
C_{2,n}
&=\frac{1}{\sqrt{n}}\sum_{i=1}^{n}\big(D_i-\eta(X_i)\big)\,
\E_n\!\big[f(X_i,1,Y_i^*(1))-f(X_i,0,Y_i^*(0))\mid X_i\big], \\[1ex]
C_{3,n}
&=\frac{1}{\sqrt{n}}\sum_{i=1}^{n}\Big\{
D_i\Big(f(X_i,1,Y_i^*(1))-\E_n\!\big[f(X_i,1,Y_i^*(1))\mid X_i\big]\Big)\\[-0.25ex]
&\hspace{4.5cm}
+\ \big(1-D_i\big)\Big(f(X_i,0,Y_i^*(0))-\E_n\!\big[f(X_i,0,Y_i^*(0))\mid X_i\big]\Big)
\Big\}.
\end{align*}

Note that $C_{1,n}$ has mean zero because
\[
\overline{\E}_n\!\big[\E_n[f(X_i,D_i,Y_i)\mid X_i]\big]
\;=\;\overline{\E}_n\!\big[f(X_i,D_i,Y_i)\big]\;=\;0.
\]
Moreover,
\[
C_{1,n}=\E_n[C_n\mid X^{(n)}],\;
C_{2,n}=\E_n[C_n\mid X^{(n)},D^{(n)}]-\E_n[C_n\mid X^{(n)}],\;
C_{3,n}=C_n-\E_n[C_n\mid X^{(n)},D^{(n)}].
\]
In the present design-based analysis, we condition on $b^{(n)}$, so $X^{(n)}$ is fixed and
\[
C_{1,n}=\E_n[C_n\mid X^{(n)}]=0,
\]
hence $V_{1,f}=0$.

Next, condition on $b^{(n)}$ and sum label by label:
\[
C_{2,n}
=\frac{1}{\sqrt{n}}\sum_{g=1}^{G_n}\ \sum_{i\in\lambda_g}\big(D_i-\eta_g\big)\,
\E_n\!\big[f(X_i,1,Y_i^*(1))-f(X_i,0,Y_i^*(0))\mid X_i\big].
\]
By Assumption~\ref{ass:TA-LB-assign}, the inner conditional expectation is constant within each label $g$, so it factors out of the inner sum. Under fixed within–label quotas, $\sum_{i\in\lambda_g}(D_i-\eta_g)=0$ for every $g$, which is the zero–imbalance case of Assumption~\ref{ass:LB-42} (i.e., $V_{\mathrm{imb},\gamma}=0$ with $\gamma(x)=\E_n[f(x,1,Y^*(1))-f(x,0,Y^*(0))\mid X=x]$). Hence $C_{2,n}\equiv 0$, and thus
\[
\rho\!\left(C_{2,n},\,\mathcal N(0,0)\,\big|\,X^{(n)}\right)\ \xrightarrow{P}\ 0.
\]

For $C_{3,n}$,
\begin{align*}
s_{3,n}^2
&:=\Var_n\!\big[C_{3,n}\mid X^{(n)},D^{(n)}\big] \\
&=\overline{\E}_n\!\Big[
D\,\Var_n\!\big(f(X,1,Y^*(1))\mid X\big)
+(1-D)\,\Var_n\!\big(f(X,0,Y^*(0))\mid X\big)
\Big].
\end{align*}
Taking expectation under the joint design conditional on $X^{(n)}$ and using Assumption~\ref{ass:LB-41} (so that $\E_n[D\mid X]=\eta(X)$) yields
\[
\E_n\!\big[s_{3,n}^2\mid X^{(n)}\big]
=\overline{\E}_n\!\Big[
\eta(X)\,\Var_n\!\big(f(X,1,Y^*(1))\mid X\big)
+\big(1-\eta(X)\big)\,\Var_n\!\big(f(X,0,Y^*(0))\mid X\big)
\Big]
=: V_{3,f,n},
\]
and Assumption~\ref{ass:LB-41} further implies $s_{3,n}^2 - V_{3,f,n}=o_p(1)$. If $V_{3,f,n}\to V_{3,f}$, then $s_{3,n}^2\ \xrightarrow{P}\ V_{3,f}$.
Moreover, writing
\begin{align*}
\zeta_{i,n}
&:=D_i\!\left(f(X_i,1,Y_i^*(1))-\E_n[f(X_i,1,Y_i^*(1))\mid X_i]\right) \\
&\qquad+(1-D_i)\!\left(f(X_i,0,Y_i^*(0))-\E_n[f(X_i,0,Y_i^*(0))\mid X_i]\right),
\end{align*}
we have $C_{3,n}=\frac{1}{\sqrt{n}}\sum_{i=1}^n \zeta_{i,n}$ with $\E_n[\zeta_{i,n}\mid X^{(n)},D^{(n)}]=0$ and conditional
independence across $i$. Let $Z_{i,n}:=\zeta_{i,n}/\sqrt{n}$ and
$s_{3,n}^2=\sum_{i=1}^n \Var_n(Z_{i,n}\mid X^{(n)},D^{(n)})$.
By Jensen’s inequality and the uniform $(2+\delta)$-moment bound on $f(X,d,Y^*(d))$,
\[
\sum_{i=1}^n \E_n\!\big[\|\zeta_{i,n}\|^{2+\delta}\mid X^{(n)},D^{(n)}\big]\ \lesssim\ n,
\]
and thus
\[
\frac{1}{s_{3,n}^{2+\delta}}\sum_{i=1}^n \E_n\!\big[\|Z_{i,n}\|^{2+\delta}\mid X^{(n)},D^{(n)}\big]
=\frac{1}{s_{3,n}^{2+\delta}}\cdot \frac{1}{n^{1+\delta/2}}
\sum_{i=1}^n \E_n\!\big[\|\zeta_{i,n}\|^{2+\delta}\mid X^{(n)},D^{(n)}\big]
\;=\;o_p(1).
\]
Hence, the conditional Lyapunov condition holds. By the conditional Lyapunov CLT (and Cramér–Wold for vector $f$),
\[
\rho\!\left(C_{3,n},\,\mathcal N(0,V_{3,f})\mid X^{(n)},D^{(n)}\right)\ \xrightarrow{P}\ 0.
\]

It remains to pass from the conditional limits to the unconditional limit for $C_n=C_{1,n}+C_{2,n}+C_{3,n}$.
Recall $C_{2,n}\equiv0$. Fix $t\in\mathbb R$ and define the conditional CDF
\(
F_{n,t}(X^{(n)},D^{(n)}) := \Pr\{\,C_{1,n}+C_{3,n}\le t \mid X^{(n)},D^{(n)}\}.
\)
By the conditional Lyapunov CLT proved above,
\(
\rho\!\big(C_{3,n},\mathcal N(0,V_{3,f})\mid X^{(n)},D^{(n)}\big)\to^P 0.
\)
Thus, by a subsequence principle, there exists a subsequence $n_k$ along which
\[
\sup_{t\in\mathbb R}\Big|
\Pr\{\,C_{1,n_k}+C_{3,n_k}\le t \mid X^{(n_k)},D^{(n_k)}\}
-\Phi\!\Big(\frac{t-C_{1,n_k}}{\sqrt{V_{3,f}}}\Big)
\Big|\ \xrightarrow{\text{a.s.}}\ 0.
\]
Taking design-based expectations and using bounded convergence,
\[
\sup_{t\in\mathbb R}\Big|
\Pr\{\,C_{1,n_k}+C_{3,n_k}\le t\}
-\E_n\!\Big[\Phi\!\Big(\frac{t-C_{1,n_k}}{\sqrt{V_{3,f}}}\Big)\Big]
\Big|\ \longrightarrow\ 0.
\]
Under the design-based view we condition on $b^{(n)}$, hence $C_{1,n}=\E_n[C_n\mid X^{(n)}]=0$ almost surely.
Therefore,
\[
\E_n\!\Big[\Phi\!\Big(\frac{t-C_{1,n_k}}{\sqrt{V_{3,f}}}\Big)\Big]
=\Phi\!\Big(\tfrac{t}{\sqrt{V_{3,f}}}\Big),
\]
and so
\[
\sup_{t\in\mathbb R}\Big|
\Pr\{\,C_{1,n_k}+C_{3,n_k}\le t\}
-\Phi\!\Big(\tfrac{t}{\sqrt{V_{3,f}}}\Big)
\Big|\ \longrightarrow\ 0.
\]
Hence $C_{1,n_k}+C_{3,n_k}\Rightarrow \mathcal N(0,V_{3,f})$. Since every subsequence admits a further subsequence with this property, we conclude $C_n\Rightarrow \mathcal N(0,V_{3,f})$.
\end{proof}

\subsubsection{Proof of Theorem \ref{thm:LB-AN}}\label{app:asymptotics-proof}

\noindent This proof adapts the argument of Theorem~3.1 from Appendix A.1 of \citet{bai2023efficiency} to the label-based triangular-array setting used here.

First, for any fixed $a\in\mathbb R^{d_\theta}$, apply Lemma~\ref{lem:LB6} to the scalar array
$a^\top \psi_n^\ast(X_i,D_i,Y_i;\theta_0)$. By the Cramér–Wold device, \eqref{eq:LB-CLT} follows from \eqref{eq:LB-IF}.
In particular, in the Lemma~\ref{lem:LB6} decomposition the “assignment–fluctuation’’ term $C_{2,n}$ is identically zero under fixed within–label quotas, because the inner conditional expectations are label–measurable and $\sum_{i\in\lambda_g}(D_i-\eta_g)=0$ for every label $g$.

To show \eqref{eq:LB-IF}, we first establish the standard linearization
\begin{equation}\label{eq:LB-linear}
\sqrt{n}(\hat\theta_n-\theta_0)
= -\,M_n^{-1}\,\frac{1}{\sqrt{n}}\sum_{i=1}^n m(X_i,D_i,Y_i;\theta_0)+o_p(1).
\end{equation}
By the proof of Theorem 5.21 in \citet{van2000asymptotic}, it suffices to verify
\begin{equation}\label{eq:LB-Ln}
\mathbb L_n(\hat\theta_n)\ \xrightarrow{p}\ 0,
\qquad
\mathbb L_n(\theta):=\big(\mathbb L_n^{(1)}(\theta),\ldots,\mathbb L_n^{(d_\theta)}(\theta)\big)^\prime,
\end{equation}
where, for $1\le s\le d_\theta$,
\begin{align*}
\mathbb L_n^{(s)}(\theta)
=
\frac{1}{\sqrt{n}}\sum_{i=1}^n\!\big(m_s(X_i,D_i,Y_i;\theta)-&\E_n[m_s(X_i,D_i,Y_i;\theta)]\big) \\
&-
\frac{1}{\sqrt{n}}\sum_{i=1}^n\!\big(m_s(X_i,D_i,Y_i;\theta_0)-\E_n[m_s(X_i,D_i,Y_i;\theta_0)]\big).
\end{align*}
By Assumption~\ref{ass:C_mod}(c)–(d), Proposition 8.11 in \citet{kosorok2008introduction}, and the measurability argument used for \eqref{eq:meas-LB},
\[
\sup_{\theta\in\Theta:\ \|\theta-\theta_0\|<\delta}\,
\big|\mathbb L_n^{(s)}(\theta)\big|
=\sup_{\theta\in\Theta^\ast:\ \|\theta-\theta_0\|<\delta}\,
\big|\mathbb L_n^{(s)}(\theta)\big|.
\]
Since $\hat\theta_n \xrightarrow{p}\theta_0$ by Lemma~\ref{lem:LB8}, to show \eqref{eq:LB-Ln} it suffices that for every $\epsilon>0$ and every sequence $\delta_n\downarrow 0$ (p.89 of \citet{vaart1997weak}),
\begin{equation}\label{eq:LB-26}
\lim_{n\to\infty}\Pr\!\left(\ \sup_{\theta\in\Theta^\ast:\ \|\theta-\theta_0\|<\delta_n}
\big|\mathbb L_n^{(s)}(\theta)\big|>\epsilon\ \right)=0.
\end{equation}
Decompose $\mathbb L_n^{(s)}(\theta)=\mathbb L_{n,1}^{(s)}(\theta)+\mathbb L_{n,0}^{(s)}(\theta)$, where
\[
\begin{aligned}
\mathbb L_{n,1}^{(s)}(\theta)
&=\frac{1}{\sqrt{n}}\sum_{i=1}^n
D_i\Big(m_s(X_i,1,Y_i^*(1);\theta)-m_s(X_i,1,Y_i^*(1);\theta_0)\Big) \\
&\qquad\qquad-\frac{1}{\sqrt{n}}\sum_{i=1}^n
\E_n\!\Big[m_s(X_i,1,Y_i^*(1);\theta)-m_s(X_i,1,Y_i^*(1);\theta_0)\Big],\\[1mm]
\mathbb L_{n,0}^{(s)}(\theta)
&=\frac{1}{\sqrt{n}}\sum_{i=1}^n
(1-D_i)\Big(m_s(X_i,0,Y_i^*(0);\theta)-m_s(X_i,0,Y_i^*(0);\theta_0)\Big) \\
&\qquad\qquad-\frac{1}{\sqrt{n}}\sum_{i=1}^n
\E_n\!\Big[m_s(X_i,0,Y_i^*(0);\theta)-m_s(X_i,0,Y_i^*(0);\theta_0)\Big],
\end{aligned}
\]
For $d\in\{0,1\}$, define
\begin{align*}
\rho_{Q_n}(\theta,\theta_0)
:=\Bigg\{\overline{\E}_n\!\Big[\Big(m_s(X,d,Y^*(d);\theta)&-m_s(X,d,Y^*(d);\theta_0) \\
&-\overline{\E}_n\!\big[m_s(X,d,Y^*(d);\theta)-m_s(X,d,Y^*(d);\theta_0)\big]\Big)^2\Big]\Bigg\}^{1/2}.
\end{align*}
Then, by assumption~\ref{ass:C_mod}(c),
\begin{align*}
\rho_{Q_n}^2(\theta,\theta_0)
&= \Var_{\overline{P}_n}\!\big(m_s(X,d,Y^*(d);\theta)-m_s(X,d,Y^*(d);\theta_0)\big) \\
&\le \ \overline{\E}_n\!\Big[\big(m_s(X,d,Y^*(d);\theta)-m_s(X,d,Y^*(d);\theta_0)\big)^2\Big]\to0,
\end{align*}
and so $\rho_{Q_n}(\theta,\theta_0)\to0$ as $\|\theta-\theta_0\|\to0$, meaning $\rho_{Q_n}(\theta,\theta_0)$ is continuous in $\theta$. \\
Fix any sequence $\tilde\delta_n\downarrow0$. For each $n$, there exists $n'$ such that
\[
\{\theta\in\Theta^\ast:\ \|\theta-\theta_0\|<\delta_{n'}\}
\subset
\{\theta\in\Theta^\ast:\ \rho_{Q_n}(\theta,\theta_0)<\tilde{\delta}_n\}.
\]
By Proposition C.1 in \citet{han2021complex}, applied label-wise to the within-label assignment (as in Lemma~\ref{lem:LB8}),
\begin{align*}
\E_n\!\left[
  \sup_{\rho_{Q_n}(\theta,\theta_0)<\tilde\delta_n}
  \big|\mathbb L^{(s)}_{n,d}(\theta)\big|
\right]
&\ \lesssim\
\E_n\!\left[
  \sup_{\rho_{Q_n}(\theta,\theta_0)<\tilde\delta_n}
  \left|
    \frac{1}{\sqrt{n}}\sum_{i=1}^n
    \Big(
      m_s\!\big(X_i,d,Y_i^*(d);\theta\big)-m_s\!\big(X_i,d,Y_i^*(d);\theta_0\big) \right.\right.\\
&\hspace{3.5em}\left.\left.
      \ -\E_n\!\big[m_s\!\big(X_i,d,Y_i^*(d);\theta\big)-m_s\!\big(X_i,d,Y_i^*(d);\theta_0\big)\big]
    \Big)
  \right|
\right]
 \to\ 0\,,
\end{align*}
where the convergence follows from Assumption~\ref{ass:C_mod}(e) and Corollary 2.3.12 in \citet{vaart1997weak}.
Markov’s inequality then gives \eqref{eq:LB-26}, hence \eqref{eq:LB-Ln}, and thus \eqref{eq:LB-linear}.
Assumption~\ref{ass:C_mod}(b) implies $M_n\to M$ (nonsingular), so $M_n^{-1}=M^{-1}+o_p(1)$.

It remains to pass from \eqref{eq:LB-linear} to \eqref{eq:LB-IF}. Let
\begin{align*}
m^{\ast}(X_i,D_i,Y_i;\theta_0)
&:= \eta(X_i)\,\E_n\!\big[m(X_i,1,Y_i^*(1);\theta_0)\mid X_i\big] \\
   &\quad+\{1-\eta(X_i)\}\,\E_n\!\big[m(X_i,0,Y_i^*(0);\theta_0)\mid X_i\big]\\
&\quad
+\mathbf 1\{D_i=1\}\Big(m(X_i,1,Y_i^*(1);\theta_0)
   -\E_n\!\big[m(X_i,1,Y_i^*(1);\theta_0)\mid X_i\big]\Big)\\
&\quad
+\mathbf 1\{D_i=0\}\Big(m(X_i,0,Y_i^*(0);\theta_0)
   -\E_n\!\big[m(X_i,0,Y_i^*(0);\theta_0)\mid X_i\big]\Big).
\end{align*}
Since
$Y_i=D_iY_i^*(1)+(1-D_i)Y_i^*(0)$, adding and subtracting conditional means given $X_i$ yields
\[
\frac{1}{\sqrt{n}}\sum_{i=1}^n m(X_i,D_i,Y_i;\theta_0)
=
\frac{1}{\sqrt{n}}\sum_{i=1}^n m^{\ast}(X_i,D_i,Y_i;\theta_0)
+\frac{1}{\sqrt{n}}\sum_{i=1}^n \big(D_i-\eta(X_i)\big)\,\Omega(X_i),
\]
where \(\Omega(x):=\E_n[m(x,1,Y^*(1);\theta_0)-m(x,0,Y^*(0);\theta_0)\mid X=x]\).
In the label-based design \(X\) is the label, so \(\Omega(X_i)\) is constant within each label:
for each \(g\) and all \(i\in\lambda_{g,n}\), \(\Omega^{(\ell)}(X_i)=\Omega^{(\ell)}_{g,n}\).
With fixed within–label quotas, \(\sum_{i\in\lambda_{g,n}}(D_i-\eta_{g,n})=0\). Hence, for each
\(1\le \ell\le d_\theta\),
\[
\frac{1}{\sqrt{n}}\sum_{i=1}^n (D_i-\eta(X_i))\,\Omega^{(\ell)}(X_i)
=\frac{1}{\sqrt{n}}\sum_{g=1}^{G_n}\Omega^{(\ell)}_{g,n}
\sum_{i\in\lambda_{g,n}}(D_i-\eta_{g,n})
=0 \quad\text{(conditionally on }X^{(n)}\text{)}.
\]
Collecting components,
\[
\frac{1}{\sqrt{n}}\sum_{i=1}^n (D_i-\eta(X_i))\,\Omega(X_i)\equiv 0,
\]
and therefore
\[
-\,M_n^{-1}\,\frac{1}{\sqrt{n}}\sum_{i=1}^n m(X_i,D_i,Y_i;\theta_0)
=
-\,M_n^{-1}\,\frac{1}{\sqrt{n}}\sum_{i=1}^n m^{\ast}(X_i,D_i,Y_i;\theta_0).
\]
Combining this identity with \eqref{eq:LB-linear} and \(M_n^{-1}=M^{-1}+o_p(1)\) yields \eqref{eq:LB-IF}.$\quad\blacksquare$

\subsection{Consistency of the Variance Estimator in the Label-Based, Heterogeneous-Shares Stratification Setting}
\label{app:var_consistency_label}

\noindent This subsection adapts the variance-consistency arguments in \citet[Appendix C, Lemmas C.1--C.3]{bai2024inference} to the label-based triangular-array stratified design and notation used here. The next three lemmas restate their moment and cross-moment limits under our block-randomization setup, and the final argument combines these limits to show consistency of the plug-in variance estimator.

\subsubsection{LLN for Label-Based Treatment Moments}

\begin{lemma}[LLN for Label-Based Treatment Moments]
\label{lem:C1-LB}
Fix $d\in\{0,1\}$ and $r\in\{1,2\}$. Under Assumptions~\ref{ass:TA-LB-PO}, \ref{ass:TA-LB-assign}, and \ref{ass:SP-LB},
\[
\frac{1}{n}\sum_{i=1}^n Y_i^{\,r}(d)\,\mathbf 1\{D_i=d\}
\ \xrightarrow{P}\ 
\E\!\big[\eta_d(X)\,Y^{\,r}(d)\big],
\qquad
\eta_1(x)=\eta(x),\ \eta_0(x)=1-\eta(x).
\]
\end{lemma}

\begin{proof}
Let $Z_i:=Y_i^{\,r}(d)$ and $\eta_{d,i}:=\eta_d(X_i)$. Decompose
\begin{align}
\frac{1}{n}\sum_{i=1}^n Z_i\,\mathbf 1\{D_i=d\}
&=\frac{1}{n}\sum_{i=1}^n\!\Big(Z_i\mathbf 1\{D_i=d\}-\E_n[Z_i\mathbf 1\{D_i=d\}\mid X^{(n)},Y^{(n)}(d)]\Big) \notag\\
&\qquad+\frac{1}{n}\sum_{i=1}^n \E_n[Z_i\mathbf 1\{D_i=d\}\mid X^{(n)},Y^{(n)}(d)] \notag\\
&=:A_n+B_n. \label{eq:C1-decomp}
\end{align}

By Assumptions~\ref{ass:TA-LB-PO}–\ref{ass:TA-LB-assign}, $\E_n[Z_i\mathbf 1\{D_i=d\}\mid X^{(n)},Y^{(n)}(d)]=\eta_{d,i}\,Z_i$, so
\[
B_n=\frac{1}{n}\sum_{i=1}^n \eta_{d,i}\,Z_i \ \xrightarrow{P}\ \E[\eta_d(X)\,Y^{\,r}(d)]
\]
by Assumption~\ref{ass:SP-LB} (triangular–array LLN/integrability).

For $A_n$, condition on $b^{(n)}$ so $X^{(n)}$ is fixed and write
\[
A_n \;=\; \frac{1}{n}\sum_{g=1}^{G_n}\ \sum_{i\in\lambda_g}\Big(Z_i\mathbf 1\{D_i=d\}-\eta_{d,g}Z_i\Big),
\qquad
\eta_{d,g}:=\eta_d(X_i)\ \text{ for }i\in\lambda_g.
\]
Let $\overline Z_g:=N_g^{-1}\sum_{i\in\lambda_g}Z_i$. Under fixed quotas $T_{g,d}=\eta_{d,g}N_g$ and simple random sampling without replacement within label $g$, for $N_g\ge2$,
\[
\Var_n\!\Big(\sum_{i\in\lambda_g}\!\big(Z_i\mathbf 1\{D_i=d\}-\eta_{d,g}Z_i\big)\,\Big|\,X^{(n)},Y^{(n)}(d)\Big)
= \eta_{d,g}(1-\eta_{d,g})\frac{N_g}{N_g-1}\sum_{i\in\lambda_g}(Z_i-\overline Z_g)^2,
\]
while for $N_g=1$ the variance is $0$. Hence, for a universal constant $C<\infty$,
\[
\Var_n\!\Big(\sum_{i\in\lambda_g}\!\big(Z_i\mathbf 1\{D_i=d\}-\eta_{d,g}Z_i\big)\,\Big|\,X^{(n)},Y^{(n)}(d)\Big)
\ \le\ C\sum_{i\in\lambda_g}(Z_i-\overline Z_g)^2.
\]
Therefore,
\[
\Var_n(A_n\mid X^{(n)},Y^{(n)}(d))
\le \frac{C}{n^2}\sum_{g=1}^{G_n}\sum_{i\in\lambda_g}(Z_i-\overline Z_g)^2.
\]
Taking design expectations and using $(a-b)^2\le 2a^2+2b^2$,
\[
\E_n\!\big[\Var_n(A_n\mid X^{(n)},Y^{(n)}(d))\big]
\ \le\ \frac{2C}{n^2}\,\E_n\!\Big[\sum_{i=1}^n Z_i^2\Big]
=\frac{2C}{n}\,\overline{\E}_n[Z_i^2]
\ \to\ 0,
\]
since $\sup_n \overline{\E}_n[Z_i^2]<\infty$ by Assumption~\ref{ass:SP-LB}. By Chebyshev, $A_n=o_P(1)$.
Combining with \eqref{eq:C1-decomp} completes the proof.
\end{proof}

\subsubsection{Consistency of Cross-Treatment Block Products}

\begin{lemma}[Consistency of Cross-Treatment Block Products]
\label{lem:LB-C2}
Suppose Assumptions~\ref{ass:TA-LB-PO}--\ref{ass:TA-LB-assign} and \ref{ass:SP-LB} hold (including fixed within–label quotas with at least two treated and two control units per label), the moment boundedness in Assumption~\ref{ass:C_mod}(e) holds, and we are in the many–small–strata regime $G_n\to\infty$ with $\max_g N_g/n\to0$.
Fix $d\neq d'$. For each label $g$, define
\[
\hat\rho_g(d,d')
\ :=\
\bigg(\frac{1}{T_g(d)}\sum_{i\in\lambda_g} Y_i\,\mathbf 1\{D_i=d\}\bigg)
\bigg(\frac{1}{T_g(d')}\sum_{k\in\lambda_g} Y_k\,\mathbf 1\{D_k=d'\}\bigg),
\]
and the size–weighted average
\[
\hat\rho_n(d,d')\ :=\ \frac{1}{n}\sum_{g=1}^{G_n} N_g\,\hat\rho_g(d,d')\, .
\]
Then
\[
\hat\rho_n(d,d')\ \xrightarrow{P}\ \E\!\big[\Gamma_d(X)\Gamma_{d'}(X)\big],
\qquad \Gamma_d(x):=\E\!\big[Y^*(d)\mid X=x\big].
\]
\end{lemma}

\begin{proof}
Condition on $b^{(n)}$, so $X^{(n)}$ and the label memberships are fixed. By consistency $Y_i=Y_i^*(D_i)$. Within a label $g$, under fixed quotas and simple random sampling without replacement,
\[
\E_n\!\left[\frac{1}{T_g(d)}\sum_{i\in\lambda_g} Y_i\,\mathbf 1\{D_i=d\}\,\Big|\,X^{(n)},Y^{(n)}(d)\right]
=\frac{1}{T_g(d)}\sum_{i\in\lambda_g} Y_i^*(d)\,\E_n\!\big[\mathbf 1\{D_i=d\}\mid X^{(n)}\big]
=\overline{Y^*_{d,g}},
\]
where $\overline{Y^*_{d,g}}:=N_g^{-1}\sum_{i\in\lambda_g}Y_i^*(d)$ and \(\E_n[\mathbf 1\{D_i=d\}\mid X^{(n)}]=T_g(d)/N_g\).
Taking the superpopulation expectation conditional on $X^{(n)}$,
\[
\E\!\Big[\ \E_n\!\Big[\frac{1}{T_g(d)}\sum_{i\in\lambda_g} Y_i\,\mathbf 1\{D_i=d\}\,\Big|\,X^{(n)},Y^{(n)}(d)\Big]\ \Big|\ X^{(n)}\Big]
=\overline\Gamma_{d,g},
\quad
\overline\Gamma_{d,g}:=\frac{1}{N_g}\sum_{i\in\lambda_g}\Gamma_d(X_i).
\]
Similarly,
\[
\E\!\Big[\ \E_n\!\Big[\frac{1}{T_g(d')}\sum_{k\in\lambda_g} Y_k\,\mathbf 1\{D_k=d'\}\,\Big|\,X^{(n)},Y^{(n)}(d')\Big]\ \Big|\ X^{(n)}\Big]
=\overline\Gamma_{d',g}.
\]

For the product, since $d\neq d'$ only cross–unit terms contribute. Under fixed quotas,
\[
\E_n\!\big[\mathbf 1\{D_i=d\}\mathbf 1\{D_k=d'\}\mid X^{(n)}\big]
=\frac{T_g(d)\,T_g(d')}{N_g(N_g-1)}\qquad (i\neq k,\ i,k\in\lambda_g).
\]
Conditioning also on $Y^{(n)}(d),Y^{(n)}(d')$,
\[
\E_n\!\big[Y_i^*(d)Y_k^*(d')\,\mathbf 1\{D_i=d\}\mathbf 1\{D_k=d'\}\mid X^{(n)},Y^{(n)}(d),Y^{(n)}(d')\big]
= Y_i^*(d)\,Y_k^*(d')\,\frac{T_g(d)\,T_g(d')}{N_g(N_g-1)}.
\]
Taking the superpopulation expectation conditional on $X^{(n)}$ and using conditional independence of potential outcomes across distinct units given labels (Assumption~\ref{ass:TA-LB-PO}), for $i\neq k$ with $i,k\in\lambda_g$,
\[
\E\!\big[Y_i^*(d)Y_k^*(d')\mid X^{(n)}\big]
=\Gamma_d(X_i)\,\Gamma_{d'}(X_k).
\]
Hence
\[
\E\!\Big[\ \E_n\!\big[\hat\rho_g(d,d')\mid X^{(n)},Y^{(n)}(d),Y^{(n)}(d')\big]\ \Big|\ X^{(n)}\Big]
=\frac{1}{N_g(N_g-1)}\!\sum_{\substack{i\neq k\\ i,k\in\lambda_g}}\Gamma_d(X_i)\Gamma_{d'}(X_k).
\]
In the label-based setup $X_i=b_g$ is label-constant, so $\Gamma_d(X_i)=\Gamma_d(b_g)$ and $\Gamma_{d'}(X_k)=\Gamma_{d'}(b_g)$ for all $i,k\in\lambda_g$, and the last display equals
\(
\overline\Gamma_{d,g}\,\overline\Gamma_{d',g}.
\)

Aggregating across labels and using $\sum_g N_g=n$,
\begin{align*}
\E\!\Big[\ \E_n\!\big[\hat\rho_n(d,d')\mid X^{(n)},Y^{(n)}(d),Y^{(n)}(d')\big]\ \Big|\ X^{(n)}\Big]
&=\frac{1}{n}\sum_{g=1}^{G_n} N_g\,\overline\Gamma_{d,g}\,\overline\Gamma_{d',g} \\
&=\frac{1}{n}\sum_{i=1}^n \Gamma_d(X_i)\Gamma_{d'}(X_i)
\\ &\xrightarrow{P}\ \E\!\big[\Gamma_d(X)\Gamma_{d'}(X)\big],
\end{align*}
by Assumption~\ref{ass:SP-LB} (triangular–array LLN for label functions with the moment bound in \ref{ass:C_mod}(e)).

For fluctuations, let $\tilde\rho_g:=\hat\rho_g-\E_n[\hat\rho_g\mid X^{(n)},Y^{(n)}(d),Y^{(n)}(d')]$ and
\[
\Delta_n:=\hat\rho_n-\E_n[\hat\rho_n\mid X^{(n)},Y^{(n)}(d),Y^{(n)}(d')]=\frac{1}{n}\sum_{g=1}^{G_n} N_g\,\tilde\rho_g.
\]
By Assumption~\ref{ass:TA-LB-assign}, assignments are independent across labels conditional on $X^{(n)}$ (with potential outcomes fixed in this conditional view), so $\{\tilde\rho_g\}$ are conditionally independent and mean zero. By \ref{ass:C_mod}(e), the within–label treated/control means have uniformly bounded second moments (the “at least two per arm” quota ensures denominators $T_g(d),T_g(d')$ are bounded away from zero); with $\max_g N_g/n\to0$, a conditional WLLN for triangular arrays gives $\Delta_n\to^P 0$. Combining with the convergence of the conditional mean yields the claim.
\end{proof}

\subsubsection{Consistency of Within-Treatment Second Moments}

\begin{lemma}[Consistency of Within-Treatment Second Moments]
\label{lem:C3}
Fix $d\in\{0,1\}$. Suppose Assumptions~\ref{ass:TA-LB-PO}--\ref{ass:TA-LB-assign}, \ref{ass:SP-LB}, and \ref{ass:C_mod}(e) hold, the many–small–strata condition $G_n\to\infty$ and $\max_{g\le G_n}N_g/n\to0$ holds, and each label $g$ satisfies $T_g(d)\ge2$. Define
\[
\hat\rho_n(d,d)
\ :=\ \frac{1}{n}\sum_{g=1}^{G_n}
N_g\,\frac{2}{T_g(d)\{T_g(d)-1\}}
\sum_{\substack{i<i'\\ i,i'\in\lambda_g}}
Y_i\,Y_{i'}\,\mathbf 1\{D_i=D_{i'}=d\}.
\]
Then
\[
\hat\rho_n(d,d)\ \xrightarrow{P}\ \E\!\big[\Gamma_d^2(X)\big],
\qquad \Gamma_d(x):=\E\!\big[Y^*(d)\mid X=x\big].
\]
\end{lemma}

\begin{proof}
Work design-based, conditioning on $b^{(n)}$, so labels and $\{X_i\}$ are fixed. Within label $g$, $X_i=b_g$ and hence $\Gamma_d(X_i)=\Gamma_{d,g}$ for all $i\in\lambda_g$. Under \ref{ass:TA-LB-assign} (fixed within-label quotas with simple random sampling without replacement), for $i\neq i'$ in $\lambda_g$,
\[
\Pr_n(D_i=D_{i'}=d\mid X^{(n)})=\frac{T_g(d)\{T_g(d)-1\}}{N_g(N_g-1)}.
\]
By independence of assignments from potential outcomes within labels (\ref{ass:TA-LB-assign}) and by \ref{ass:TA-LB-PO} (potential outcomes are independent across distinct units given labels),
\[
\E\!\big[Y_i^*(d)Y_{i'}^*(d)\mid X^{(n)}\big]
=\Gamma_d(X_i)\,\Gamma_d(X_{i'})=\Gamma_{d,g}^2,
\]
so
\[
\E_n\!\big[Y_i\,Y_{i'}\,\mathbf 1\{D_i=D_{i'}=d\}\mid X^{(n)}\big]
=\Gamma_{d,g}^2\,\frac{T_g(d)\{T_g(d)-1\}}{N_g(N_g-1)}.
\]
Averaging over unordered pairs in $\lambda_g$ with the prefactor $2/\{T_g(d)[T_g(d)-1]\}$ yields
\[
\E_n\!\left[
\frac{2}{T_g(d)\{T_g(d)-1\}}
\sum_{i<i'\in\lambda_g} Y_i Y_{i'} \mathbf 1\{D_i=D_{i'}=d\}
\,\Big|\,X^{(n)}\right]
=\Gamma_{d,g}^2.
\]
Therefore,
\[
\E_n\!\big[\hat\rho_n(d,d)\mid X^{(n)}\big]
=\frac{1}{n}\sum_{g=1}^{G_n} N_g\,\Gamma_{d,g}^2
=\frac{1}{n}\sum_{i=1}^n \Gamma_d^2(X_i)
\ \xrightarrow{P}\ \E\!\big[\Gamma_d^2(X)\big],
\]
by the triangular-array law of large numbers in \ref{ass:SP-LB} under the many–small–strata condition.

To show $\hat\rho_n(d,d)-\E_n[\hat\rho_n(d,d)\mid X^{(n)}]\to^P 0$, note that, conditional on $X^{(n)}$, the label-level terms are independent across $g$ (by \ref{ass:TA-LB-PO} for potential outcomes and \ref{ass:TA-LB-assign} for assignments). Let
\[
U_g
:=\frac{2}{T_g(d)\{T_g(d)-1\}}
\sum_{i<i'\in\lambda_g} Y_i Y_{i'} \mathbf 1\{D_i=D_{i'}=d\}.
\]
Standard variance bounds for order-2 U-statistics under sampling without replacement, together with \ref{ass:C_mod}(e), imply
\[
\E_n\!\big[U_g^2\mid X^{(n)}\big]\ \lesssim\ \frac{1}{T_g(d)}\cdot
\frac{1}{N_g}\sum_{i\in\lambda_g}\E\!\big[Y_i^2(d)\mid X_i\big],
\]
which is uniformly bounded in $n$ since $T_g(d)\ge2$ and $\sup_n n^{-1}\sum_{i=1}^n \E[Y_i^2(d)]<\infty$. Hence
\[
\Var_n\!\big(\hat\rho_n(d,d)\mid X^{(n)}\big)
=\Var_n\!\Big(\frac{1}{n}\sum_{g=1}^{G_n} N_g\,U_g\ \Big|\ X^{(n)}\Big)
\ \lesssim\ \frac{1}{n^2}\sum_{g=1}^{G_n} N_g^2
\ \le\ \frac{\max_g N_g}{n}\ \longrightarrow\ 0,
\]
because $\sum_g N_g=n$ and $\max_g N_g/n\to0$. Chebyshev’s inequality then gives
\(
\hat\rho_n(d,d)-\E_n[\hat\rho_n(d,d)\mid X^{(n)}]\to^P 0.
\)
Combining with the convergence of the conditional mean proves the claim.
\end{proof}

\subsubsection{Consistency of the Label-Based Variance Estimator}

Let
\[
\widehat V_n
\;:=\;
\hat\rho_n(1,1)\;+\;\hat\rho_n(0,0)\;-\;2\,\hat\rho_n(1,0),
\]
where for $d\in\{0,1\}$ and $d'\in\{0,1\}$ the quantities
$\hat\rho_n(d,d')$ are the size–weighted within–label products defined
in Lemma~\ref{lem:LB-C2} (for $d\neq d'$) and Lemma~\ref{lem:C3} (for $d=d$).

\noindent\textbf{Claim.} Under Assumptions~\ref{ass:TA-LB-PO}--\ref{ass:TA-LB-assign}, \ref{ass:SP-LB}, and \ref{ass:C_mod}(e), with $G_n\to\infty$ and $\max_g N_g/n\to0$ and with the within–label quotas as stated, the variance estimator $\widehat V_n$ is consistent for
\begin{align*}
&V\;:=\;\E\!\big[(\Gamma_1(X)-\Gamma_0(X))^2\big]
\;=\;\E\!\big[\Gamma_1^2(X)\big]\;+\;\E\!\big[\Gamma_0^2(X)\big]\;-\;2\,\E\!\big[\Gamma_1(X)\Gamma_0(X)\big], \\
&\Gamma_d(x):=\E\!\big[Y^*(d)\mid X=x\big].
\end{align*}

\begin{proof}[Argument]
By Lemma~\ref{lem:C3} (applied with $d=1$ and $d=0$ respectively),
\[
\hat\rho_n(1,1)\ \xrightarrow{P}\ \E\!\big[\Gamma_1^2(X)\big],
\qquad
\hat\rho_n(0,0)\ \xrightarrow{P}\ \E\!\big[\Gamma_0^2(X)\big].
\]
By Lemma~\ref{lem:LB-C2} (with $d\neq d'$),
\[
\hat\rho_n(1,0)\ \xrightarrow{P}\ \E\!\big[\Gamma_1(X)\Gamma_0(X)\big].
\]
Therefore, by the continuous mapping theorem (linearity and continuity of
$(a,b,c)\mapsto a+b-2c$),
\[
\widehat V_n
=\hat\rho_n(1,1)+\hat\rho_n(0,0)-2\hat\rho_n(1,0)
\ \xrightarrow{P}\
\E\!\big[\Gamma_1^2(X)\big]+\E\!\big[\Gamma_0^2(X)\big]-2\,\E\!\big[\Gamma_1(X)\Gamma_0(X)\big]
=V.
\]
Hence $\widehat V_n$ is a consistent estimator of $V$.
\end{proof}

\noindent\textbf{Remark.} In our asymptotic distribution result, $V$ is the population limit of the variance component driving the leading term of the estimator’s influence function (the “between–label” variation in $\Gamma_1(X)-\Gamma_0(X)$). The proofs of Lemmas~\ref{lem:LB-C2} and \ref{lem:C3} verify that each plug–in building block converges to its corresponding population moment under the many–small–strata regime and fixed within–label quotas, which is exactly what is needed for the consistency of $\widehat V_n$.

\subsection{Equivalence of few large labels and many small labels}
\label{app:few-vs-many}

Fix a label $\lambda_g$ of size $N_g$ with target $T_g$ treated units. Let $R_g\ge1$ and let $(n_{g,1},\dots,n_{g,R_g})$ be positive integers that sum to $N_g$. Partition $\lambda_g$ into ordered subblocks $B_{g,1},\dots,B_{g,R_g}$ with $|B_{g,r}|=n_{g,r}$. Draw the ordered partition $\Pi_g=\{B_{g,1},\dots,B_{g,R_g}\}$ uniformly over all ordered partitions with that size profile. Draw the partitions independently across labels. The draws are ancillary and independent of the full collection of potential variables $(Y^*(1),Y^*(0),S(1),S(0),X)$. Conditional on $\Pi_g$, draw a quota vector $t_{g,\cdot}=(t_{g,1},\dots,t_{g,R_g})$ with $\sum_r t_{g,r}=T_g$ from the multivariate hypergeometric law
\[
\Pr\!\big(t_{g,1},\dots,t_{g,R_g}\mid \Pi_g\big)
=\frac{\prod_{r=1}^{R_g} {\,n_{g,r}\,\choose\, t_{g,r}\,}}{{\,N_g\,\choose\, T_g\,}}.
\]
Within each $B_{g,r}$ assign $t_{g,r}$ treated units uniformly without replacement. Do this independently across labels. All statements below condition on realized labels $b^{(n)}$.

\begin{proposition}[Equivalence of assignment laws]\label{prop:equivalence}
For any assignment vector $d_g\in\{0,1\}^{N_g}$ with $\sum_{i\in\lambda_g} d_{g,i}=T_g$,
\[
\Pr(D_g=d_g)=\frac{1}{{\,N_g\,\choose\, T_g\,}}.
\]
Hence the induced assignment on $\{0,1\}^{N_g}$ equals the uniform $T_g$ of $N_g$ law, and the joint assignment across labels matches the baseline blockwise design.
\end{proposition}

\begin{proof}
Fix $\Pi_g$. For a given $d_g$, define $t_{g,r}(d_g,\Pi_g)=\sum_{i\in B_{g,r}} d_{g,i}$. Then
\[
\Pr(D_g=d_g\mid \Pi_g)
=\sum_{t_{g,\cdot}}\Pr(t_{g,\cdot}\mid \Pi_g)\,
\prod_{r=1}^{R_g}\frac{\mathbf 1\{t_{g,r}=t_{g,r}(d_g,\Pi_g)\}}{{\,n_{g,r}\,\choose\, t_{g,r}\,}}
=\frac{1}{{\,N_g\,\choose\, T_g\,}}.
\]
The same value holds unconditionally. Independence across labels holds by construction. This is the uniform $T_g$ of $N_g$ assignment.
\end{proof}

The equivalence of laws implies that for any statistic $\mathcal T=\mathcal T(\{(X_i,D_i,Y_i)\}_{i\le n})$ that is measurable in the data, the finite-sample distribution under the subblock scheme is identical to that under the original uniform $T_g$ of $N_g$ design. Linearizations, influence functions, and design variances coincide. The fast-balancing term equals zero because within each label exactly $T_g$ units are treated and $\sum_{i\in\lambda_g}(D_i-\eta_g)=0$ with $\eta_g=T_g/N_g$.

We now transfer the many small labels asymptotics to the few large labels regime using the random partition as a proof device. Consider a fixed label $g$ with $N_g\to\infty$. Choose size profiles so that $R_g\to\infty$, $\max_r n_{g,r}/N_g\to0$, and $\min_r n_{g,r}\to\infty$ as $N_g\to\infty$. Maintain treatment shares $\eta_g=T_g/N_g$ uniformly bounded away from $0$ and $1$. Consider the array of subblocks $\{B_{g,r}\}$ across all labels when $n=\sum_g N_g$ grows.

\begin{theorem}[Transfer from many small to few large]\label{thm:transfer}
Work under the label-based triangular array with independence across units given $b^{(n)}$ and with within-label uniform assignment with fixed totals. Let $\hat\theta_n$ be a method-of-moments estimator that admits the linearization and central limit theorem proved for the many small labels regime in the main text, with influence function $\psi_n^\ast$ and asymptotic variance $V_\ast$. Suppose $G$ is fixed and $N_g\to\infty$ for each $g$, and choose the random partitions and quotas in each label as above with $R_g\to\infty$, $\max_r n_{g,r}/N_g\to0$, and $\min_r n_{g,r}\to\infty$. Then the array of subblocks satisfies $\sum_g R_g\to\infty$ and $\max_{g,r} n_{g,r}/n\to0$. Conditional on $(\Pi_g,t_{g,\cdot})$ assignments are independent across subblocks, and within each subblock treatment is uniform without replacement. Conditional on $(\Pi_g,t_{g,\cdot})$ and $b^{(n)}$ the subblocks satisfy the same label-wise i.i.d.\ triangular array structure and the same fixed-total assignment as in the many small results. The many small central limit theorem applies to $\hat\theta_n$ with the same influence function $\psi_n^\ast$ and the same variance $V_\ast$. Removing the conditioning leaves the limit unchanged because the unconditional assignment law on the original units is identical to the baseline design and $(\Pi_g,t_{g,\cdot})$ are ancillary and independent of the potential variables. The same central limit theorem holds in the few large labels regime.
\end{theorem}

\begin{proof}
Condition on $(\Pi_g,t_{g,\cdot})$ and on $b^{(n)}$. The subblocks $B_{g,r}$ are disjoint sets of units. Potential variables are independent across units given $b^{(n)}$. Assignments are independent across subblocks by construction, and within each subblock they are uniform without replacement with total $t_{g,r}$. The growth conditions $\sum_g R_g\to\infty$ and $\max_{g,r} n_{g,r}/n\to0$ follow from $R_g\to\infty$ and $\max_r n_{g,r}/N_g\to0$ with $G$ fixed and $n=\sum_g N_g$. Since $\min_r n_{g,r}\to\infty$ and $\eta_g$ is bounded away from $0$ and $1$, both $t_{g,r}$ and $n_{g,r}-t_{g,r}$ are at least $2$ with probability approaching one. The feasible variance estimators that require at least two per arm in each stratum are therefore well defined with probability approaching one. The many small labels central limit theorem from the main text applies under these conditional laws and yields the same influence function and the same variance $V_\ast$ because the linearization and the variance decomposition depend only on within-stratum uniform assignment with fixed totals and on independence across strata. Finally, remove the conditioning. Proposition~\ref{prop:equivalence} implies the unconditional assignment law on the original units coincides with the baseline uniform $T_g$ of $N_g$ law. The partitions and quotas are ancillary and independent of the potential variables, hence averaging over them does not alter the limit.
\end{proof}

The partition is a proof device rather than a change in implementation. In applications keep the actual labels and use the design-consistent plug-in variance based on within-label moments and the known shares. The inference is valid when the design features few large labels that grow within label and it is valid when the design features many small labels that grow in number.

\newpage

\singlespace

\putbib[bib]
%The file containing the bibliography

\end{bibunit}

\pagebreak

\end{document}